\documentclass[12pt]{article}
\usepackage{amsmath}
\usepackage{amsthm}  
\usepackage{graphicx,psfrag,epsf}
\usepackage{bbm}
\usepackage{color}
\usepackage{enumitem} 
\usepackage{natbib}
\usepackage{hyperref} 
\usepackage{amssymb}   
\usepackage{rotating}  
\usepackage{rotating}  
\usepackage[table]{xcolor}
\usepackage{algorithm}  
\usepackage{algpseudocode}  
\usepackage[color]{changebar}
\cbcolor{red}
\usepackage[normalem]{ulem}

\definecolor{myblue}{HTML}{6699CC}
\newcolumntype{d}[1]{D{.}{.}{#1}}  

\graphicspath{{./figures/} }   

\renewcommand{\hat}[1]{\widehat{#1}}
\renewcommand{\tilde}[1]{\widetilde{#1}}

\newcommand{\A}{\mathbf{A}}

\DeclareMathOperator*{\argmin}{argmin} 
\newcommand{\B}{\mathbf{B}}

\newcommand{\bbeta}{\boldsymbol{\beta}}
\newcommand{\bOmega}{\mathbf{\Omega}}
\newcommand{\bPsi}{\mathbf{\Psi}}

\newcommand{\bTheta}{\mathbf{\Theta}}
\newcommand{\btheta}{\boldsymbol{\theta}}

\newcommand{\C}{\mathbf{C}}
\newcommand{\D}{\mathbf{D}}
\newcommand{\E}{\mathbf{E}}
\newcommand{\e}{\mathbf{e}}

\renewcommand{\H}{\mathbf{H}}
\newcommand{\I}{\mathbf{I}}
\newcommand{\M}{\mathbf{M}}
\newcommand{\R}{\mathbb{R}}
\newcommand{\bR}{\mathbf{R}}
\newcommand{\s}[1]{\mathcal{#1}}
\renewcommand{\S}{\mathbf{S}}
\renewcommand{\vec}[1]{{\rm vec}({#1})}
\newcommand{\U}{\mathbf{U}}
\newcommand{\V}{\mathbf{V}}
\newcommand{\W}{\mathbf{W}}

\newcommand{\X}{\mathbf{X}}
\newcommand{\x}{\mathbf{x}}
\newcommand{\Y}{\mathbf{Y}}
\newcommand{\y}{\mathbf{y}}
\newcommand{\Z}{\mathbf{Z}}
\newcommand{\z}{\mathbf{z}}
\newtheorem{theorem}{Theorem}

\newtheorem{proposition}[theorem]{Proposition}
\newtheorem{conditions}{Conditions}

\newcommand\Mark[1]{\textsuperscript{#1}}

\newcommand{\blind}{0}

\begin{document}

\newcommand{\spacingset}[1]{\renewcommand{\baselinestretch}{#1}\small\normalsize}


\if0\blind
{
    \title{\bf Inferring Influence Networks from Longitudinal Bipartite Relational Data\\[.75ex] }
  \author{\normalfont\large 
    Frank W. Marrs\Mark{1}\thanks{
\noindent 
 For correspondence, contact: frank.marrs@colostate.edu. Reproduction code available at \url{https://github.com/fmarrs3/BLIN}.}, Benjamin W. Campbell\Mark{2}, Bailey K. Fosdick\Mark{1},\\ Skyler J. Cranmer\Mark{2}, and Tobias B{\"o}hmelt\Mark{3}
 \\ [.5ex] }
\date{
{\normalfont\normalsize\noindent
\begin{flushleft}
\hspace{.25in}\Mark{1}Colorado State University, Department of Statistics, Fort Collins, CO, USA\\
\hspace{.25in}\Mark{2}The Ohio State University, Department of Political Science, Columbus, OH, USA\\
\hspace{.25in}\Mark{3}University of Essex, Department of Government, Colchester, England, UK
\end{flushleft}
}}
\maketitle
} \fi

\if1\blind
{
  \bigskip
  \bigskip
  \bigskip
  \begin{center}
    {\LARGE\bf  Inferring Influence Networks from\\[5pt] Longitudinal Bipartite Relational Data}
\medskip
\end{center}
} \fi

\medskip
\begin{abstract}
\noindent Longitudinal bipartite relational data characterize the evolution of relations between pairs of actors, where actors are of two distinct types and relations exist only between disparate types.  A common goal is to understand the temporal dependencies, specifically which actor relations incite later actor relations.  There are two primary existing approaches to this problem. The first projects the bipartite data in each time period to a unipartite network and uses existing unipartite network models.  Unfortunately, information is lost in calculating the projection and generative models for networks obtained through this process are scarce.  The second approach represents dependencies using two unipartite \emph{influence networks}, corresponding to the two actor types.  Existing models taking this approach are bilinear in the influence networks, creating challenges in computation and interpretation.  We propose a novel generative model that permits estimation of weighted, directed influence networks and does not suffer from these shortcomings. The proposed model is linear in the influence networks, permitting inference using off-the-shelf software tools.  We prove our estimator is consistent under cases of model misspecification and nearly asymptotically equivalent to the bilinear estimator. We demonstrate the performance of the proposed model in simulation studies and an analysis of weekly international state interactions.
\end{abstract}

\noindent%
\vfill

\spacingset{1.5}

\newpage

\vspace{.5in}

\section{Introduction}  Longitudinal bipartite relational data are being collected at unprecedented rates to study complex phenomena in both the social and biological sciences. These data characterize the evolution of
relations between pairs of actors, where each actor is one of two distinct types, and relations exist only between disparate actor types.
Studies involving such data have focused on, e.g., films \citep{watts1998collective}, international relations \citep{boulet201635, campbell2018latent}, metabolic interactions \citep{jeong2000large}, recommender systems \citep{linden2003amazon}, or transportation systems \citep{zhang2006model}.  
For example, in international affairs, researchers might study countries' financial contributions to  international organizations over the past few decades.  Here, the countries and international organizations are 
the actors and the relations of interest are yearly financial contributions.

In many studies of longitudinal bipartite relational data, the relevant scientific questions surround relations among actors of a single type; the set of which can be represented in a \emph{unipartite} network. For example, researchers may be interested in the (unobserved) relationships among countries that affect the amount of financial contributions to different international organizations.   
The financial contribution of China to the UN, for instance, may be influenced by
the US's recent announcement to cut its budget obligations for the next years. The degree of change in China's contribution based on the US could be viewed as a measure of US influence on China. Such influences may exist between international organizations as well: the contribution of the US to the United Nations (UN) can be related to its financial obligations to the World Trade Organization (WTO). In the examples given, the influences are occurring over time and are allowed to be asymmetric such that, for example, China may be influenced by the US to a high degree while the US is not influenced by China. A goal then in studying longitudinal bipartite data is to infer the relationships among actors of each type. We term these sets of unipartite relations \emph{influence networks}.  For example, in the US-China illustration, the country influence network is denoted $\A = \{a_{ij} \}_{i,j=1}^S$, where $S$ is the number of countries and $a_{ij}$ represents the amount of influence country $i$ has on country $j$.  Similarly, the organization influence network is denoted $\B = \{b_{ij} \}_{i,j=1}^L$, where $L$ is the number of international organizations and $b_{ij}$ represents the influence of organization $i$ on organization $j$.  
Inferring these latent influences are of substantive interest in many studies and, in the case of the states and their contributions to international organizations, have the potential to inform international policy-making in effective and previously unknown ways.

The idea of summarizing bipartite data in terms of unipartite influence networks is not new.  
\cite{newman2001structure} analyzes the relationships among academic authors by estimating the unipartite author-author network from data on academics and the papers they authored over five years. 
\cite{newman2001structure} ignores the temporal component of the data and defines the relationship $a_{ij}$ between author $i$ and author $j$ as the the number of papers $i$ and $j$ co-authored during the five-year period.  The resulting influence network is often referred to as a (one-mode) projection.  If the binary data matrix of authorship is denoted by a rectangular matrix $\Y=\{y_{ik}\}$, where $y_{ik}$ is an indicator of whether academic $i$ authored paper $k$, then the author influence network can be expressed as $\A = \Y\Y^T$.    Notice that $\A$ is symmetric by construction and represents behavioral co-occurrence (in this case, co-authorship), rather than influences over time as described earlier.  
Investigating temporal patterns in publications,  \cite{barabasi2002evolution} estimated yearly influence networks among academic authors using one-mode projections  
and analyzed the evolution of summary statistics of the yearly projections.
Extensions of one-mode projections exist for longitudinal bipartite networks \citep{wu2014temporal}, weighted bipartite networks \citep{newman2004analysis,liu2009personal}, and for creating directed influence networks \citep{zhou2007bipartite}.  These various extensions involve different weightings of the original bipartite relations  (e.g., weight each paper by the number of co-authors). 

Due to the plethora of tools available for unipartite networks, bipartite data are often cast into one-mode projections that can be subsequently analyzed using standard network modeling techniques \citep{zhou2007bipartite}.  Although various weighting schemes have been investigated for one-mode projections  \citep{wu2014temporal}, a key disadvantage of this approach is that information in the original bipartite data is inherently lost in the projection, regardless of the selected weighting scheme.  In addition, since the data naturally arise in a bipartite format, specifying a generative model for the projection on which to base inference is fundamentally challenging. 

There exists some previous work that directly models the observed bipartite network as well. \cite{skvoretz1999logit} and \cite{wang2009exponential}, for example, propose generative Exponential Random Graph Models (ERGMs) for bipartite networks, however this work aims to infer the effect of certain network motifs (such as triangles) on the strength of relations rather than infer the latent influence networks. Another thread of research seeks to explain network formation of heterogeneous information networks, those that consist of disparate node types \emph{and} links, of which bipartite networks are a subset \citep{sun2011co, sun2012mining, shi2017survey}. 
A particularly similar line of work to ours
explains network formation as a function of influence networks among node types, although these models have no temporal component and thus model simultaneous influence rather than the particular sequential influence we consider \citep{liu2010mining, liu2012learning}.

In this paper, we propose a novel bipartite longitudinal influence network (BLIN) model, which permits inference on the influence networks for each set of actors.  This work builds upon recent developments on statistical models for longitudinal unipartite relational data \citep{carley2013models, snijders2005models, krackhardt2007heider, sewell2015latent, carnegie2015approximation}. Specifically, \cite{ almquist2013dynamic, almquist2014logistic} propose an autoregressive model for unipartite networks that may be expressed as a generalized linear model. In a similar vein, we propose an autoregressive, generalized linear model for bipartite networks, wherein the influence networks are autoregressive parameters. 
Although the proposed model is conceptually similar to an existing diffusion model \citep{desmarais2015persistent} and a bilinear regression model \citep{hoff2015multilinear}, our model has key advantages over these existing methods with regard to estimability and interpretability.

The rest of the article is organized as follows. We introduce the BLIN model in Section~\ref{sec_model}. We discuss  approaches to modeling longitudinal bipartite data and then explore various extensions to our model. We describe maximum likelihood estimation procedures for the BLIN model in Section \ref{sec_est} and give properties of the resulting estimators in Section \ref{sec_est_prop}, including performance under misspecification. In Section~\ref{section:simulation}, we compare the performance of our model to existing approaches in simulation studies. In Section~\ref{section_temporal_SID}, we demonstrate our methodology using a data set of material and verbal interactions between international states, where the disparate actor types are the source countries and the target countries of these actions (e.g. humanitarian aid, boycotting, or intent to negotiate). Finally, we discuss future work in Section~\ref{sec:conc}.

\section{BLIN Model}
\label{sec_model}
Let the matrix $\Y_t = \{y_{ij}^t\} \in \mathbb{R}^{S \times L}$ denote the $t^{\text{th}}$ observation of the bipartite relations among $S$ actors of one type (e.g., countries) and $L$ actors of a second type (e.g., international organizations), where the time index  $t \in \{1,2,\ldots,T \}$.  For example, in the illustration introduced above, the $(i,j)$ entry in $\Y_t$, $y_{ij}^t$, is the financial contribution by country $i$ to international organization $j$ in year $t$. 
The BLIN model expresses the relations at time $t$, $\Y_t$, as a function of the $p$ previous relations $\{\Y_{k}: k\in\{t-1, t-2, \ldots, t-p\}\}$, and the influence matrices $\A$ and $\B$:
\begin{align}
\Y_t  &= \A^T \sum_{k=1}^{p_A} \Y_{t - k}  + \sum_{k=1}^{p_B} \Y_{t - k}\B + \E_t, \label{eqbiten}
\end{align}
where $\E_t$ is an $S \times L$ matrix of mean zero, independent and identically distributed  errors and $p= \text{max}(p_A, p_B)$.  The constants $p_A$ and $p_B$ represent the number of previous time periods which influence $\Y_t$ through the networks $\A$ and $\B$, respectively.

The BLIN model can alternatively be expressed as:
 \begin{align}
y_{ij}^t  &=  \sum_{s=1}^S a_{si} \left( \sum_{k=1}^{p_A} y_{s j}^{t-k} \right) +  \sum_{\ell = 1}^L b_{\ell j} \left( \sum_{k=1}^{p_B} y_{i \ell}^{t-k}\right)   + e_{ij}^t.\label{eqbitenSum}
\end{align}
From this representation, it is more easily seen that $y_{ij}^t$ is exclusively a function of those $y_{k \ell}^{t-k}$ when either $k=i$ or $\ell = j$, or both. Consider $p_A=p_B=1$. 
In the context of the international affairs example, this means that China's financial contribution to the UN in 2015 depends on the contributions of all other countries to the UN in 2014 through entries in $\A$, and on China's financial contributions to all other international organizations in 2014 through the entries in $\B$. The interpretation of the individual $\A$ and $\B$ parameters follow from the linearity of the BLIN model. For example, $a_{si}$ is the expected increase in financial contributions of country $i$ to a given international institution when country $s$ has raised its contribution by one unit to the same organization in the previous year. Similarly, the coefficient $b_{\ell j}$ is the expected increase in budget obligations to international organization $j$ by a given country when international organization $\ell$ has received one unit of financial contributions from the same country in the previous year.  
Figure~\ref{fig:dep} depicts these influences using the international affairs example. A positive value of $a_{US,CHI}$ in $\A$, corresponding to the influence of the US on China, means that if the US increased its contributions to the UN in 2014, then China is expected to increase its expenditure to the UN in 2015. Similarly, a positive value of $b_{UN,WTO}$ in $\B$, corresponding to influence of financial obligations to the UN on contributions to the WTO, implies that if the US spent more money on the UN in 2014, then it is expected to increase its WTO expenditure in 2015. Since the influence matrices are time invariant, the entries in  $\A$ and $\B$ represent the average influence over all time periods under consideration.

Figure~\ref{fig:dep} depicts types of direct influence patterns the BLIN model captures (i.e. those between $y_{ij}^t$ and $y_{k \ell}^{t-1}$ where $i=k$ and/or $j=\ell$).  Note, however, that secondary influences (such as that between $y_{ij}^t$ and $y_{k \ell}^{t-s}$ where $i\not=k$, and $j\not=\ell$) may propagate through the BLIN model over multiple time periods, i.e. for $s>1$. For example, although US contributions to the UN in 2014 may not affect China's contribution to the WTO in 2015, it may do so in 2016. This may occur if, say, US financial transfers to the UN in 2014  affect China's UN expenditure in 2015 via a nonzero value of $a_{US, \ CHI}$. Then, China's contribution to the UN in 2015 impact its own contribution to the WTO in 2016 through a nonzero value $b_{UN, \ WTO}$. In this way, the BLIN model allows both direct and secondary influences through different mechanisms.

\spacingset{1}
\begin{figure}[ht]
\centering
\renewcommand{\arraystretch}{1.5}
\begin{tabular}{|c|c|c|}
\hline
  \textbf{Dependence} & $\mathbf{t-1}$ & $\mathbf{t}$ \\
  \hline
  Direct country ($a_{US, CHI}$) & 
  \raisebox{-.5\height}{\includegraphics[width=.15\textwidth]{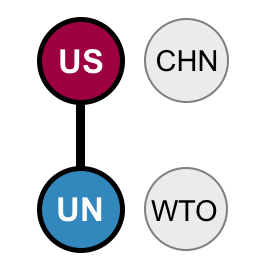}}  &
  \raisebox{-.5\height}{ \includegraphics[width=.15\textwidth]{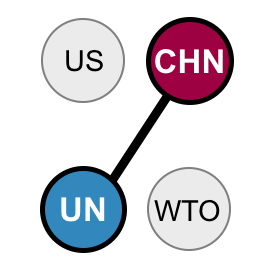}}
  \\
  \hline
  Direct organization ($b_{UN, WTO}$) &
    \raisebox{-.5\height}{\includegraphics[width=.15\textwidth]{a11_issue.png}}  &
  \raisebox{-.5\height}{ \includegraphics[width=.15\textwidth]{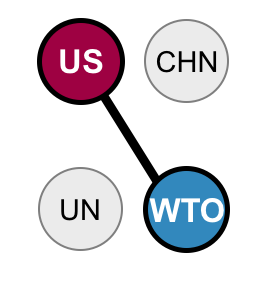}} \\
  \hline
  Secondary 
  &\raisebox{-.5\height}{\includegraphics[width=.15\textwidth]{a11_issue.png}}  &
  \raisebox{-.5\height}{ \includegraphics[width=.15\textwidth]{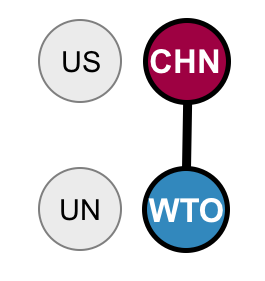}} \\
  \hline
\end{tabular}
\caption{Influence types in longitudinal bipartite relational data in the parlance of the country/international organization example for two countries and two institutions when $lag=1$. Dark red nodes represent countries (US and China) and light blue nodes represent international organizations (UN and WTO).}
\label{fig:dep}
\end{figure}
\spacingset{1.5}

A key flexibility of the BLIN model is that it allows for $p_A \neq p_B$; that is, $\A$ and $\B$ may represent influences over differing time scales.
This is natural. For example, in the international affairs example, country contributions may be influenced only by other countries' contributions in the past year ($p_A=1$). However, the US may have long-standing pattern of contributions to the UN and WTO such that the dependence through $\B$ is much longer, say $p_B = 5$.

A key property of the BLIN model is that it may be written as a linear model. Letting $\y_{t}$ and $\e_t$ denote the column-wise vectorization of matrices $\Y_t$ and $\E_t$, respectively, the model in \eqref{eqbitenSum} can alternatively be expressed 
 \begin{align}
\y_t &= \left(\sum_{k=1}^{p_A} \Y_{t - k}^T \otimes \I_S \right) {\rm vec}(\A^T) + \left(\I_L \otimes \sum_{k=1}^{p_B} \Y_{t - k} \right) {\rm vec}(\B) + \e_t, \label{eqRegBiten} \\
  &:= \mathbb{X}_B^{(t)} \btheta + \e_t, \label{eq_blinvec}
\end{align}
where `$\otimes$' is the Kronecker product and ${\rm vec} (\B)$ denotes the column-wise vectorization of matrix $\B$.  In the second line, $\mathbb{X}_B^{(t)}$ is the $SL \times (S^2 + L^2)$ matrix 
$[\sum_{k=1}^{p_A} \Y_{t - k}^T \otimes \I_S, \, \I_L \otimes \sum_{k=1}^{p_B} \Y_{t - k}]$  
and $\btheta^T = [{\rm vec}(\A^T)^T , \, \vec{\B}^T ]$ is the vector of parameters. Since the BLIN model may be written as a linear model, numerous off-the-shelf tools exist for estimation (including regularization) of the BLIN model, making inference on the influence networks straightforward. In what follows, we focus on the model without covariates, although 
we may easily incorporate covariate information by adding the term $\W_t \bbeta$ to the right hand side of \eqref{eqRegBiten}, for $\bbeta \in \R^p$ and $\W_t \in \R^{SL \times p}$. Thus, we assume throughout the paper that $\Y_t$ is mean zero for all $t \in \{1,2,\ldots,T \}$ without loss of generality.

Another useful representation of the BLIN model is as a vector autoregressive (VAR) model, a generalization of the univariate autoregressive model \citep{sims1980macroeconomics}. Letting $\bTheta_1 = \I_L \otimes \A^T + \B^T \otimes \I_S$ and $\bTheta_2 = \mathbbm{1}_{[p_A > p_B]}\left( \I_L \otimes \A^T \right) + \mathbbm{1}_{[p_B > p_A]} \left( \B^T \otimes \I_S \right)$, the BLIN model in \eqref{eqbiten} may be rewritten 
\begin{align}
\y_t = \bTheta_1 \left(\sum_{k=1}^{q} \y_{t-k} \right) + \bTheta_2 \left(\sum_{k=q+1}^{p} \y_{t-k} \right) + \e_t, \label{eqbiten_VAR}    
\end{align}
where $q = \text{min}(p_A, p_B)$. We note that, when $p_A = p_B = p = 1$, the VAR representation of the BLIN model reduces to $\y_t = \bTheta_1 \y_{t-1} + \e_t$, the standard lag-1 VAR model. An unstructured coefficient matrix $\bTheta_1$ has $S^2 L^2$ unknown parameters, while the BLIN $\bTheta_1$ has only $S^2 + L^2$ unknowns. In this light, the BLIN model may be viewed as reducing the number of unknown parameters in the coefficient matrix $\bTheta_1$ by imposing interpretable bipartite structure on $\bTheta_1$. 

The BLIN model in~\eqref{eqbiten} is not fully identifiable such that for any $c \in \R$, the transformation $\{ \A, \B\} \rightarrow \{\A + c \I_S, \, \B - c \I_L \}$ results in the \emph{exact} same model for the data $\Y_t$.
This non-identifiability means that we are unable to determine $a_{ii}$ and $b_{jj}$ separately, but that the sum $a_{ii} + b_{jj}$ is identifiable. Specifically, we may estimate the effect of $y_{ij}^{t-1}$ on $y_{ij}^t$, but we cannot decompose this effect into the contribution of country $i$ and organization $j$. 
Nevertheless, we may compare the marginal country effect among countries and among international organizations, respectively. For example, although the absolute values of $a_{US,US}$ and $a_{CHI,CHI}$ are not identifiable, their difference $a_{US,US} - a_{CHI, CHI}$ is identifiable. If this difference is positive, we may conclude, for example, that a US increase in financial contributions given a unit expenditure in the previous year is higher than China's increase in contributions. 
 
\subsection{Comparison to Existing Approaches}
\label{sec_bilinear_prop}
Diffusion models \citep[e.g.,][]{berry1990state} are popular for studying the interdependencies of institutions in political science \citep{desmarais2015persistent}. In these models, an outside institution
puts transmission pressure on for a particular policy on the focal institution, making the latter more likely to adopt that policy. The network of these transmissions forms a directed tree, where there is at most one path from one institution to another.
The diffusion model is distinct from the approach we propose in several ways. In the parlance of the international affairs example, the former supposes that each country's financial contribution to a specific international organization is influenced by at most a single other country. In addition,  a  binary network is inferred, rather than  a weighted network which can encode both positive and negative influences.  Furthermore, methods for quantifying uncertainty in the estimated network and incorporating covariates are unavailable.

\cite{hoff2015multilinear} proposes a generative model for bipartite longitudinal data termed the bilinear model, which can be expressed
\begin{align}
\Y_t &= \A^T \sum_{k=1}^{p}\Y_{t-k} \B + \E_t. 
\label{eqBilinearMat}
\end{align}
\cite{hoff2015multilinear} presents an estimator that proceeds by alternating estimates of $\A$ and $\B$, however this estimator is guaranteed to converge only to a local optimum. Thus, global optimality of the existing estimator is not guaranteed.
We illustrate this issue in simulation studies (Section~\ref{section:simulation}). Also, unlike the BLIN model, since the bilinear model is nonlinear, many standard off-the-shelf tools for regularization and uncertainty quantification are not applicable.

The matrices $\A$ and $\B$ in the bilinear model, as in the BLIN model, measure actor influences.
However, 
the bilinear model combines the direct and secondary dependencies in Figure~\ref{fig:dep} into the same mechanism.  This results in a different interpretation of the parameters.  To illustrate the interpretation, we rewrite 
the bilinear model in \eqref{eqBilinearMat} as
\begin{align}
y_{ij}^t  &=  \sum_{s=1}^S \sum_{\ell = 1}^L a_{si}  b_{\ell j} \left( \sum_{k=1}^{p} y_{s \ell}^{t-k} \right) + e_{ij}^t.
\label{eqBilinear}
\end{align}
Here we see $\Y_t$ depends on \emph{every} entry in $\Y_{t-1}$, as $y_{i j}^{t}$ may be affected by both $y_{s j}^{t-1}$ (direct) and $y_{s \ell}^{t-1}$ (secondary) through $a_{si}$.  A consequence of the multiplicative nature of the model is that the influence parameters must be interpreted in conjunction with one another.  For example, $a_{US, \ CHI}$ represents the expected increase in Chinese contributions to international institution $k$ \emph{for each $b_{jk}$ unit} of expenditure of the US to organization $j$ in the previous year.   Thus, the interpretations of the country influences in $\A$ and the international-organization influences in $\B$ are intertwined. 
While there may be instances where $y_{s\ell}^t$ is influenced by $y_{kj}^{t-1}$, we argue that secondary dependencies are often likely of a smaller magnitude than the direct dependencies. In these cases, it would be undesirable to use, for example, a single parameter $a_{ik}$ to simultaneously capture direct and secondary influences.  The BLIN model assumes secondary dependencies are zero and focuses estimation on the direct dependencies.

The influence matrices $\A$ and $\B$ in the bilinear model are identifiable up to a multiplicative constant. For any $c \in \R$, the transformation $\{ \A, \B\} \rightarrow \{c \A , \B/c\}$ leaves the model for $\Y_t$ invariant. This implies that the relative scales  of the networks represented by $\A$ and $\B$ and the signs of the elements are not estimable. However, the ratio of elements within each influence network is identifiable, e.g. $a_{US, CHN} / a_{US, UK}$.

\subsection{Extensions of the BLIN Model}
\label{sec_ext_blin}
In this section, we discuss various extensions of the BLIN model. First, in the definition of the BLIN model in \eqref{eqbiten}, the first observation of the bipartite relations $\Y_t$ is a function of the $p$ past observations $\{\Y_{t-k} \}_{k=1}^{p}$. For simplicity, each past observations affects the current observation equally. Obviously, more complex dependence functions may be considered, such as an exponentially decaying influence of the past time periods on the current time period. In general,
$\sum_{k=1}^{p_A} \Y_{t-k}$ and $\sum_{k=1}^{p_B} \Y_{t-k}$ in \eqref{eqbiten} may be replaced by any $S\times L$ matrix function of the past observations, $f(\Y_{t-})$, where $\Y_{t-} := \{\Y_{t-k} \}_{k=1}^{p}$ and again $p=\text{max}(p_A, p_B)$.  
The selection and estimation of such functions is a current area of research: see \cite{krackhardt2007heider, krivitsky2009statistical, hanneke2010discrete,almquist2014logistic} for a discussion of general unipartite temporal network models and autoregressive models for unipartite temporal networks.

The linear nature of the BLIN model simplifies its extension to other types of outcomes, e.g. binary or count observations. Let $y_{ij}^t$ be a general measure of the  relation between actors $i$ and $j$ at time $t$. Then, a general BLIN model may be expressed
\begin{align}
g(E[\Y_t | \Y_{t-}])  &= \A^T \sum_{k=1}^{p_A} \Y_{t - k}  + \sum_{k=1}^{p_B} \Y_{t - k}\B, \label{eqbiten_glm}
\end{align}
where $g(.)$ is an appropriate link function based on the form of $\Y_t$ \citep{mccullagh1989generalized}.  
Off-the-shelf tools are again available for estimation of the model in \eqref{eqbiten_glm} if $g$ is a standard link function. 
When $g$ is the canonical link function for logistic regression, for example, the BLIN model in \eqref{eqbiten_glm} may be viewed as a conditional logistic discrete choice model \citep{mcfadden1973conditional}; models of this type have recently been employed in network representations \citep{overgoor2018choosing}.

\section{Estimation of the BLIN model}
\label{sec_est}
In the following, we discuss several estimation procedures for the BLIN model. First, we propose an estimator that results from minimizing a least squares criterion. 
We then consider more parsimonious estimators using sparsity-inducing penalties and reduced-rank approaches.
For ease of notation, we define the regressor matrices in \eqref{eqbiten}
as $\X_t := \sum_{k=1}^{p_A} \Y_{t - k}$ and  
$\Z_t := \sum_{k=1}^{p_B} \Y_{t - k}$, such that the BLIN model can be expressed
\begin{align}
\Y_t  &= \A^T \X_t + \Z_t\B + \E_t. \label{blin_Xt}
\end{align}
In the theory that follows in this section and the next, we treat $\X_t$ and $\Z_t$ generally as, in principle, they may be any sequence of matrices of appropriate size.

\subsection{Least Squares Estimator}
\label{sec_estimation}
Based  on the vector representation of the BLIN model in \eqref{eqRegBiten}, we propose minimizing the following least squares criterion to construct an estimator for the  $\A$ and $\B$ matrices: 
\begin{align}
\hat{\btheta} &=\argmin_{\btheta} (\y - \mathbb{X}_B \btheta)^T (\y - \mathbb{X}_B \btheta), \label{eqopt}\\
& = \argmin_{ \{\A,\B\} } \sum_t ||\Y_t - \A^T \X_t - \Z_t \B ||^2_F
\label{eq_LS_matrix}
\end{align}
where $\y^T := [\y_1^T, \y_2^T, \ldots, \y_T^T]$ 
such that $\y \in \R^{SLT}$ and $\mathbb{X}_B \in \R^{SLT \times (S^2 + L^2)}$ is the column-wise stacking of the design matrices $\{ \mathbb{X}_B^{(t)} \}_{t=1}^T$ in the vector representation of the BLIN model in \eqref{eq_blinvec}. 
The explicit solution to \eqref{eqopt} is
{ \small
\begin{align}
\left[ \begin{array}{c} 
\vec{\hat{\A}^T}  \\
\vec{\hat{\B}} 
\end{array} \right]
=
\left[ \begin{array}{cc} 
\left(\sum_{t=1}^T \X_t \X_t^T \right) \otimes \I_S & 
\sum_{t=1}^T \X_t \otimes \Z_t \\
\sum_{t=1}^T \X_t^T \otimes \Z_t^T & 
\I_L \otimes \left(\sum_{t=1}^T  \Z_t^T \Z_t \right)
\end{array} \right]^{-}
\left[ \begin{array}{c} 
\vec{\sum_{t=1}^T \Y_t \X_t^T }  \\
\vec{\sum_{t=1}^T \Z_t^T \Y_t  }
\end{array} \right],\label{eqABexact}
\end{align}}
where $\mathbf{H}^{-}$ denotes the generalized inverse of square matrix $\mathbf{H}$.

Computing the solution in \eqref{eqABexact} requires inversion of a matrix of dimension $S^2 + L^2$. Using an iterative algorithm to solve \eqref{eqopt}, this computation can be replaced by repeated inversions of square matrices of dimensions $S$ and $L$.  Specifically, Algorithm \ref{alg:BLIN_it} details a block coordinate descent procedure, which alternates between solving for $\A$ and $\B$. 
Particularly, the iteration scheme reduces memory demand and reduces the complexity of computation from $O((S^2 + L^2)^3)$ in \eqref{eqABexact} to $O(T\cdot \text{max}(S,L)^3)$. In the data analysis in Section~\ref{section_temporal_SID} (with $S=L=25$ and $T = 543$), we find the iteration scheme advantageous over building the design matrix $\mathbb{X}_B$.
The estimator in Algorithm~\ref{alg:BLIN_it} converges to a  unique minimum when $\left( \sum_t \Z_t^T \Z_t \right)$ and $\left( \sum_t \X_t \X_t^T \right)$ are full rank. 
\spacingset{1}
\begin{algorithm}
\caption{Block coordinate descent  estimation of BLIN model}
\label{alg:BLIN_it}
\begin{enumerate}
  \setcounter{enumi}{-1}
\item Set threshold for convergence $\eta$. Set number of iterations $\nu = 1$. Initialize $\hat{\A}^{(0)} = \I_S $, $\hat{\B}^{(0)} = \I_L$ and $Q_0 = \sum_t ||\Y_t||^2_F$.
\item Compute $\hat{\B}^{(\nu)} = \left( \sum_t \Z_t^T \Z_t \right)^{-1} \left( \sum_t \Z_t^T \tilde{\Y}_t^{(A)} \right)$, where $\tilde{\Y}_t^{(A)} := \Y_t - (\hat{\A}^{(\nu - 1)} )^T \X_t$ for all $t$.
\item Compute $( \hat{\A}^{(\nu)} )^T = \left( \sum_t (\tilde{\Y}_t^{(B)})^T \X_t \right) \left( \sum_t \X_t \X_t^T\right)^{-1}$, where $\tilde{\Y}_t^{(B)} := \Y_t -  \Z_t \hat{\B}^{(\nu)}$ for all $t$.
\item Compute the least squares criterion $Q_{\nu} = \sum_t ||\Y_t - (\hat{\A}^{(\nu)})^T \X_t - \Z_t \hat{\B}^{(\nu)} ||^2_F$.  If  $|Q_{\nu} - Q_{\nu - 1}| > \eta$, increment $\nu$ and return to 1.  
\end{enumerate}
\end{algorithm}
\spacingset{1.5}

\subsection{Sparse Coefficients}
\label{sec_sparse}
Although influence may be multifactorial, it is easy to imagine scenarios where many entries in $\A$ and $\B$ are small or zero. In the international affairs example, democracies may only influence other democracies, or organizations dealing with issues at a global level may only be influenced by other global institutions.
To leverage this fact in parameter estimation, we propose augmenting the least-squares criterion of the
of the BLIN model in~\eqref{eqopt} with a sparsity-inducing penalty. 
Here, we consider the Lasso penalty \citep{tibshirani1996regression}, which uses an $L^1$ norm on $\btheta$ to simultaneously perform variable selection and regularization.  We term this model the \emph{sparse BLIN} model as elements of $\hat{\btheta}$ are forced to the zero.  The estimation objective function is
\begin{align}
\hat{\btheta} =\argmin_{\btheta} (\y - \mathbb{X}_B \btheta)^T  (\y - \mathbb{X}_B \btheta) + \lambda ||\btheta||_1, \label{eqLasso}
\end{align}
where $\lambda$ is a tuning parameter, 
and larger values of $\lambda$ correspond to more regularization. Since the BLIN model is linear, the vector of parameters $\btheta$ may be  regularized with any of a host of existing penalty terms \citep{hoerl1970ridge,friedman2001elements}.
We note that, as presented in \eqref{eqLasso}, the sparse BLIN estimator in  loses some of the scalability of the full BLIN estimator presented in Algorithm~\ref{alg:BLIN_it}. However, it may be possible to implement an estimator of the sparse BLIN model in the spirit of Algorithm~\ref{alg:BLIN_it}, that is, with alternating updates of $\A$ and $\B$. As we find no issue estimating \eqref{eqLasso} in the data analysis in Section~\ref{section_temporal_SID} with $S=L=25$ and $T=543$, we leave this implementation for future work.

\subsection{Reduced-Rank Coefficients}
\label{sec_reduced}
Thus far we discussed sparsity-inducing penalties for the vector of regression model coefficients $\btheta$.
However, several other 
penalties on the singular value decomposition (SVD) of coefficient matrices $\A$ and $\B$ have been proposed.  These penalties result in coefficient estimates of $\A$ and $\B$ with reduced-rank, or approximately reduced-rank. 
\cite{yuan2007dimension} propose a nuclear norm penalty, which is an $L^1$ penalty of the singular values of $\A$.
Similarly, \cite{bunea2011optimal} recommend a rank selection criteria penalty that is proportional to the rank of $\A$, i.e. a $L^0$ penalty on the singular values $\A$. This second approach provides simultaneous shrinkage on $\A$ and consistent estimation of its (reduced) rank. 

Here we consider estimation of reduced-rank $\A$ and $\B$, i.e., ${\rm rank}(A) = k < S$ and ${\rm rank}(B) = m < L$. This assumption may be appropriate when there is lower-dimensional structure inherent in $\A$ and $\B$: for example, the influences among countries may be grouped by region. This approach is employed in reduced-rank regression,  first developed by \cite{anderson1951estimating}, and has connections to principal component analysis, as shown in~\cite{izenman1975reduced}. For any ranks $k < S$ and $m < L$ of $\A$ and $\B$, respectively, we may define a \emph{reduced-rank BLIN model} by writing $\A^T = \U \V^T$ for 
$\U,\V \in \R^{S \times k}$ and $\B = \bR \S^T$ for
$\bR, \S \in \R^{L \times m}$. These decompositions are not identifiable up to a full-rank transformation of the decompositions, e.g. $\{\U,\V \}$ and $ \{ \U \mathbf{G},  \mathbf{G}^{-1} \V\}$ result in the same influence matrix $\A$ for any invertible matrix $\mathbf{G}$. However, the estimands $\A$ and $\B$ remain identifiable up to an additive constant along the diagonal, as discussed in Section~\ref{sec_model}. 

We estimate a reduced-rank BLIN model by minimizing the least-squares criteria in~\eqref{eq_LS_matrix}  with $\A$ replaced by $\U\V^T$, $\B$ replaced by $\bR \S^T$, and minimizing over $\{\U,\V,\bR,\S\}$. This optimization problem is easily solved using a block coordinate descent algorithm 
similar to Algorithm~\ref{alg:BLIN_it} for the full BLIN model and with similar computational complexity
(see Section~\ref{sec:MLE_reduced_rank} of the Supplementary Material for details).
In what follows, we refer to the  BLIN model with no constraints on the parameters $\A$ and $\B$ as the \emph{full BLIN} model, in order to distinguish it from the sparse and reduced-rank versions.

\section{Estimator Properties}
\label{sec_est_prop}
This section examines properties of the least squares estimators of the full and reduced-rank BLIN models (all proofs are provided in Section~\ref{AppProofs} of the Supplementary Material).  We show that the least squares estimators are unique for a relatively small number of observations $T$ and give some sufficient conditions for their asymptotic normality and efficiency. We then examine the properties of the least squares estimators under misspecification, providing sufficient conditions for their consistency.
As in the previous section, we use the representation of the BLIN model in \eqref{blin_Xt}.

\subsection{Uniqueness and Efficiency of Least Squares Estimators}
\label{sec_unique}
When $\X_t$ and $\Z_t$ in \eqref{blin_Xt} represent past observations of $\Y_t$, the BLIN model is a VAR model as discussed following \eqref{eqbiten_VAR}. For the estimators in Section~\ref{sec_ext_blin} to be consistent, stationarity of the time series is required \citep{brockwell1991time}. When $p = 1$, a sufficient condition for stationarity is that the eigenvalues of $\bTheta_1$ in \eqref{eqbiten_VAR} all have modulus less than one. To define a sufficient condition for stationarity in the general case of $p>1$ and $p_A \neq p_B$, we rewrite the VAR version of the BLIN model in \eqref{eqbiten_VAR} in its {companion form} \citep{zivot2006vector} by augmenting the vector of observations at $t$ with the $p$ previous observations, $\boldsymbol{\xi}_t^T = [\y_t^T, \y_{t-1}^T, \ldots, \y_{y-p}^T]$,  and augmenting the error vector at time $t$ with an appropriate number of zeros, $\mathbf{v}_t^T = [\e_t^T, \mathbf{0}^T, \ldots, \mathbf{0}^T]$. Then, the BLIN model can be expressed
\begin{align}
&\boldsymbol{\xi}_t = \mathbf{F} \boldsymbol{\xi}_{t-1} + \mathbf{v}_t, \hspace{.4in}
\mathbf{F} := \begin{bmatrix}
    \mathbf{1}^T_{q} \otimes \bTheta_1 \,\, ; \,\, \mathbf{1}^T_{p - q} \otimes \bTheta_2 \\
   \I_{SL(p-1)} \,\, ; \,\, \mathbf{0} \\
    \end{bmatrix}, \label{eq_var_companion}
\end{align}
where $q = {\min}(p_A, p_B)$ and $\mathbf{1}_n$ is the vector of $n$ ones. 
Then, a sufficient condition for stationarity of the BLIN model is that moduli of the eigenvalues of the {companion matrix} $\mathbf{F}$ are all less than one \citep{zivot2006vector}.

The optimization problem described in~\eqref{eqopt} is convex, and under stationarity, by the Gauss-Markov theorem, it has a unique solution whenever $\mathbb{X}_B$ is full rank \citep[see, e.g.,][]{graybill1976theory}.
Due to the non-identifiability of the diagonal entries of $\A$ and $\B$, $\mathbb{X}_B$ is never full rank. However, the projection of $\y$ onto the column space of $\mathbb{X}_B$,  i.e., $\hat{\y} = \mathbb{X}_B \hat{\btheta}$, is always unique. Thus, when the column space of $\mathbb{X}_B$ spans the space of possible $\A$ and $\B$ matrices up to their non-identifiability, the BLIN estimator in \eqref{eqABexact} is unique (up to the non-identifiability properties). This is true 
when the rank of $\mathbb{X}_B$ is maximized, that is one less than the number of columns: ${\rm rank}(\mathbb{X}_B) = S^2 + L^2 - 1$. The `$-1$' results from the fact that if a single diagonal entry in  $\{ a_{ii} \}_{i=1}^S$ or $\{ b_{jj} \}_{j=1}^L$ is known, then the rest of the diagonal entries are identifiable.  We now provide a proposition that states conditions under which $\mathbb{X}_B$ is maximal rank; these conditions are satisfied with probability one when, e.g., $\{ \X_t \}_{t=1}^T$ and $\{ \Z_t \}_{t=1}^T$ are distributed array normal. 
\begin{proposition}
\label{prop_rankXb}
Without loss of generality, take $S \le L$. Assume that the $TS \times L$ matrices formed by the column-wise concatenation $[\X_1 ; \X_2 ; \ldots ;\X_t]$ and $[\Z_1 ; \Z_2 ; \ldots ;\Z_t]$ are full rank. Then, the design matrix has ${\rm rank}(\mathbb{X}_B) = {\rm min}(TSL, S^2 + L^2 - 1). $
\end{proposition}
A consequence of Prop.~\ref{prop_rankXb} is that the full BLIN model has a unique solution when $TSL \ge S^2 + L^2 - 1$.  A key implication of this is that a unique solution exists for relatively small $T$. For instance, if $S=L$, then the BLIN model has a unique solution (modulo the non-identifiability) when $T=2$. 
Figure~\ref{fig:uniquesol} plots values of $S$, $L$, and $T$ 
for which the full BLIN model has a unique solution. As $T$ grows, the space of values for which the BLIN model has a unique solution rapidly spans all values of $S$ and $L$, except when their values are extremely disparate.

\spacingset{1}
\begin{figure}[ht]
\centering
\begin{tabular}{c c c c} 
\includegraphics[width=.2\textwidth]{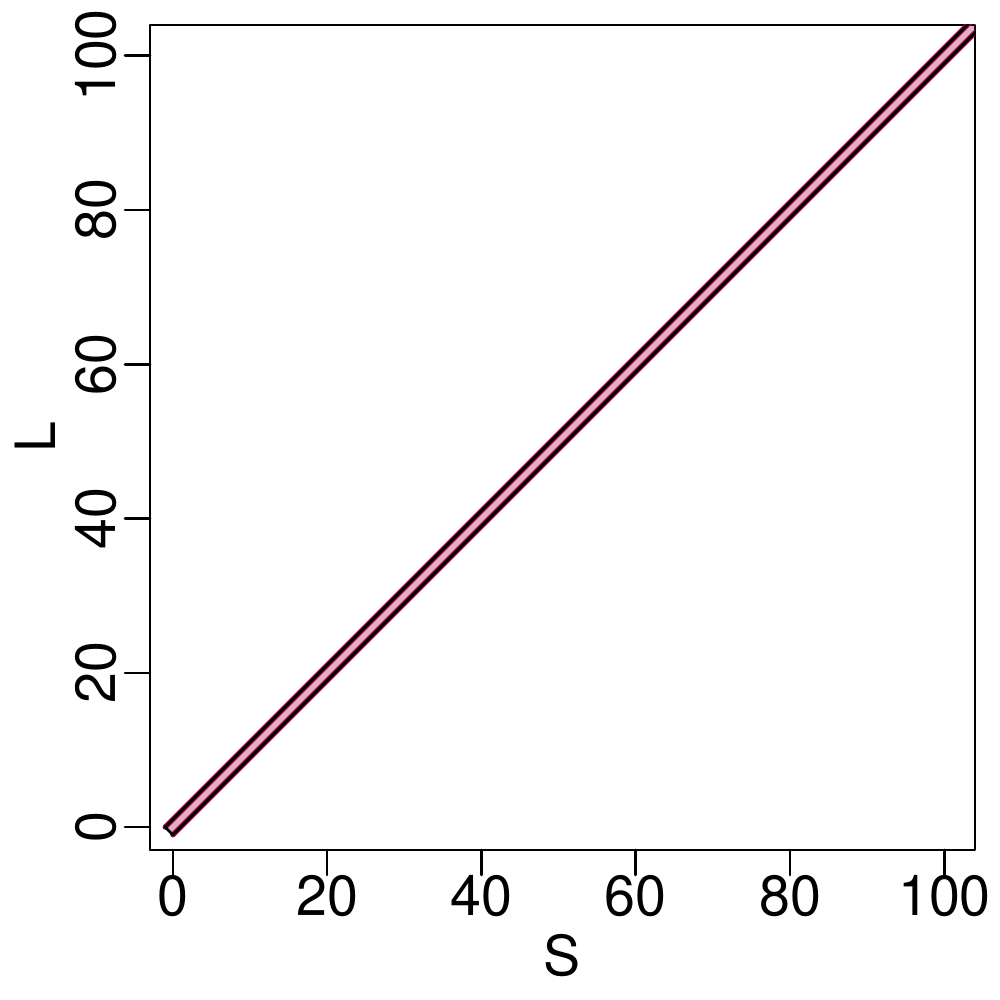} &
\includegraphics[width=.2\textwidth]{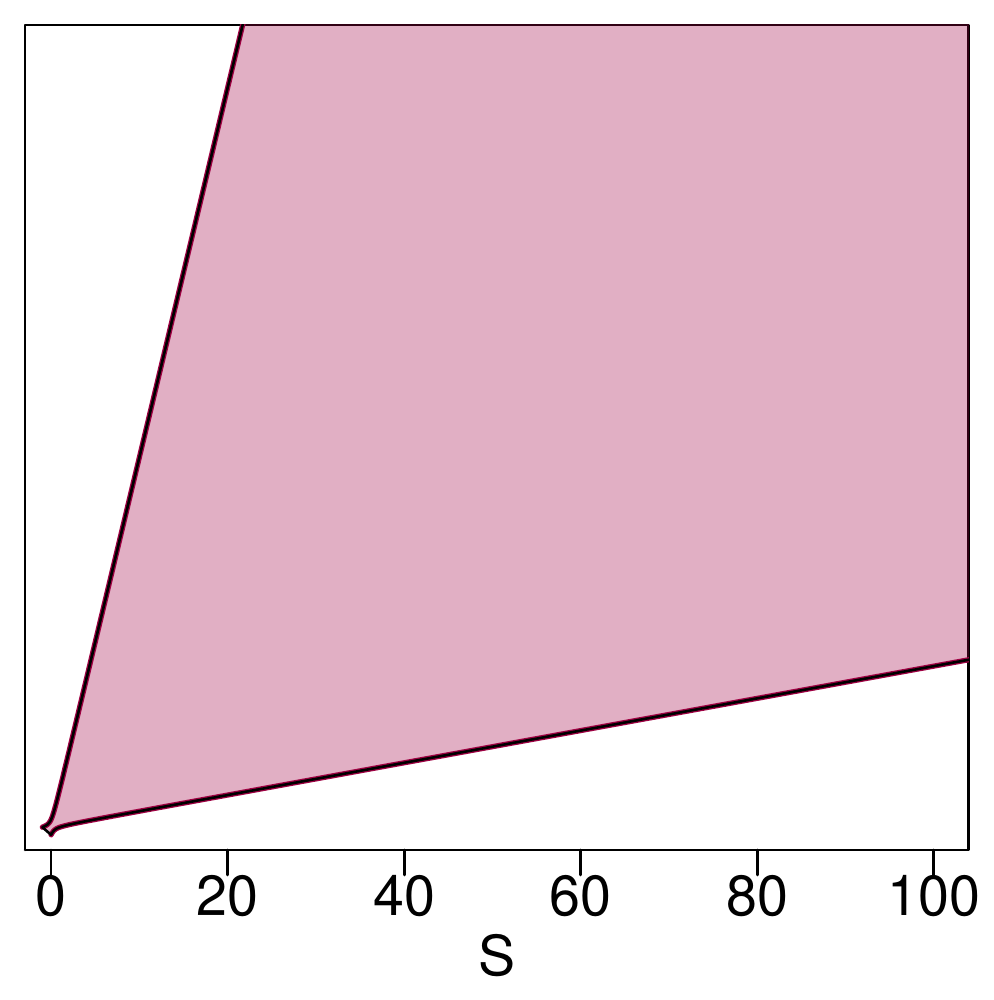} &
\includegraphics[width=.2\textwidth]{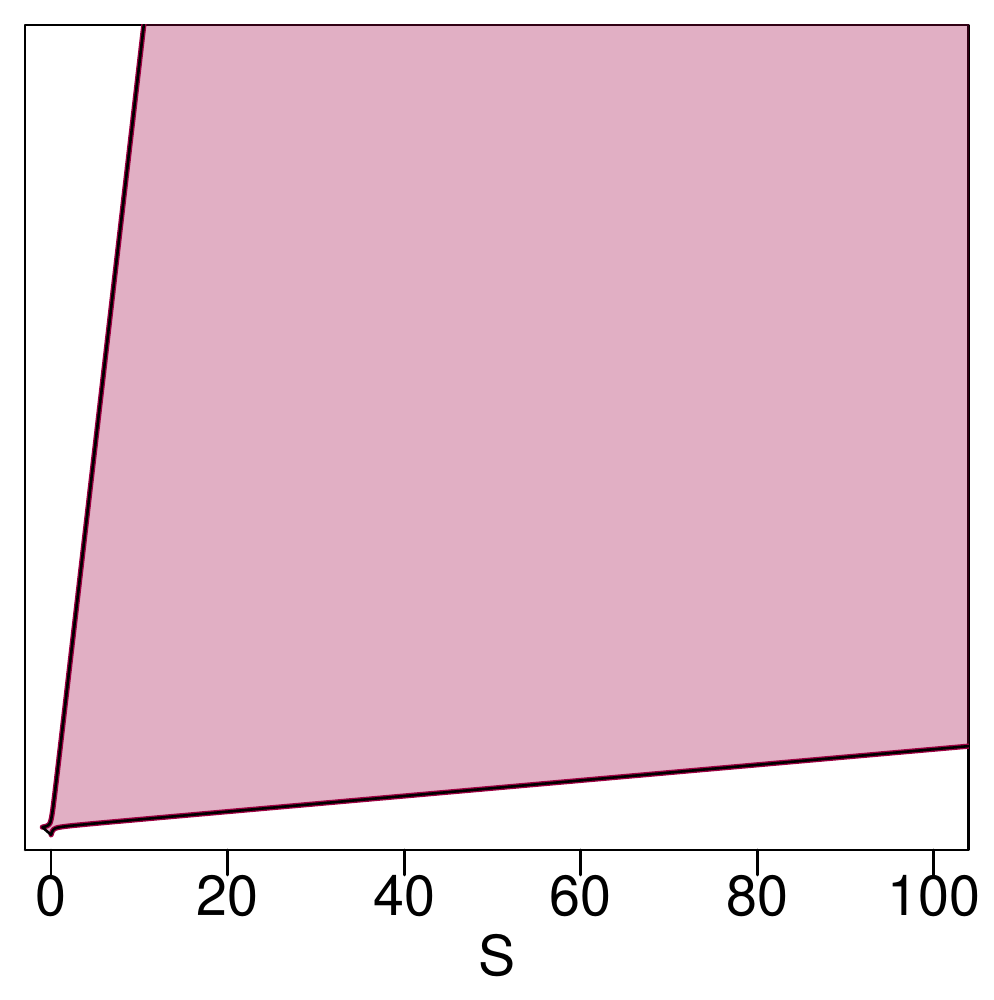} &
\includegraphics[width=.2\textwidth]{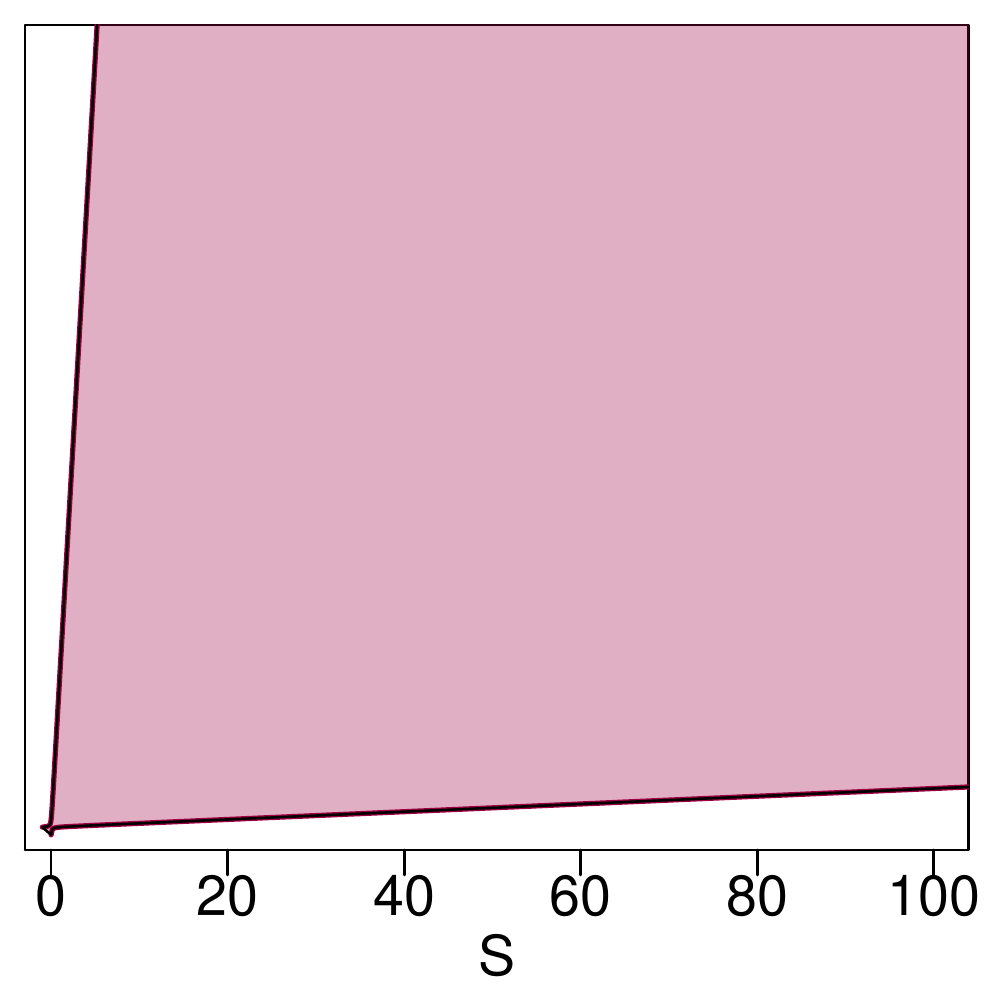}  \\
$T = 2$ &
$T = 5$ &
$T = 10$ &
$T = 20$ 
\end{tabular}
\caption{Dimensions of square matrices $\A$ and $\B$, $S$ and $L$ respectively, when the BLIN model has a unique solution. Shaded areas denote unique solutions. }
\label{fig:uniquesol}
\end{figure}
\spacingset{1.5}

Under the BLIN model in \eqref{eqbiten} when $\X_t$ is independent $\E_t$
and the true coefficients have companion matrix $\mathbf{F}$ in \eqref{eq_var_companion} with eigenvalues in the unit circle, 
the Gauss-Markov theorem \citep{graybill1976theory} states that the BLIN estimator $u^T\hat{\btheta}$, where $\hat{\btheta}$ is any solution to \eqref{eqopt} and  $u^T\btheta$ defines any identifiable linear combination of $\btheta$, is the best linear unbiased estimator (BLUE) of $u^T\btheta$. 
Additionally, these estimates are asymptotically normal at rate $\sqrt{T}$.
Finally, we note that the estimator in \eqref{eqopt} is the maximum likelihood estimator when $\y$ is normally distributed with homogeneous variance.
Thus, under regularity conditions \citep{lehmann2006theory}, the limiting normal random variable for $u^T \hat{\btheta}$, has minimum asymptotic variance.

Recall that the least squares estimates of the reduced-rank BLIN model
can be obtained using an iterative block-coordinate descent algorithm (see Section~\ref{sec:MLE_reduced_rank} of the Supplementary Material).  When every matrix inverse in the update equations is unique, the estimates $\hat{\A}$ and $\hat{\B}$ from this algorithm converge to a local minimum, which is unique up to non-identifiabilities in $\A$ and $\B$ provided that ${\rm rank}(\mathbb{X}_B) = S^2 + L^2 - 1$.  
The reduced-rank estimators are the maximum likelihood estimators for $\A$ and $\B$ with ranks $k$ and $m$, respectively, under the assumptions of normally distributed independent errors $\E_t$ with homogeneous variance.  Therefore, under these conditions, $\hat{\A}$ and $\hat{\B}$ resulting from the iterative block-coordinate descent algorithm are consistent and asymptotically normal with minimum asymptotic variance.

\subsection{Least Squares Estimator Properties under Misspecification}
\label{section_misspec} 
Thus far, we have discussed the attractive properties of the least squares estimators of the BLIN model, $\hat{\A}$ and $\hat{\B}$, under the assumption that the model is correctly specified. It is useful to determine the limiting values of the estimators when the model is misspecified.
The limiting values of $\hat{\A}$ and $\hat{\B}$, which we denote $\tilde{\A}$ and $\tilde{\B}$, respectively, are referred to as pseudo-true parameters in the model misspecification literature, first investigated by~\cite{huber1967behavior}.  
The pseudo-true parameters, by definition, are 
\begin{align}
\{ \tilde{\A}, \tilde{\B} \} = \argmin_{ \{\A, \B \} } \ E \left[ \sum_{t= 1}^T  || \mu(\X_t) - \A^T \X_t - \Z_t \B ||_F^2 \right], \label{eqRR2norm} 
\end{align}
where $\Y_t$ has expectation $E[\Y_t | \X_t, \Z_t] = \mu(\X_t, \Z_t)$, some general function of $\X_t$ and $\Z_t$, and the expectation in \eqref{eqRR2norm} 
is over $\X_t$ and $\Z_t$. Note that this expression holds regardless of the distribution of $\Y_t$, $\X_t$, and $\Z_t$  and the form of $\mu(\X_t, \Z_t)$.  We note that the limiting values $\tilde{\A}$ and $\tilde{\B}$ minimize the Kullback-Leibler divergence from the true distribution of $\Y_t$ to the distribution of $\Y_t$ under the BLIN model with Gaussian errors. 
In this section, we assume that the variance-covariance matrices of $\x_t := {\rm vec}(\X_t)$ and $\z_t := {\rm vec}(\Z_t)$ are Kronecker-structured, that is $E[\x_t  \x_t^T] = \bOmega_X \otimes \bPsi_X$ and $E[\z_t  \z_t^T] = \bOmega_Z \otimes \bPsi_Z$ for $\X_t$ and $\Z_t$ mean zero. 
This is the case if, for example, $\X_t$ and $\Z_t$ are distributed matrix normal \citep{gupta2000}. This assumption is consistent with the theoretical treatment of the bilinear model in \cite{hoff2015multilinear}. We also assume that the errors are additive. Below, we formalize the conditions for the theory to follow.
\begin{conditions} \mbox{}
\label{cond:cond}
\begin{enumerate}
\item Each $\{ \X_t \}_{t=1}^T$ is identically distributed with mean zero and $E[\x_t \x_t^T] = \bOmega_X \otimes \bPsi_X$, and each $\{ \Z_t \}_{t=1}^T$ is identically distributed with mean zero and $E[\z_t \z_t^T] = \bOmega_Z \otimes \bPsi_Z$, for positive-definite matrices $\bOmega_X, \bOmega_Z \in \R^{L \times L}$ and $\bPsi_X, \bPsi_Z \in \R^{S \times S}$ with finite entries. 
\item For all $t \in \{1,2,\ldots, T \}$, $\Y_t = \mu(\X_t, \Z_t) + \E_t$, where $\mu(\X_t, \Z_t)$ is a general function of $\X_t$ and $\Z_t$, and every entry in $\E_t$ is independent and identically distributed with homogeneous, finite variance.
\item The sequence of $\{ \Y_t \}_{t=1}^T$ is weakly stationary. 
\end{enumerate}
\end{conditions}

In the remainder of this section, we examine some properties of the pseudo-true parameters (all proofs are provided in Section~\ref{AppProofs} of the Supplementary Material). 
The $(i,j)$ entry in $\tilde{\A}$, denoted $\tilde{a}_{ij}$, estimates the linear relationship between row $i$ of $\X_t$ and row $j$ of $\Y_t$ across all time. Similarly, the $(i,j)$ entry in $\tilde{\B}$, denoted $\tilde{b}_{ij}$, estimates the linear relationship between column $i$ of 
$\Z_t$ and column $j$ of $\Y_t$ across all time. The first proposition states that when there is no linear relationship, the appropriate pseudo-true parameter is zero. We denote row $i$ of $\X_t$ as $\x_{i \cdot t}$ and column $j$ of $\X_t$ as $\x_{\cdot j t}$, and do the same for $\Y_t$ and $\Z_t$. 

\begin{proposition}
Under Conditions~\ref{cond:cond}, if $\bOmega_X$ is diagonal and $E[\y_{j \cdot t}^T \x_{i \cdot t}] =0$ for all $t$, then
$\tilde{a}_{ij}$ = 0. Alternatively, if $\bPsi_Z$ is diagonal and $E[\z_{\cdot i t}^T \y_{\cdot j t} ] = 0$ for all $t$, then 
$\tilde{b}_{ij}$ = 0. \label{propPseudo1}
\end{proposition}

In the setting of the international policy example, Proposition~\ref{propPseudo1} states that if country $i$'s expenditure is uncorrelated with country $j$'s contribution to the same organization in the following year, then the least squares estimator of $a_{ij}$ in the BLIN model will converge to $\tilde{a}_{ij} = 0$.
  
The following proposition provides alternative conditions under which the pseudo-true parameters are equal to zero.  It allows for more general covariance structure at the cost of assuming that the conditional mean of $\y_t$ is linear in $\x_t$. 
\begin{proposition}
Assume that there exists a linear relationship $E \big[ \y_t | \x_t, \z_t \big] = \bTheta_X \x_t + \bTheta_Z \z_t$ for all $t$ and Conditions~\ref{cond:cond} hold.  If all entries in $\bTheta_X$ relating row $i$ in $\X_t$ to row $j$ of $\Y_t$ are zero and $i \neq j$, then
$\tilde{a}_{ij}$ = 0. If all entries in $\bTheta_Z$ relating column $i$ in $\Z_t$ to column $j$ of $\Y_t$ are zero and $i \neq j$, then 
$\tilde{b}_{ij}$ = 0. \label{propPseudo2}
\end{proposition}

\subsection{Comparison of BLIN and Bilinear Least Squares Estimators}
\label{sec_comp_LS}
We now compare least squares estimates of the BLIN model to those of the bilinear model. 
As the bilinear model accommodates only a single lag for $\A$ and $\B$, in all that follows we fix $p_A = p_B = p$ so that $\Z_t = \X_t$, $\bOmega_X = \bOmega_Z =\bOmega$,  and $\bPsi_X = \bPsi_Z =\bPsi$. 
Both the BLIN and bilinear 
models aim to quantify the influences of the rows (columns) of $\X_t$ on the rows (columns) of $\Y_t$, but with different emphasis on the type of influence quantified (recall the discussion in the Introduction and Section~\ref{sec_bilinear_prop}). We provide a theorem and proposition stating that, when data are generated from the bilinear model, the least squares BLIN estimators of the off-diagonal entries in $\A$ and $\B$ converge to the corresponding $\A$ and $\B$ values used in the bilinear generative model, up to numerical constants. The analogous result holds when switching the roles of the bilinear and BLIN models. 
See Section~\ref{AppProofs} of the Supplementary Material for proofs.

\begin{theorem}
\label{thm:BLIN_bilinear_offdiag}
Under Conditions~\ref{cond:cond},
\begin{enumerate}
\item
If $\{ \Y_t \}_{t=1}^T$ are generated from the bilinear model in \eqref{eqBilinearMat}, then for all $i \not= j$ and $k \not= \ell$
\[\tilde{a}_{ij}  =\frac { {\rm tr}(\bOmega \B)} {{\rm tr}(\bOmega) } a_{ij},  \text{ \ \  and \ \ } \tilde{b}_{k \ell}  = \frac { {\rm tr}(\bPsi \A)} {{\rm tr}(\bPsi) } b_{k \ell}. \]

\item
If $\{ \Y_t \}_{t=1}^T$ are generated from the BLIN model in \eqref{eqbiten}, then for all $i \not= j$ and $k \not= \ell$
\[\bar{a}_{ij}  = \frac {{\rm tr}(\bOmega \bar{\B} )} {{\rm tr}(\bOmega \bar{\B} \bar{\B}^T)}a_{ij}  \text{ \ \  and \ \ } \bar{b}_{k \ell}  = \frac {{\rm tr}(\bPsi \bar{\A} )} {{\rm tr}(\bPsi \bar{\A} \bar{\A}^T)} b_{k \ell},  \]
where $\bar{\A} =\{\bar{a}_{ij}\}$ and $\bar{\B}=\{\bar{b}_{k \ell}\}$ are the pseudo-true parameters of least-squares estimation of the bilinear model.  
\end{enumerate}
\end{theorem}

We now address the diagonals of the $\A$ and $\B$ matrices. We provide conditions under which the diagonals (up to their non-identifiabilities) are asymptotically equivalent. To be clear, we compare the estimated influence of $x_{ij}^{t}$ on $y_{ij}^{t}$ from the two models. Under the BLIN model, this influence is $a_{ii} + b_{jj}$, whereas under the bilinear model the influence is $a_{ii}b_{jj}$. 
\begin{proposition}
\label{prop:BLIN_bilinear_diag}
Suppose the true $\A$ and $\B$ matrices are constant along the diagonal with nonzero values  $\alpha$ and $\beta$, respectively, and $\alpha + \beta \neq 0$. Then, under Conditions~\ref{cond:cond},
\begin{enumerate}
\item
If  $\{ \Y_t \}_{t=1}^T$ are generated from the bilinear model in \eqref{eqBilinearMat}, then the pseudo-true parameter $\tilde{a}_{ii} + \tilde{b}_{jj}$  of least-squares estimation of the BLIN model is equal to the true diagonal specification $\alpha \beta$, and
\item
If  $\{ \Y_t \}_{t=1}^T$ are generated from the BLIN model in \eqref{eqbiten}, then the pseudo-true parameter $\bar{a}_{ii}\bar{b}_{jj}$ of least-squares estimation of the bilinear model does \underline{not} equal the true diagonal specification $\alpha +  \beta$ in general.
\end{enumerate}
\end{proposition}
The conditions for equivalence of the diagonals are more restrictive than those for the equivalence of the off-diagonal elements of $\A$ and $\B$.  Furthermore, we see the BLIN model diagonals are consistent in misspecification situations when the bilinear diagonals are not consistent. 
In Section~\ref{section_sim_conv} of the Supplementary Material, we evaluate the convergence of the bilinear and full BLIN estimators in support of Theorem~\ref{thm:BLIN_bilinear_offdiag} and Proposition~\ref{prop:BLIN_bilinear_diag}.

Although much of this section  describes the similarities between the BLIN and bilinear models,  we emphasize that these similarities lie in estimation of the coefficient matrices $\A$ and $\B$. The estimated mean function of the BLIN and bilinear models are very different (even asymptotically), which has implications for model fit and prediction. To illustrate the difference in mean functions between the BLIN and bilinear models, we show that equivalence of $\A$ and $\B$ between the BLIN and bilinear models implies equivalence in the estimated mean $\hat{\Y}_t$ only when $\A$ and $\B$ are diagonal. 
\begin{proposition}
\label{prop:BLIN_bilinear_mean}
Let the estimators from the BLIN and bilinear models have diagonals that are equal up to a constant, as in Theorem \ref{thm:BLIN_bilinear_offdiag}. Additionally, let the estimators of the diagonals be constant
as in Proposition~\ref{prop:BLIN_bilinear_diag}, part 1. Then, 
under the conditions of Proposition~\ref{prop:BLIN_bilinear_diag} with $\bOmega$ and $\bPsi$ diagonal, equivalence of $\hat{\Y}_t$ for both models implies that the BLIN and bilinear estimators, $\{ \tilde{\A}, \tilde{\B}\}$ and $\{ \bar{\A}, \bar{\B}\}$, respectively, are all diagonal.
\end{proposition}

\section{Simulation Study}
\label{section:simulation}

As discussed in the previous section, although the BLIN and bilinear estimators of $\A$ and $\B$ are asymptotically equivalent under certain conditions, their mean functions are only equivalent under strict conditions (see Proposition~\ref{prop:BLIN_bilinear_mean}). Thus, the ability of these models to represent the variation in a bipartite relational data set will heavily depend on whether the true mean is of BLIN or bilinear form. To compare the ability of each model to represent mean structure under model misspecification, we 
conducted a simulation study. 

We generated from a lag-1 vector-autoregressive model \begin{align}
    \y_t = \bTheta \y_{t-1} + \e_t, \hspace{.25in} t \in \{1,2,\ldots,T \}, \label{eq_varsim}
\end{align}
where $\y_t = \text{vec}(\Y_t)$ is the columnwise vectorization of the $10 \times 10$ matrix of bipartite relations $\Y_t$, $\e_t$ consists of i.i.d. standard normal entries, and the number of time periods $T \in \{10, 20, 50 \}$. Recall that both the BLIN and bilinear models may be cast as VAR models of the form of \eqref{eq_varsim}, where $\bTheta = \A^T \otimes \I_L + \I_S \otimes \B^T $ for the BLIN model and $\bTheta = \B^T \otimes \A^T$ for the bilinear model. We created weighted, directed $\A$ and $\B$ matrices and used these to generate data from both models.  

The process for specifying $\A$ and $\B$ was motivated by a desire to construct matrices with network structure that we might expect in an influence network.  Initially, $\A$ and $\B$  were  randomly generated as matrices of rank 1, which may be viewed as latent factor models of rank 1 \citep{hoff2008modeling, li2011generalized}.   The diagonal entries and smallest $q=0.9$ fraction of off-diagonal entries were set to zero, such that the matrices were approximately low rank and sparse. Finally, we scaled the $\bTheta$ matrix of each generating model to control the signal-to-noise ratio, such that the true model had an $R^2$ of approximately 0.75 for $T$ approaching infinity. For further details on the simulation study, including investigation of $q=0.5$ and $q=0.0$, please see Section~\ref{section_sim_details} of the Supplementary Material.  

We generated 100 data sets from both the BLIN and bilinear model for $T=50$, and evaluated the out-of-sample predictions from the models in a 10-fold cross validation study. To compare performance of the estimators on data sets containing varying amounts of time periods, we evaluated performance on the complete data sets with $T=50$, as well as the performance when the data were trimmed to include only the last $10$ time periods and only the last $20$ time periods.  In each of these three scenarios, i.e. $T \in \{10, 20, 50 \}$,  we performed a 10-fold cross validation.  The time periods were randomly partitioned into 10 sets (the partitions were the same for all data sets of a given size $T$ for the sake of equal comparison). Models were fit to the data in nine of the ten sets, and then predictions for the values in the left-out time period set were obtained.  This fitting procedure was repeated 10 times: once for each of the partitions.   Models were evaluated based on the $R^2$ value between the ten sets of predicted values and the true values, for each data set and each $T$.  $R^2$ is a natural measure of model fit as the data are normally distributed and hence high $R^2$ values corresponds to large likelihood values.  For each of the two generative models (BLIN and bilinear) and each of the three data set sizes ($T \in \{10, 20, 50 \}$), we compare the performance of four models: the bilinear model and the full, reduced rank (with rank set to 1), and sparse BLIN models.

We plot the resulting $R^2$ values in Figure~\ref{fig:misspec2} for sparsity of the generating coefficients $q=0.9$. 
A dotted horizontal line is drawn at $R^2=0$, which denotes the expected performance of fitting no model at all, that is, predicting $\hat{\Y}_t = \mathbf{0}$ for all $t$. An additional dotted horizontal line is drawn at $R^2=0.75$, the expected large-sample $R^2$ value when the true model is known. When generating from the BLIN model (left panel of Figure~\ref{fig:misspec2}), the estimated full and sparse BLIN models perform well for all values of $T \in \{10,20,50 \}$. This is as expected, since the true matrices $\A$ and $\B$ are sparse. The reduced rank BLIN model is not able to represent this sparse structure, and thus its out-of-sample performance falls short of the full and sparse BLIN models.
When generating from the BLIN model, the bilinear model results in extremely poor predictions when $T=10$ and marginal performance when $T=20$, yet it is on par with the BLIN models for $T=50$.

\spacingset{1}
\begin{figure}[ht!]
\centering
\begin{tabular}{r c c }
  &\multicolumn{2}{c}{\textbf{Generating model}} \\
  &&\\
  & { BLIN } 
  \hspace{-.2in} & \hspace{-.2in} 
  Bilinear \\
\begin{sideways} \hspace{.45in} Out-of-sample $R^2$ \end{sideways} &
 \includegraphics[width=.46\textwidth]{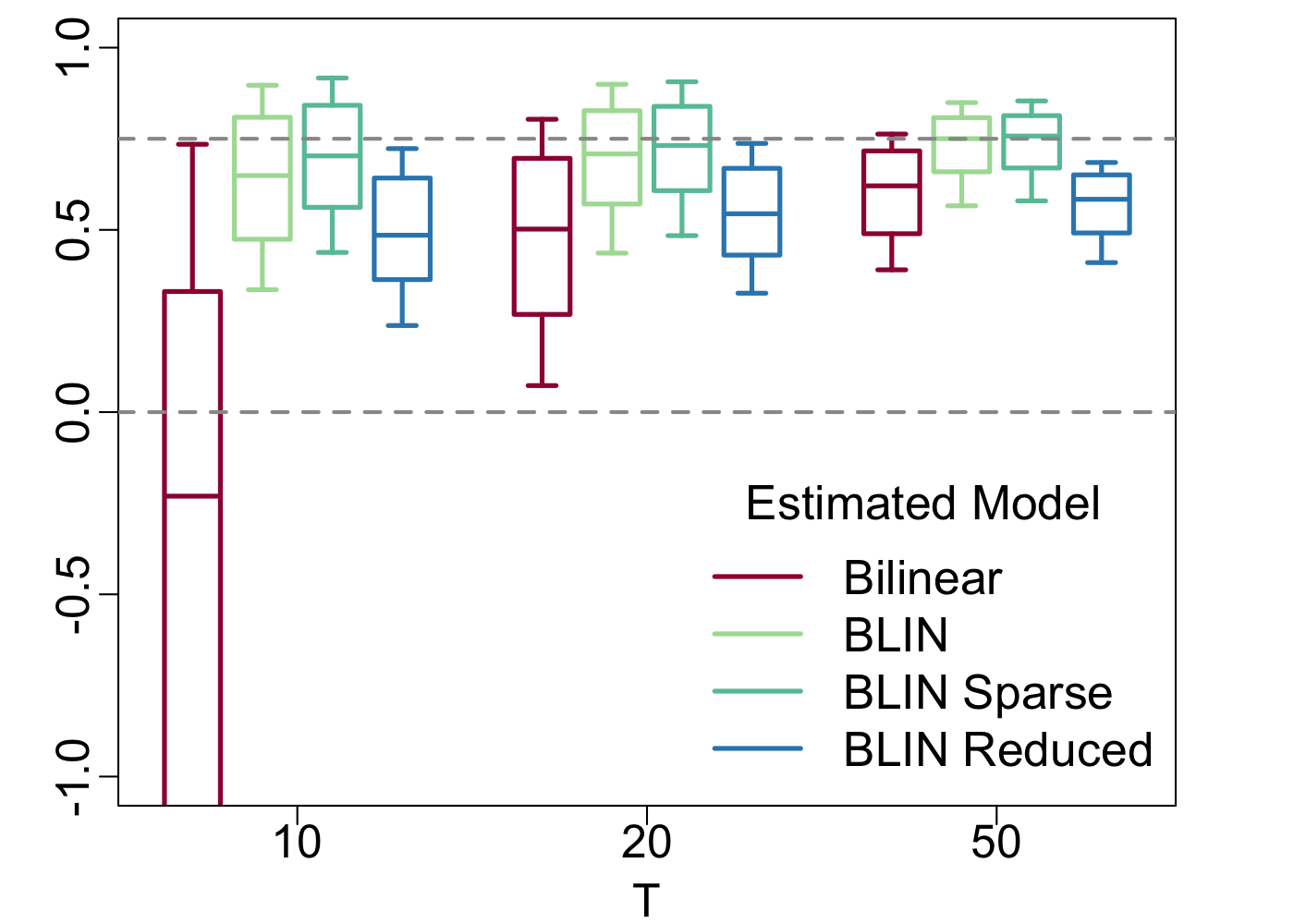}  
\hspace{-.2in} & \hspace{-.2in} 
 \includegraphics[width=.46\textwidth]{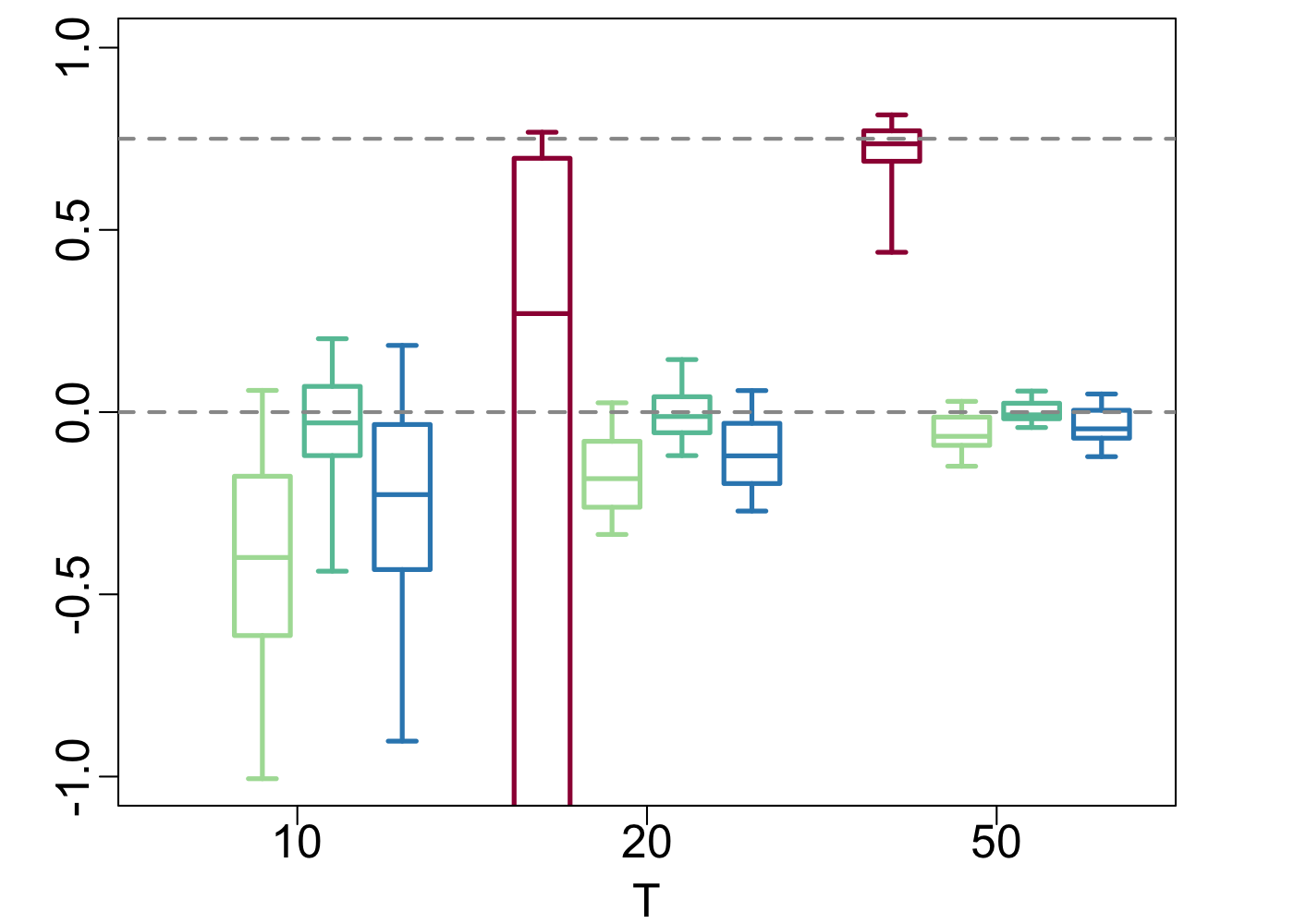} \\
\end{tabular}
\caption{Out-of-sample $R^2$ values for each estimation procedure applied to 100 data realizations generated from the BLIN model (left panel) and generated from the bilinear model (right panel).
The centers of the boxplots represent the median $R^2$ value, the boxes represent the middle 80\% of $R^2$ values, and the whiskers correspond to the maximum and minimum $R^2$ values across 100 simulated data sets. 
Plots are truncated such that $R^2$ values less than $-1$ are not shown.}
\label{fig:misspec2}
\end{figure}
\spacingset{1.5}

When generating from the bilinear model (right panel of Figure~\ref{fig:misspec2}), as expected, the bilinear model performs best for the largest number of replications $T=50$.  Surprisingly, only for this value does the bilinear model consistently outperform predicting simply $\hat{\Y}_t = \boldsymbol{0}$, even though the data are generated from the bilinear model.  
Overall, the BLIN models perform poorly when generating from the bilinear model. However, we note that the sparse BLIN model guards against poor performance (the typical $R^2$ is always about $0$), whereas the full and reduced rank BLIN models do not.

We investigated the source of errors when the data were generated from the bilinear model by examining the likelihood surface of the bilinear model near the estimated and true values of the influence matrices (Section~\ref{section_sim_details} in the Supplementary Material). These investigations showed that, when $T \in \{ 10, 20 \}$, the likelihood surface of the BLIN model is multimodal and that the highest mode may be ``far'' from the mode that is nearest the true parameter values. This means that estimating the model on one portion of data may not generalize well to other portions of the data, which is exactly what we observe in the negative $R^2$ values of the bilinear estimator when the data are generated from the bilinear model for $T \in \{10,20 \}$ (Figure~\ref{fig:misspec2}). 
A similar analysis for the BLIN model confirmed the unimodality implied by Propostion~\ref{prop_rankXb}. The results of the simulation study suggest that the bilinear model may be difficult to estimate unless a large number of replications $T$ are observed. In addition, although the coefficients estimated by the BLIN and bilinear models may be similar for large $T$ (as suggested by theory and confirmed by simulation in Section~\ref{section_sim_details} in the Supplementary Material), neither model may be a good predictive substitute for the other in small samples.

\section{Temporal State Interaction Data Analysis}
\label{section_temporal_SID}

Power, defined as the ability of an actor to influence another to do something they may not otherwise do, is perhaps the most essential concept in the study of politics broadly, and interstate relations specifically \citep{morgenthau1948struggle, dahl1957concept, waltz1979theory, lukes2004power}. 
Thus, both network researchers and political scientists are frequently interested in inferring patterns of influence in interactions between states \citep{minhas2017influence}.  Examining country behaviors, we may use the proposed approach to directly infer the most powerful and influential actors in modern international relations and also gain insights that may yield valuable predictions of international policy changes, such as ratification of environmental treaties  \citep{campbell2018latent}. Inferring these influences resolves a significant shortcoming in the international relations literature, which has been forced to measure power in terms of the relative military or economic strength because of the empirical difficulties associated with measuring power as influence \citep{singer1972capability, hart1976three, johnson2017external}. 

To infer country influences, we analyzed
country interactions at weekly intervals from 2004 to mid-2014, giving $T=543$ weeks of data. The relations were obtained from the Integrated Crisis Early Warnings System (ICEWS) \citep{28075_2015}, previously analyzed in \cite{hoff2015multilinear}, which automatically identifies and encodes interaction intensities (between -10 and 10) from news stories. Each relation is one of four interaction types from a source state to a target state: material negative actions (\emph{mn}), material positive actions (\emph{mp}), verbal negative actions (\emph{vn}), and verbal positive actions (\emph{vp}). We analyzed only the 25 most active states. An example of each interaction type is boycotting for leadership change (\emph{mn}),  providing humanitarian aid (\emph{mp}), denying accusations by a target country (i.e. charges of genocide or other human rights violations) (\emph{vn}), and expressing intent to negotiate (\emph{vp}).

 We denote the intensity of relation in week $t$ as $y_{ijk}^t$, where $i$ is the source state, $j$ is the target state, $k$ is the relation type, and $t$ is the week of the observation. Each time series $\{ y_{ijk}^t \}_{t=1}^T$ is centered and standardized.  Here we analyze the \emph{change} in relations using the differences $d_{ijk}^{t} := y_{ijk}^{t} - y_{ijk}^{t-1}$. We note that the relations in the ICEWS data are not classically bipartite, as the source states and target states consist of the same set of actors, namely, the 25 most active states. However, we use the BLIN model to make the desired estimates of influences among source and target states separately. 
 
 The BLIN model defined in Section \ref{sec_model} is for longitudinal bipartite data, however the ICEWS data set contains tripartite data, such that observations are indexed by (source, target, type) triples.  For this reason, here we introduce an extension of the BLIN model for tripartite data, which we abbreviate the TLIN (tripartite longitudinal influence network) model, and point interested readers to Section~\ref{sec_multi} of the Supplementary Material for details on an extension of the BLIN model specification for arbitrary multipartite data.  In the TLIN model, the change in relation intensity in each week $t$, $d_{ijk}^{t}$, depends on the previous actions of states that influence source state $i$ through source influence matrix $\A$, the previous actions of states that influence target state $j$ through target influence matrix $\B$, and interaction types that influence interaction type $k$ through the interaction type influence matrix $\C$. For the ICEWS data the TLIN model can be expressed 
\begin{align}
d_{ijk}^{t} &= \sum_{s=1}^{25} a_{si} \left(\sum_{r=1}^{p_A}  d_{sjk}^{t - r} \right) + 
\sum_{\ell=1}^{25} b_{\ell j}  \left( \sum_{r=1}^{p_B} d_{i \ell k}^{t - r} \right) 
+ \sum_{u=1}^{4} c_{u k}  \left( \sum_{r=1}^{p_C} d_{i j u}^{t - r} \right)
+ e_{ijk}^t, \label{eq_blin_icews}
\end{align}
where each $e_{ijk}^t$ is an independent, mean zero random error. The TLIN model in \eqref{eq_blin_icews} states that, if there is a positive source influence from the Russia to France in $\A$, then an observed increase in, say, boycott intensity from Russia to China implies that we should expect France to increase its boycotting of China in the following weeks. 
For $\B$, the  model states that if there exists a negative target influence from Lebanon to Pakistan, then an increase in US humanitarian aid sent to Lebanon would indicate a decrease in aid sent from US to Pakistan in the following weeks.  Finally, a positive influence of \emph{vn} on \emph{mn} in $\C$ suggests that verbal negative interactions precede material negative interactions; that is, if the US threatens to increase boycotts on North Korea, then we might expect the US to increase tariffs on North Korea in the following weeks. 
Finally, we contend that imposing a boycott is fundamentally different than having a boycott imposed upon one's state, and thus,
although the country sets of sources and targets are the same, we treat these nodes as separate types and infer separate influence networks among source and target countries.   

To choose the lag values $\{p_A, p_B, p_C\}$, corresponding to influences among source countries, target countries, and interaction types, respectively, we fit a sparse TLIN model to the ICEWS data for a range of lags using an $\ell^1$ penalty on the entries in the influence networks. We then chose the lags that gave the best balance of model fit and model parsimony based on a comparison of likelihood values of the estimated models, penalized by twice the number of nonzero estimated parameters, in the vein of Akaike's Information Criteria \citep{akaike1998information}.  This procedure resulted in lag values
$\{p_A = 5, p_B = 3, p_C=1 \}$.  We note that the choice to difference the responses is a departure from the analysis in \cite{hoff2015multilinear}. For more details of the data analysis, see Section~\ref{section_SID_details} of the Supplementary Material.

The estimated source country ($\A$), target country ($\B$), and interaction type ($\C$) influence networks are depicted in Figure~\ref{fig:mp_ICEWS}, where we focus on the entries that constitute the largest 5\% of magnitudes across all networks (see Figure~\ref{fig:coef_ordered} of the Supplementary Material).
Across all networks, there are generally larger and more positive entries than negative entries
(Figure~\ref{fig:mp_ICEWS}). This fact suggests that positive influence is more consequential than negative influence: e.g., increases in aid generally lead to other increases in aid, rather than decreases. For example, in the target influence network $\B$, there is a positive influence from North Korea (PRK) to China (CHN) and a smaller, negative influence from North Korea to Afghanistan (AFG). This suggests that an increase in boycotts by Great Britain on North Korea leads one to expect an increase in boycotts by Great Britain on China in the following three weeks, and a smaller decrease in boycotts by Great Britain on Afghanistan.

\spacingset{1}
\begin{figure}[ht!]
\centering
\includegraphics[width=.99\textwidth]{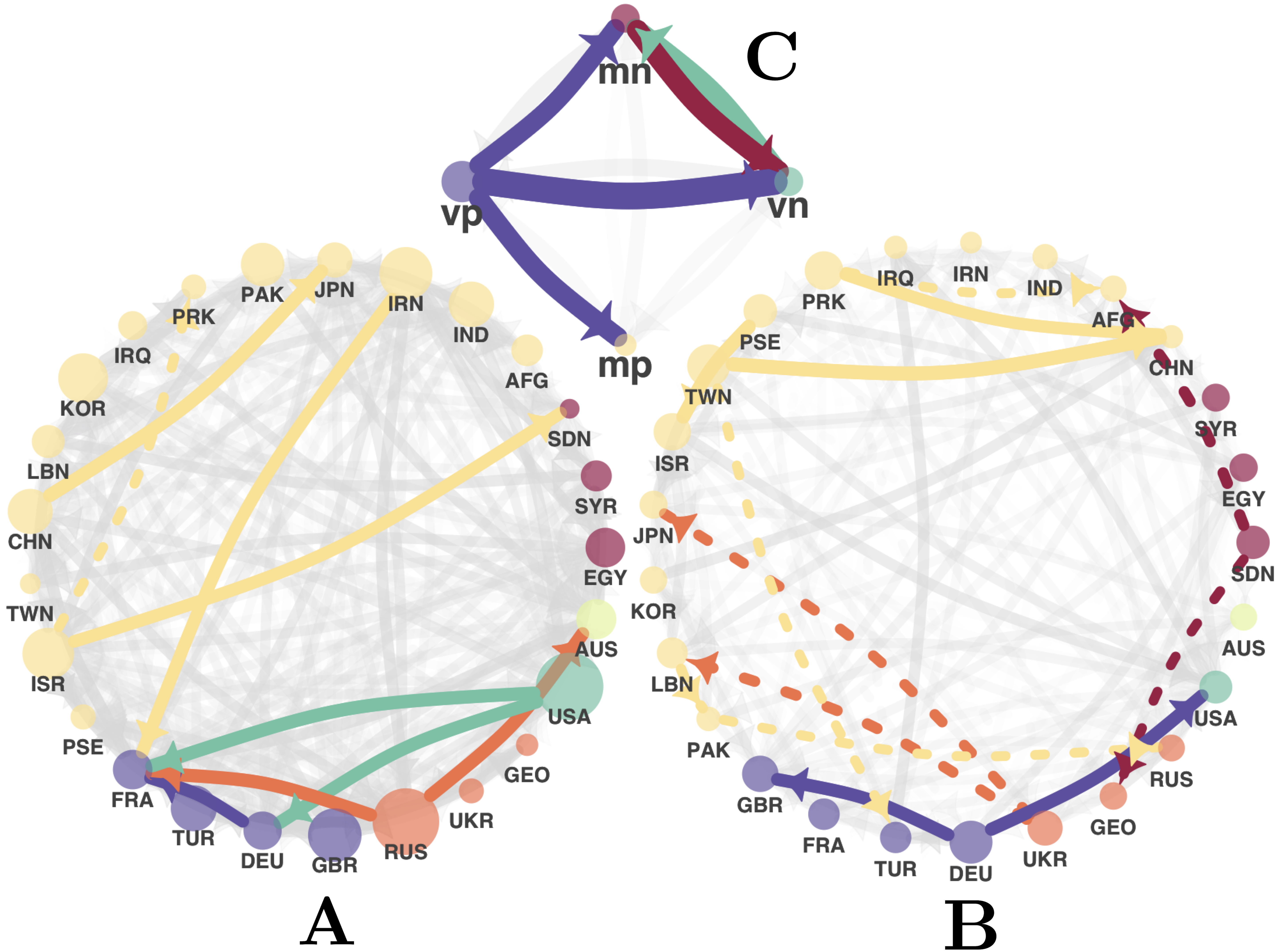} 
\caption{The largest 5\% of edges (in magnitude) are shown in color; the smaller relations are shown in grayscale proportional to their magnitude. The source country network is in the bottom-left panel ($\A$), the target country network is in the bottom-right panel ($\B$), and the interaction type network is in the top-center panel ($\C$). Nodes are sized proportional to the sum of the magnitudes of outgoing relations (comparable across networks) and nodes in $\A$ and $\B$ are colored according to continent. Edges are sized proportional to absolute edge weight (comparable across networks) and colored according to the originating node. Solid lines denote positive relations and dotted lines denote negative relations. 
}
\label{fig:mp_ICEWS}
\end{figure}
\spacingset{1.5}

One of the largest positive influences in the $\A$ network is that of the US on Germany (DEU). This means that if the US increases humanitarian aid sent to Syria, for example, then we expect Germany to increase the amount of humanitarian aid it sends to Syria in the following $p_A = 5$ weeks.  
This result matches conventional wisdom as the US and Germany have been among the closest North Atlantic Treaty Organization (NATO) allies following the Cold War, also sharing many common geopolitical interests outside this particular alliance.  
Within $\A$, we note that the US and Russia have the strongest ties with other countries and the most ties across continents.  This also may have been expected since great powers should possess the most influence relationships and the most substantively important influence relationships \citep{morgenthau1948struggle, waltz1979theory, mearsheimer2001tragedy}. 

In the target country $\B$ influence network, we observe that the most influential country is Germany (DEU), suggesting again that Germany is very central in the international community.  We see that the $\B$ network has many more large negative entries than $\A$ and $\C$. This fact is natural, as, for some relations, target states are competing for resources and hence, influences may be  limited. 
The material positive and material negative interactions are zero-sum, that is, a dollar sent to Sudan cannot also be sent to Afghanistan. 
For example, there is a large negative influence in $\B$ from the Sudan (SDN) to Afghanistan (AFG).  This influence indicates that, when Sudan receives an increase in, say, humanitarian aid from the US, Afghanistan is expected to have reduced aid from the US in the coming weeks. Hence, these countries are essentially competing for the resources of source countries, of which resources there are finite amounts.  We do not observe this phenomenon in the $\A$ network as the Germany following the US in aiding Sudan employs both the resources of the German and US people.

From $\C$ we glean understanding of the relationships between interaction types. For example, we observe positive relations between material negative and verbal negative interaction types, which means, for example, that increases (or decreases) in verbal negative interactions from the US to China may signal increases (decreases) in material negative interactions from the US to China. We also observe large positive influences from \emph{vp} to the other relation types. This fact indicates that changes in verbal positive interactions are often followed by changes in other interaction types, but  that a change in verbal positive relations is relatively uninformative to what type of interaction may change afterwards.  Finally, since $\C$ is entirely positive, this indicates positive feedback loops among the interaction types and a lack of negative feedback loops.

With the development of the BLIN model and multipartite extensions, international relations scholars no longer must rely upon the widely used proxy measures of power and influence, such as national capabilities or economic development. 
Instead, they are now permitted to measure power directly, as the influence exercised by one state over another. 
This approach has already been adopted and shown to be powerful in modeling the influence that countries exercise over one another in the ratification of environmental treaties \citep{campbell2018latent}.

\section{Discussion}
\label{sec:conc}
In this paper, we present the Bipartite Longitudinal Influence Network model, a novel generative model for the evolution of bipartite relational data over time. The BLIN model allows for the estimation of the weighted and directed influence networks among each of the two actor types in bipartite data, 
each with their own separate time scale of influence. The BLIN model can be expressed as a generalized linear model, lending itself to use with a litany of off-the-shelf tools for estimation and to straightforward parameter interpretation.

In the BLIN model, the entries in the influence networks $\A$ and $\B$ may be interpreted as ``average'' influences over the entire data set.  There are scenarios where the influences among countries may evolve over time. For example, major international trade agreements or wars might change the structures of the influence networks.  To determine whether $\A$ and $\B$ evolve over time, we might consider testing for changepoints in the influence networks.  Of course, this requires the development of a more flexible model with time-varying influence networks, such as modeling $\A$ and $\B$ as linear functions of known covariates as in \cite{minhas2017influence}.  \\

\vspace{.1in}
\if0\blind
{
\noindent\textbf{Acknowledgements  }  
We would like to thank the anonymous reviewers, the associate editor,  and the editor, whose helpful suggestions greatly improved the quality of the paper. 
This work was partially supported primarily by the National Science Foundation under Grant SES-1461493 to Cranmer and Grant SES-1461495 to Fosdick. Cranmer was also supported by NSF grants SES-1357622 and SES-1514750 as well as NIH R-34, DA043079 and the Alexander von Humboldt Foundation's fellowship for experienced researchers. This work utilized the RMACC Summit supercomputer, which is supported by the National Science Foundation (awards ACI-1532235 and ACI-1532236), the University of Colorado Boulder, and Colorado State University. The RMACC Summit supercomputer is a joint effort of the University of Colorado Boulder and Colorado State University. 
} \fi

\spacingset{1}
\bibliography{BLIN.bib}
\bibliographystyle{apalike}
\spacingset{1.5}

\newpage
\appendix
\bigskip
\begin{center}
{\Large\bf Supplementary Material}
\end{center}

\pagenumbering{arabic}   

\section{Least Squares Estimation of Reduced-Rank BLIN Model}
\label{sec:MLE_reduced_rank}
In this section,
we provide a procedure for obtaining the maximum likelihood estimator of the BLIN model assuming reduced rank coefficient matrices $\A$ and $\B$, where the respective ranks are known. The log-likelihood of the data $\{ \Y_t \}_{t=1}^T$ is simply the sum of the log-likelihood at each time period since we assume the errors $\E_t$ are independent of each other.  The log-likelihood, in terms of the unknown matrices $\{\U, \V, \bR, \S \}$, is proportional to
\begin{align}
\ell &\propto \frac{-1}{2}\sum_{t=1}^T || \Y_t - \U \V^T \X_t - \Z_t \bR \S^T ||_F^2, \label{eq_ll_rr}
\end{align}
where $\U \V^T$ is a rank $k$ decomposition of $\A$, and $\bR \S^T$ is a rank $m$ decomposition of $\B$.

We introduce a block coordinate descent algorithm in Algorithm~\ref{alg:BLIN_it_rr} to obtain the maximum likelihood estimator of the unknown matrices $\{\U, \V, \bR, \S \}$.
Algorithm~\ref{alg:BLIN_it_rr} estimates each unknown matrix $\{\U, \V, \bR, \S \}$ 
in turn until convergence, in an analogous procedure to the estimation of the full BLIN model in Algorithm~\ref{alg:BLIN_it}. The update equation for each unknown matrix is derived by differentiating \eqref{eq_ll_rr} with respect to the unknown matrix, e.g. $\U$, setting this derivative to zero, and solving for the unknown matrix.

\spacingset{1}
\begin{algorithm}
\caption{Block coordinate descent LS estimation of reduced-rank BLIN model}
\label{alg:BLIN_it_rr}
\begin{enumerate}
  \setcounter{enumi}{-1}
\item Set threshold for convergence $\eta$. Set number of iterations $\nu = 1$. Initialize $\{ \hat{\U}^{(0)}, \hat{\V}^{(0)}, \hat{\bR}^{(0)}, \hat{\S}^{(0)} \}$ with independent standard normal entries and $Q_0 = \sum_t ||\Y_t||^2_F$.
\item 
Compute
{\footnotesize
\[(\hat{\U}^{(\nu)})^{T} = \left( (\hat{\V}^{(\nu-1)})^T \sum_{t=1}^T  \left(\X_{t} \X_{t}^T \right)\hat{\V}^{(\nu-1)} \right)^{-1}  (\hat{\V}^{(\nu-1)})^T \left( \sum_{t=1}^T  \X_{t}\Y_{t}^T - \sum_{t=1}^T \X_{t} \hat{\S}^{(\nu-1)} (\hat{\bR}^{(\nu-1)})^T \Z_{t}^T  \right)\]
}
\item Compute 
{\footnotesize
\[\hat{\V}^{(\nu)} = \left( \sum_{t=1}^T \X_t  \X_t ^T \right)^{-1} \left( \sum_{t=1}^T  \X_t \Y_t ^T - \sum_{t=1}^T \X_t  \S^{(\nu - 1)}(\bR^{(\nu-1)})^T \Z_t ^T  \right) \hat{\U}^{(\nu)} \left((\hat{\U}^{(\nu)})^T \hat{\U}^{(\nu)} \right)^{-1}\]
}
\item Compute 
{\footnotesize
\[\hat{\bR}^{(\nu)} = \left( \sum_{t=1}^T  \Z_t ^T  \Z_t \right)^{-1} \left[  \sum_{t=1}^T  \Z_t ^T \Y_t  - \sum_{t=1}^T  \Z_t ^T  \U^{(\nu)} (\V^{(\nu)})^T  \X_t \right] \S \left(\S^T  \S \right)^{-1}\]
}
\item Compute 
{\footnotesize
\[(\hat{\S}^{(\nu)})^T = \left( (\bR^{(\nu)})^T  \left(\sum_{t=1}^T  \Z_t ^T  \Z_t  \right) \bR^{(\nu)} \right)^{-1} (\bR^{(\nu)})^T  \left[ \sum_{t=1}^T  \Z_t ^T \Y_t   - \sum_{t=1}^T  \Z_t ^T  \U^{(\nu)} (\V^{(\nu)})^T  \X_t \right]\]
}
\item Compute the least squares criterion 
{\footnotesize
\[Q_{\nu} = \sum_t ||\Y_t - \U^{(\nu)} ( \V^{(\nu)} )^T \X_t - \Z_t \bR^{(\nu)} (\S^{(\nu)})^T ||^2_F.\]
}
If  $|Q_{\nu} - Q_{\nu-1}| > \eta$, increment $\nu$ and return to 1.  
\end{enumerate}
\end{algorithm}
\spacingset{1.5}

\clearpage
\section{Proofs of Theoretical Results}
\label{AppProofs}

We begin with the proof of Proposition~\ref{prop_rankXb}.
We then prove Theorem~\ref{thm:BLIN_bilinear_offdiag}, as expressions in this proof support the proofs of
the remaining Propositions~\ref{propPseudo1},  \ref{propPseudo2}, \ref{prop:BLIN_bilinear_diag}, and \ref{prop:BLIN_bilinear_mean}. We take $\X_t$ and $\Y_t$ mean zero without loss of generality.

\vspace{.2in}
\begin{proof}[Proof of Proposition~\ref{prop_rankXb} \\]
Noting that $\mathbb{X}_B$ is of dimension $TSL \times S^2 + L^2 - 1$, it is sufficient to show that $\mathbb{X}_B$ is full rank under the assumptions. We treat two cases: (1) $TSL \le S^2 + L^2 - 1$ and (2) $TSL > S^2 + L^2 - 1$. 

\noindent \textbf{Case (1): \\}
\noindent We first show that $T S L \le (S^2 + L^2 - 1)$ implies $TS \le L$. 
Assume towards a contradiction that $TS > L$ and let $S = \alpha L$ for $\alpha \in (0,1]$. Then,
\begin{align}
TSL &\le S^2 + L^2 - 1, \nonumber \\
T\alpha L^2 & \le (1+\alpha^2)L^2 - 1, \nonumber \\
T &\le 1/\alpha + \alpha -1/L^2. \label{eq:T_ub}
\end{align}
If $TS > L$, then $T > 1/ \alpha$. As there is no integer $T$ that satisfies \eqref{eq:T_ub} and $T > 1/ \alpha$, we have that $TS \le L$.

Now, as $TSL \le S^2 + L^2 - 1$, we have that ${\rm rank}(\mathbb{X}_B) \le TSL$. Assume towards a contradiction that ${\rm rank}(\mathbb{X}_B) < TSL$. Then, for some nonzero $u \in \R^{TSL}$, 
the assumption implies that $u^T \mathbb{X}_B = 0$. Consider the columns $S^2 +1$ through $S^2 +1 +L$ of $\mathbb{X}_B$. The assumption implies that, for some nonzero $v \in \mathbb{R}^{TS}$, that $v^T [\Z_1 ; \Z_2 ; \ldots ;\Z_t]$. For $TS \le L$, this is a contradiction of the assumption that $[\Z_1 ; \Z_2 ; \ldots ;\Z_t]$ is full rank.
Thus, we have that $\mathbb{X}_B$ is full rank in case (1).

\noindent \textbf{Case (2): \\}
\noindent Now we take $TSL > S^2 + L^2 - 1$, such that ${\rm rank}(\mathbb{X}_B) \le S^2 + L^2 - 1$. Assume towards a contradiction that ${\rm rank}(\mathbb{X}_B) <  S^2 + L^2 - 1$.
Then, there exists some $u_1 \in \mathbb{R}^{S^2}$ and $u_2 \in \mathbb{R}^{L^2}$ such that $\mathbb{X}_B u = 0$, for $u^T = [u_1 , u_2]$, nonzero and not utilizing the nonidentifiability exclusively.
By the single nonidentifiability of the BLIN model, for $\mathbb{X}_B u = 0$, either $u$ utilizes the nonidentifiability or both $u_1$ and $u_2$ are in the null spaces of $[\X_1^T \otimes \I_L ; \X_2^T \otimes \I_L; \ldots ;\X_t^T \otimes \I_L]$ and $[\I_S \otimes \Z_1 ; \I_S \otimes \Z_2 ; \ldots ; \I_S \otimes \Z_t]$, respectively. However, 
by the discussion under case (1), we have that $TS \ge L$ such that, by assumption, $[\X_1 ; \X_2 ; \ldots ;\X_t]$  and $[\Z_1 ; \Z_2 ; \ldots ;\Z_t]$ are of rank $L$. This also implies that the $TL \times S$ matrix  $[\X_1^T ; \X_2^T ; \ldots ;\X_t^T]$ is full rank, which is rank $S$ as $TL \ge S$. Then, again using the discussion in case (1), the two matrices $[\I_S \otimes \Z_1 ; \I_S \otimes \Z_2 ; \ldots ; \I_S \otimes \Z_t]$ and $[\X_1^T \otimes \I_L ; \X_2^T \otimes \I_L; \ldots ;\X_t^T \otimes \I_L]$ are full rank, which are $S^2$ and $L^2$, respectively.  
So, the only way that $\mathbb{X}_B u = 0$ and $u$ is nonzero is utilizing the nonidentifiability. This is a contradiction, as this implies that ${\rm rank}(\mathbb{X}_B) = S^2 + L^2 - 1$, and $\mathbb{X}_B$ is full rank under case (2).

\end{proof}

\vspace{.25in}
\begin{proof}[Proof of Theorem~\ref{thm:BLIN_bilinear_offdiag} \\]
Recall that we work in the setting of $p=p_A=p_B$, as the bilinear model cannot accommodate different lags for the different influence types. Thus, we have $\Z_t = \X_t$ for all $t \in \{1,2,\ldots, T \}$ and the covariance matrices $\bOmega_X = \bOmega_Z =\bOmega$ and $\bPsi_X = \bPsi_Z =\bPsi$.

For either generative model, we may write the useful form
\begin{align}
E[\y_t|\x_t] = \bTheta \x_t, \label{eq_bigtheta}
\end{align}
such that $\bTheta = [\I_L \otimes \A ; \B^T \otimes \I_S]$ for the BLIN model and $\bTheta = \B^T \otimes \A$ for the bilinear model.
We examine the impacts of specifying either the BLIN or bilinear model as the generating model in the proof. First, however, we write the pseudo-true parameters for least squares estimators of the BLIN and bilinear models, respectively, under the general generative structure in \eqref{eq_bigtheta}. 

For least squares estimation of the BLIN model, we may write the pseudo-true parameters for $\A$ as 
\begin{align}
\tilde{\A}^T &= E\left[ \Y_t \X_t^T - \X_t \tilde{\B} \X_t ^T \right] E \left[\X_t \X_t^T \right]^{-1}, \label{eqA0}  \\
&= E \left[\sum_{j = 1}^L E [\y_{\cdot jt} \ | \ \x_t] \x_{\cdot jt}^T - \sum_{i = 1}^L \sum_{j = 1}^L\tilde{b}_{ij} \x_{\cdot it} \x_{\cdot jt}^T \right] E \left[\X_t \X_t^T \right]^{-1}, \label{eqA0_sum}
\end{align}
where we obtain \eqref{eqA0} by maximizing the expression for the BLIN pseudo-true parameters in \eqref{eqRR2norm} with respect to $\A$. Then, we may exploit the assumption that $E [\y_{\cdot jt} \ | \ \x_t]$ is a linear function of $\x_t$. Letting $\s{C}_{k \ell}$ be the $S \times S$ partition of $\bTheta$ relating column $\ell$ of $\X_t$ to column $k$ of $\Y_t$, substituting into \eqref{eqA0_sum}, we obtain
\begin{align}
 \tilde{\A}^T &= \left( \sum_{k=1}^L \sum_{\ell=1}^L \s{C}_{k \ell} \omega_{k \ell} - {\rm tr}(\bOmega \tilde{\B})  \I_S \right) \Big/ {\rm tr}(\bOmega), \label{eqA1}
\end{align}
where and $ \omega_{k \ell }$ is the $(k, \ell)$ entry in $\bOmega$ (and we use the symmetric property of $\bOmega$) and the terms concerning $\bPsi$ cancel. We note that the BLIN pseudo-true parameters $\tilde{\A}$ exist by applying Conditions~\ref{cond:cond} to the explicit expression for $\hat{\A}$ in \eqref{eqABexact}, using the law of large numbers. 
\cite{hoff2015multilinear}
writes the pseudo-true parameters for the bilinear estimator of $\A$, which we denote $\bar{\A}$, under least squares estimation as:
\begin{align}
\bar{\A}^T &= E\left[ \Y_t \bar{\B} \X_t \right] E \left[\X_t \bar{\B} \bar{\B}^T \X_t^T \right]^{-1}, \\
&= {\rm tr}(\bOmega \bar{\B} \bar{\B}^T)^{-1} \sum_{j=1}^L \sum_{k=1}^L  \left( \sum_{\ell = 1}^L \bOmega_{k \ell} \bar{b}_{j\ell} \right) \s{C}_{jk}, \label{eq:pseudo_bil}
\end{align}
where $\bar{b}_{j \ell}$ is the $(j, \ell)$ entry in  $\bar{\B}$, the pseudo-true parameters for the bilinear estimator of $\B$. The bilinear pseudo-true parameters $\bar{\A}$ exist as Conditions~\ref{cond:cond} satisfy those given in \cite{hoff2015multilinear}.

We have derived the pseudo-true parameters for the least squares estimators of the BLIN and bilinear models in \eqref{eqA1} and \eqref{eq:pseudo_bil}, respectively, whenever $E[\y_t|\x_t] = \bTheta \x_t$. Now, we address the specific BLIN and bilinear generating models. That is, we specify the $S \times S$ matrices $\{ \s{C}_{k \ell}\}_{k, \ell}$ that partition $\bTheta$ under each generating model and examine the resulting pseudo-true parameters. 
 
When the data $\{ \Y_t \}_{t=1}^T$ are generated by the BLIN model, 
the matrix $\s{C}_{jk} = \A^T + b_{jk}\I_s$ when $k=\ell$ and $\s{C}_{jk}$ is diagonal otherwise. Substituting into \eqref{eq:pseudo_bil}, we see:
\begin{align}
\bar{\A}^T &= {\rm tr}(\bOmega \bar{\B} \bar{\B}^T)^{-1} \sum_{j=1}^L \left( \sum_{\ell = 1}^L \omega_{j \ell} \bar{b}_{j\ell} \right) \A^T + c_1 \I_s, \\
&= \frac {{\rm tr}(\bOmega \bar{\B} )} {{\rm tr}(\bOmega \bar{\B} \bar{\B}^T)} \A^T + \frac {{\rm tr}(\bOmega \bar{\B} \B^T)} {{\rm tr}(\bOmega \bar{\B} \bar{\B}^T)} \I_s. \label{eqA_pseudo_bilin_final}
\end{align}
Thus, the off-diagonal pseudo-true parameters 
$\bar{\A}$ are equal to the off-diagonal entries in $\A$, up to a multiplicative constant, when the data are generated by the BLIN model.

When the data $\{ \Y_t\}_{t=1}^T$ are generated by the bilinear model, the matrix $\s{C}_{jk} = {b}_{jk} \A^T$ for all $(j,k)$. Substituting into \eqref{eqA1}, we see
\begin{align}
\tilde{\A}^T &= \left( \sum_{k=1}^L \sum_{\ell=1}^L b_{k \ell} \omega_{\ell k} \A^T - {\rm tr}(\bOmega \tilde{\B})  \I_S \right) \Big/ {\rm tr}(\bOmega), \\
&= \frac { {\rm tr}(\bOmega \B)} {{\rm tr}(\bOmega) } \A^T - \frac { {\rm tr}(\bOmega  \tilde{\B})} {{\rm tr}(\bOmega) } \I_s, \label{eqA_pseudo_final}
\end{align}
and ${\rm tr}(\bOmega) \neq 0$ by assumption of $\bOmega$ positive definite.
Thus, the off-diagonal pseudo-true parameters for $\tilde{\A}$
are equal to the off-diagonal entries in $\A$, up to a multiplicative constant, when the data are generated by the bilinear model. 

These results hold for the off-diagonal pseudo-true parameters as estimated for the BLIN and bilinear models, $\tilde{\B}$ and $\bar{\B}$, respectively. This fact can be seen by transposing the BLIN and bilinear models and swapping the roles $\B$ for $\A$ in the above proof, i.e. writing the BLIN model as $E[\Y_t^T | \X_t] = \B^T \X_t^T + \X_t^T \A$ and applying the above arguments for $\A$ to $\B^T$. 
\end{proof}

\vspace{.25in}
\begin{proof}[Proof of Proposition~\ref{propPseudo1} \\]

The pseudo-true parameter under least squares estimation of the BLIN model, $\tilde{\A}$, is given in \eqref{eqA0}.
Under the assumptions in Proposition~\ref{propPseudo1}, the matrix $E \left[\X_t \X_t^T \right]^{-1}$ is diagonal with entries $1/ E[\x_{i \cdot t}^T \x_{i \cdot t}]$. Then, the off-diagonal $(i,j)$ entry of $\tilde{A}$ is
\begin{align}
\tilde{a}_{ij} = \frac{E[\x_{i\cdot t}^T \y_{j \cdot t}] - {\rm tr}( E[ \x_{i\cdot t} \z_{j \cdot t}^T] \tilde{\B}^T )}{E[\x_{i \cdot t}^T \x_{i \cdot t}]}, \ \ i\neq j. \label{eq_aij}
\end{align}
By the assumptions given in Proposition~\ref{propPseudo1}, both terms in the numerator are zero as is the coefficient. 

The result for $\tilde{b}_{ij}$ when $i \neq j$ follows from considering the transpose of the BLIN model as in the proof of Theorem~\ref{thm:BLIN_bilinear_offdiag}, that is, swapping the roles of $\A$ and $\B$ and the roles of $\X_t$ and $\Z_t$ in \eqref{eq_aij}. 
\end{proof}

\vspace{.25in}

\begin{proof}[Proof of Proposition~\ref{propPseudo2} \\]

Refer to the expression for the pseudo-true parameter under least squares estimation of the BLIN model, $\tilde{\A}$, in \eqref{eqA1}.
The entries relating row $i$ of $\X_t$ to row $j$ of $\Y_t$ are the $(j,i)$ entries in $\s{C}_{k \ell}$ for all $(k, \ell)$. These are zero by assumption when $i \neq j$. Thus, the assumptions imply $\tilde{a}_{ij} = 0$ when $i \neq j$. There is no issue when $\Z_t \neq \X_t$, as this change only enters \eqref{eqA1} through $\bOmega$, which value is immaterial to the argument. 

Again, the result for $\tilde{b}_{ij}$ when $i \neq j$ follows from considering the transpose of the BLIN model as in the proof of Theorem~\ref{thm:BLIN_bilinear_offdiag}.
\end{proof}

\vspace{.25in}
\begin{proof}[Proof of Proposition~\ref{prop:BLIN_bilinear_diag}]
As in the proof of Theorem~\ref{thm:BLIN_bilinear_offdiag}, we work in the setting of $p=p_A=p_B$, Thus, we have $\Z_t = \X_t$ for all $t \in \{1,2,\ldots, T \}$ and the covariance matrices $\bOmega_X = \bOmega_Z =\bOmega$ and $\bPsi_X = \bPsi_Z =\bPsi$. 

Under least squares estimation of the BLIN model, using \eqref{eqA_pseudo_final}, we may write the pseudo-true diagonal specification as
\begin{align}
\tilde{a}_{ii} + \tilde{b}_{jj} &= \frac{{\rm tr}(\B)}{L} a_{ii} + \frac{{\rm tr}(\A)}{S} b_{jj} - \frac{{\rm tr}(\A) {\rm tr}(\B)}{SL},
\end{align}
where we use the condition of matrices $\bOmega$ and $\bPsi$ proportional to the identities of appropriate size. Then, the traces of $\A$ and $\B$ are simply $\alpha S$ and $\beta L$, respectively. Substituting, we find 
\begin{align}
\tilde{a}_{ii} + \tilde{b}_{jj} &= \alpha \beta + \alpha \beta - \frac{SL \alpha \beta}{SL} = \alpha \beta \ \forall \ i,j,
\end{align}
which is the specification of the diagonal entries under the bilinear model, $\alpha \beta = a_{ii} b_{jj} \ \forall \ i,j$.

We now turn to least squares estimation of the bilinear model. It is sufficient to provide a counterexample to prove the claim. Set $\A = \alpha \I_S + \A_0$, where $\A_0$ has diagonal of all zeros, and the same for $\B = \beta \I_L + \B_0$, for some $\alpha + \beta \neq 0$. 
Using the expression for the bilinear pseudo-true parameters in  \eqref{eqA_pseudo_bilin_final}, we may write
pseudo-true diagonal specification under least squares estimation of the bilinear model as
\begin{align}
\bar{a}_{ii} \bar{b}_{jj}  &= \frac{ {\rm tr}(\bar{\A}) {\rm tr}(\bar{\B}) }{ {\rm tr}(\bar{\A}^T \bar{\A} ) {\rm tr}(\bar{\B}^T \bar{\B} ) } 
\left( 
\alpha+\beta + \frac{{\rm tr}(\bar{\A} \A_0^T)}{{\rm tr}(\bar{\A})} \right)
\left( 
\alpha+\beta + \frac{{\rm tr}(\bar{\B} \B_0^T)}{{\rm tr}(\bar{\B})} \right), \label{eq_diag_bil_1}
\end{align}
for any pair $(i,j)$.
Again using the expression for the bilinear pseudo-true parameters in  \eqref{eqA_pseudo_bilin_final}, we have that 
\begin{align}
\frac{ {\rm tr}(\bar{\A})  }{ {\rm tr}(\bar{\A}^T \bar{\A} )  } \frac{ {\rm tr}(\bar{\B})  }{ {\rm tr}(\bar{\B}^T \bar{\B} )  }
&= \alpha+\beta+ \frac{{\rm tr}(\bar{\A} \A_0^T)}{{\rm tr}(\bar{\A})}+\frac{{\rm tr}(\bar{\B} \B_0^T)}{{\rm tr}(\bar{\B})}. \label{eq_diag_bil_2}
\end{align}
Substituting \eqref{eq_diag_bil_2} into \eqref{eq_diag_bil_1}, we see that 
\begin{align}
\bar{a}_{ii} \bar{b}_{jj}  &=  \alpha + \beta + \frac{c_A c_B}{\alpha + \beta + c_A + c_B}, \label{eqDiag_pseudo_bilinear} \\
\text{ where \ } c_A :&= \frac{{\rm tr}(\bar{\A} \A_0^T)}{{\rm tr}(\bar{\A})}, \nonumber \\
c_B :&= \frac{{\rm tr}(\bar{\B} \B_0^T)}{{\rm tr}(\bar{\B})}. \nonumber
\end{align}
This expression is equal to the true diagonal specification $\alpha + \beta$ if and only at least one of $c_A$ or $c_B$ is equal zero. For example, take all entries in $\A_0$ and $\B_0$ to be zero except for $a_{12}$ and $b_{12}$. Then, 
\begin{align}
c_A c_B &= \frac{a_{12}^2 b_{12}^2}{SL \left( \alpha + \beta + c_A\right) \left( \alpha + \beta + c_B\right) } \neq 0.
\end{align}
This case establishes a counterexample for $\bar{a}_{ii} \bar{b}_{jj} \neq \alpha + \beta$.
\end{proof}

\vspace{.25in}
\begin{proof}[Proof of Proposition~\ref{prop:BLIN_bilinear_mean}]
By assumption, the diagonals of the BLIN estimators $\{ \tilde{\A}, \tilde{\B} \}$ are constant; we let these values be $\alpha$ and $\beta$, respectively. Thus, we have
\begin{align}
\tilde{\A} &= \alpha \I_S + \A_0, \hspace{.4in}    \tilde{\B} = \beta \I_L + \B_0, \label{eq_AB_blin_equiv}
\end{align}
where $\A_0$ and $\B_0$ are matrices with zeros along the diagonal. 

Now, the assumption of equivalences of the estimators states that the bilinear estimator, for example $\bar{\A}$, has off-diagonal entries that are equivalent to $\A_0$ multiplied by the diagonal entry $\beta$, by Theorem~\ref{thm:BLIN_bilinear_offdiag}. The same is true for the off-diagonal entries of $\bar{\B}$, which are equivalent to $\alpha \B_0$. The diagonal entries of the bilinear estimators $\bar{\A}$ and $\bar{\B}$ are not unique, but $\bar{a}_{ii} = \text{sign}(\alpha + \beta)\sqrt{|\alpha + \beta|}$ and $b_{jj} = |\alpha + \beta|$ satisfies the equivalence relation $\bar{a}_{ii} \bar{b}_{jj} = \tilde{a}_{ii} + \tilde{b}_{jj}$ for all $i \in \{1,2,\ldots, S \}$ and $j \in \{1,2,\ldots, L \}$. Thus, we have the blinear estimators
\begin{align}
\bar{\A} &= \text{sign}(\alpha + \beta)\sqrt{|\alpha + \beta|} \I_S + \beta \A_0, \hspace{.4in}    \tilde{\B} = |\alpha + \beta| \I_L + \alpha \B_0, \label{eq_AB_bilinear_equiv}
\end{align}

Now assume towards a contradiction that the estimated means of the BLIN and bilinear models are equal, that is,
\begin{align}
\tilde{\A}^T \X_t + \X_t \tilde{\B} = \bar{\A}^T \X_t \bar{\B}, \ \ \ t \in \{1,2,\ldots, T \}, \label{eq_BLIN_blinear_prop6_equiv}
\end{align}
Then, substituting the matrices in \eqref{eq_AB_blin_equiv} and \eqref{eq_AB_bilinear_equiv}, \eqref{eq_BLIN_blinear_prop6_equiv} implies that
\begin{align}
\beta {\A_0}^T \X_t + \alpha \X_t {\B_0} &= \beta {\A_0}^T \X_t + \alpha \X_t {\B_0} + {\A_0}^T \X_t {\B_0}, \\
\mathbf{0} &= {\A_0}^T \X_t {\B_0}, \ \ \ t \in \{1,2,\ldots, T \},
\end{align}
which is true for general $\X_t$ with probability zero unless both $\A_0$ and $\B_0$ are zero. 
\end{proof}

\clearpage
\section{Details of Simulation Studies}
\label{section_sim_details}
We provide details of the simulation studies performed to compare the bilinear and BLIN models. We first detail the cross-validation study. We then discuss a convergence study verifying Theorem~\ref{thm:BLIN_bilinear_offdiag}~and Proposition~\ref{prop:BLIN_bilinear_diag}.

\subsection{Cross-Validation Study}
We began by generating weighted networks of rank 1, that is,
\begin{align*}
&\A_0 = \mathbf{u} \mathbf{v}^T,\hspace{.4in} \B_0 = \mathbf{r} \mathbf{s}^T, 
\end{align*}
with each vector $\mathbf{u}$, $\mathbf{v}$, $\mathbf{r}$, and $\mathbf{s}$ of length 10 and consisting of independent and identically distributed standard normal random variables. 
Then, to remove self-loops in the networks, we set the diagonal entries of $\A_0$ and $\B_0$ to zero, thus arriving at true influence networks $\A$ and $\B$.
To generate values of $\A$ and $\B$ with fractions of zeros $q=0.5$ and $q = 0.9$, we simply set the absolute smallest 50\% and 90\% of entries in $\A$ and $\B$ to zero, respectively.

To control the signal-to-noise ratio of the models, we set the large-sample $R^2$ of each generative model to be about $0.75$. To do so, we scaled the true $\A$ and $\B$ by a multiplicative constant based on the generating model and sparsity level. We now derive the formulas for these constants for the general lag 1 autoregressive model in \eqref{eq_varsim}. By definition, the in-sample $R^2$ of any estimator of the mean of $\y_t$, denoted here $\{ \hat{\y}_t \}_{t=1}^T$, is
\begin{align}
R^2 =  \frac{\frac{1}{T}\sum_{t=1}^T || \y_t ||_2^2 - \frac{1}{T}\sum_{t=1}^T || \hat{\y}_t - \y_t ||_2^2}{\frac{1}{T}\sum_{t=1}^T || \y_t ||_2^2}. \label{eq_r2_def}
\end{align}
Now, we assume that $\bTheta$ is known, such that $\hat{\y}_t = \bTheta \y_{t-1}$ for all $t$. Then, for large samples and independent errors, the Law of Large Numbers implies that the sample averages in \eqref{eq_r2_def} converge (in probability) to their expectations. Thus, we have that, for large samples,
\begin{align}
R^2 &\approx \frac{\text{tr} \left( \bTheta \bTheta^T E[\y_{t-1} \y_{t-1}^T] \right)}{\text{tr} \left( \bTheta \bTheta^T E[\y_{t-1} \y_{t-1}^T] \right) + E[\e_t^T \e_t]} = \frac{\text{tr} \left( \bTheta \bTheta^T E[\y_{t-1} \y_{t-1}^T] \right)}{\text{tr} \left( \bTheta \bTheta^T E[\y_{t-1} \y_{t-1}^T] \right) + SL},  
\end{align}
with $S=L=10$ in this simulation. Further, it can be shown that for any stationary mean-zero VAR model of the form of \eqref{eq_varsim} that
\begin{align}
\text{vec} \left( E[\y_{t-1} \y_{t-1}^T] \right) = \left(\I_{S^2 L^2} - \bTheta \otimes \bTheta \right)^{-1} \text{vec} \left(\I_{SL} \right).
\end{align}
For the BLIN model, we scaled both $\A$ and $\B$ by a constant $k_{BLIN}$ and defined $\bTheta_{k,q} = k_{BLIN} [\I_L \otimes \A_q ; \B_q^T \otimes \I_S]$, where the subscript `$k$' emphasizes that $\bTheta_{k,q}$ depends on $k_{BLIN}$ and the subscript `$q$' signifies the $\A_q$ and $\B_q$ influence networks correspond to one of the sparsity levels $q\in\{0.0,0.5,0.9\}$.  We then defined 
\begin{align}
g(k_{BLIN}, q) &= \text{tr} \left( \bTheta_{k,q} \bTheta_{k,q}^T E[\y_{t-1} \y_{t-1}^T] \right) = \text{tr} \left( \bTheta_{k,q} \bTheta_{k,q}^T E[\y_{t-1} \y_{t-1}^T] \right) \\ 
&=  \text{vec} \left(\bTheta_{k,q} \bTheta_{k,q}^T \right)  \left(\I_{S^2 L^2} - \bTheta_{k,q} \otimes \bTheta_{k,q} \right)^{-1} \text{vec} \left(\I_{SL} \right).
\end{align}
 For each $q \in \{0.0, 0.5, 0.9 \}$, we selected $k_{BLIN}$ such that $ 0.75 \approx g(k_{BLIN}, q) / (g(k_{BLIN}, q) + SL)$, where the `$\approx$' simply indicates that $k_{BLIN}$ was selected by discrete search. 
We repeated the same procedure to select $k_{bilinear}$, although the $\bTheta_{k,q}$ in this case is $\bTheta_{k,q} = k_{bilinear}^2 \B_q^T \otimes \A_q$. 
Every generating model selected was stationary.

For each combination of $q \in \{0.0, 0.5, 0.9 \}$ and generating model, e.g. BLIN or bilinear, ($3 \times 2 = 6$ total simulation settings), we generated 100 data realizations. To ensure stationarity was reached, we simulated 100 time periods of each realization and used the final $T$ observations, for $T \in \{10, 20, 50\}$, as simulated data realizations. To compute out-of-sample $R^2$ for each simulated data set, we estimated the sparse, (reduced) rank 1, full BLIN models, and the bilinear model on each of 10 training data sets within a 10-fold cross-validation. This procedure is described in detail in Section~\ref{section:simulation}. We focused on $q = 0.9$ in the text in Section~\ref{section:simulation}, however, the out-of-sample results for $q \in \{0.0, 0.5 \}$ are given in Figure~\ref{fig:misspec_appx}. 

\spacingset{1}
\begin{figure}[ht!]
\centering
\begin{tabular}{r c c }
&\multicolumn{2}{c}{\textbf{Generating model}} \\
  &&\\
  & { BLIN } 
  \hspace{-.2in} & \hspace{-.2in} 
  Bilinear \\
\begin{sideways} \hspace{.55in} $\mathbf{q = 0.5}$ \end{sideways} &
 \includegraphics[width=.4\textwidth]{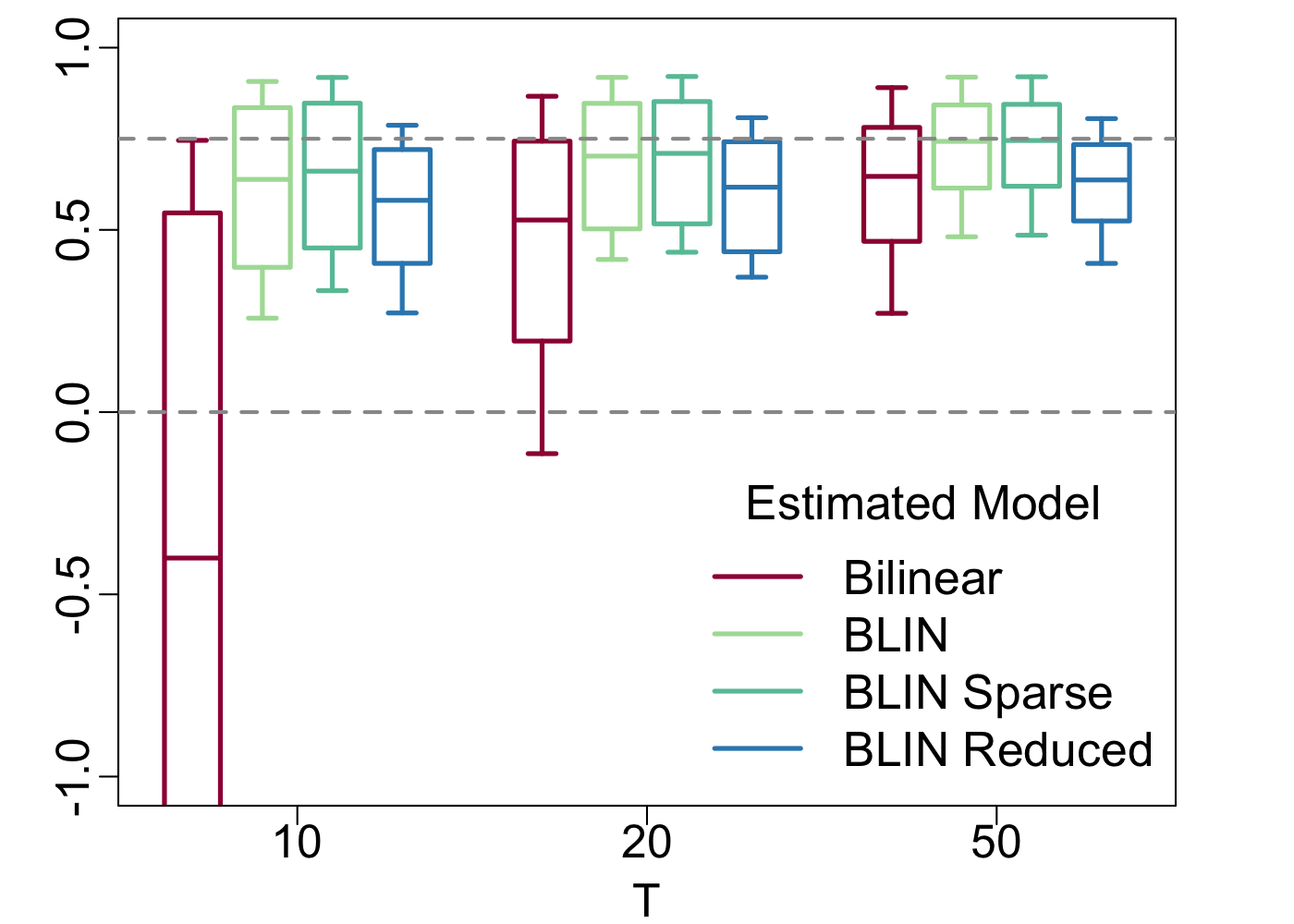}  
\hspace{-.2in} & \hspace{-.2in} 
 \includegraphics[width=.4\textwidth]{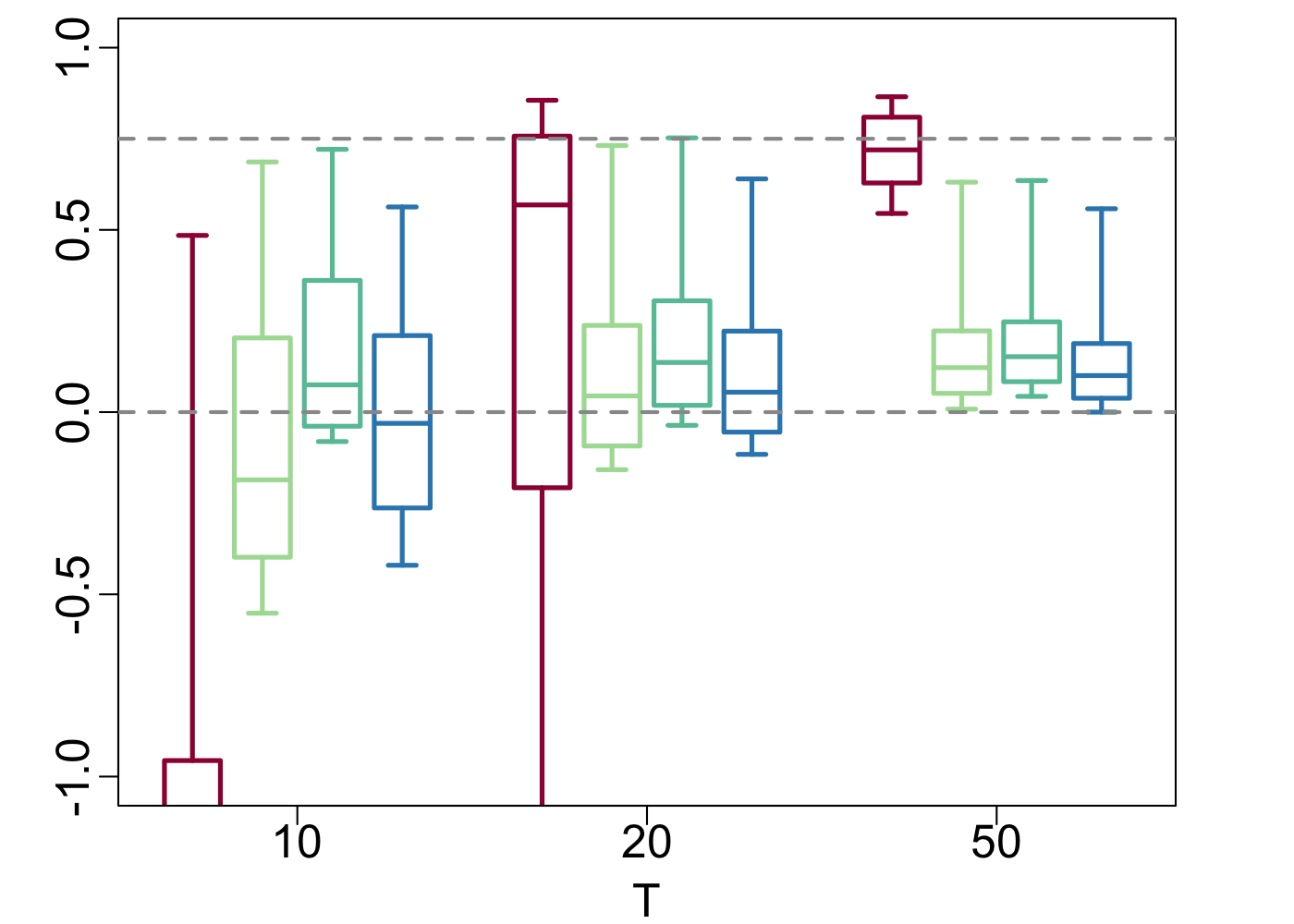} \\
 
 \begin{sideways} \hspace{.55in} $\mathbf{q = 0.0}$ \end{sideways} &
 \includegraphics[width=.4\textwidth]{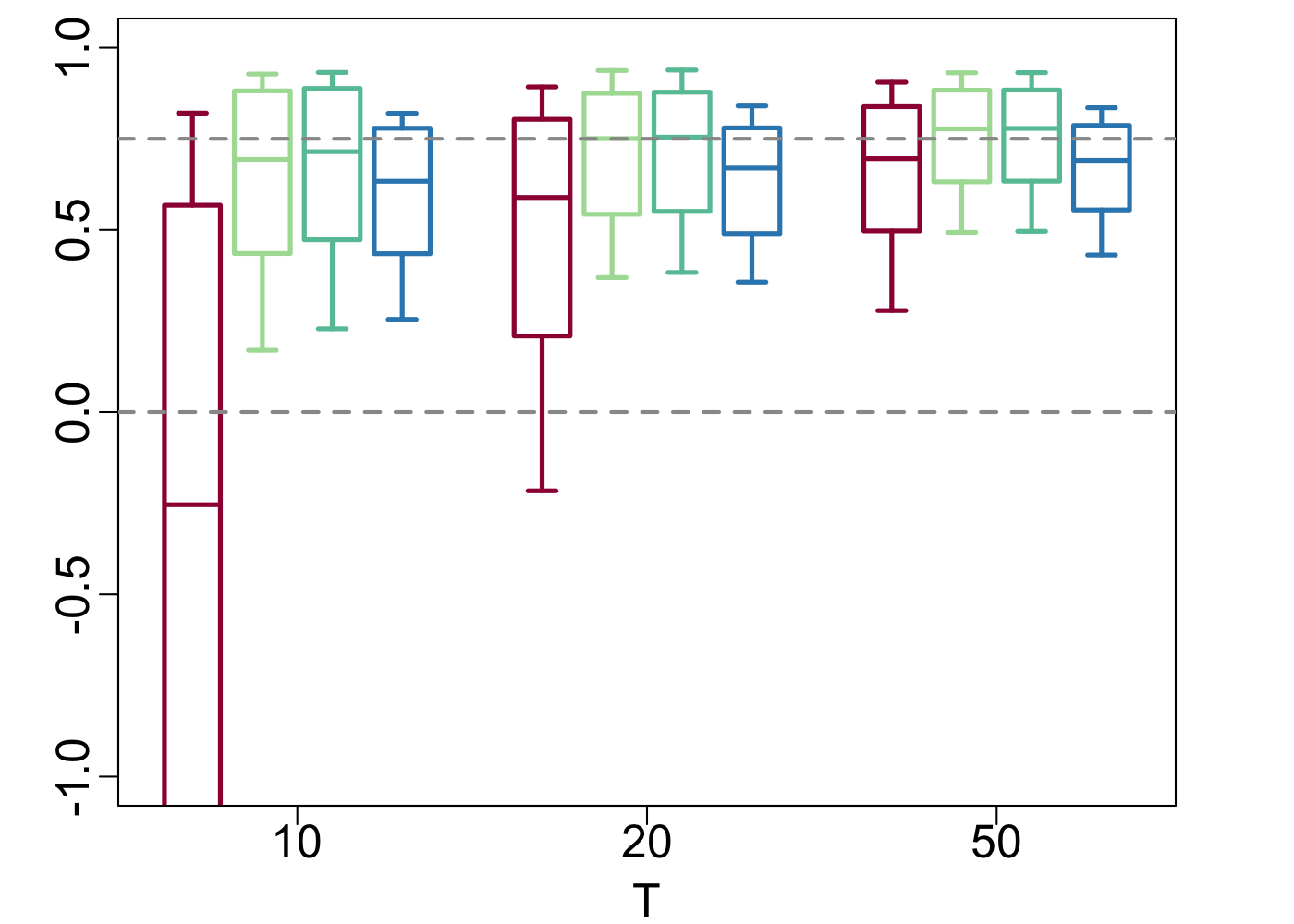}  
\hspace{-.2in} & \hspace{-.2in} 
 \includegraphics[width=.4\textwidth]{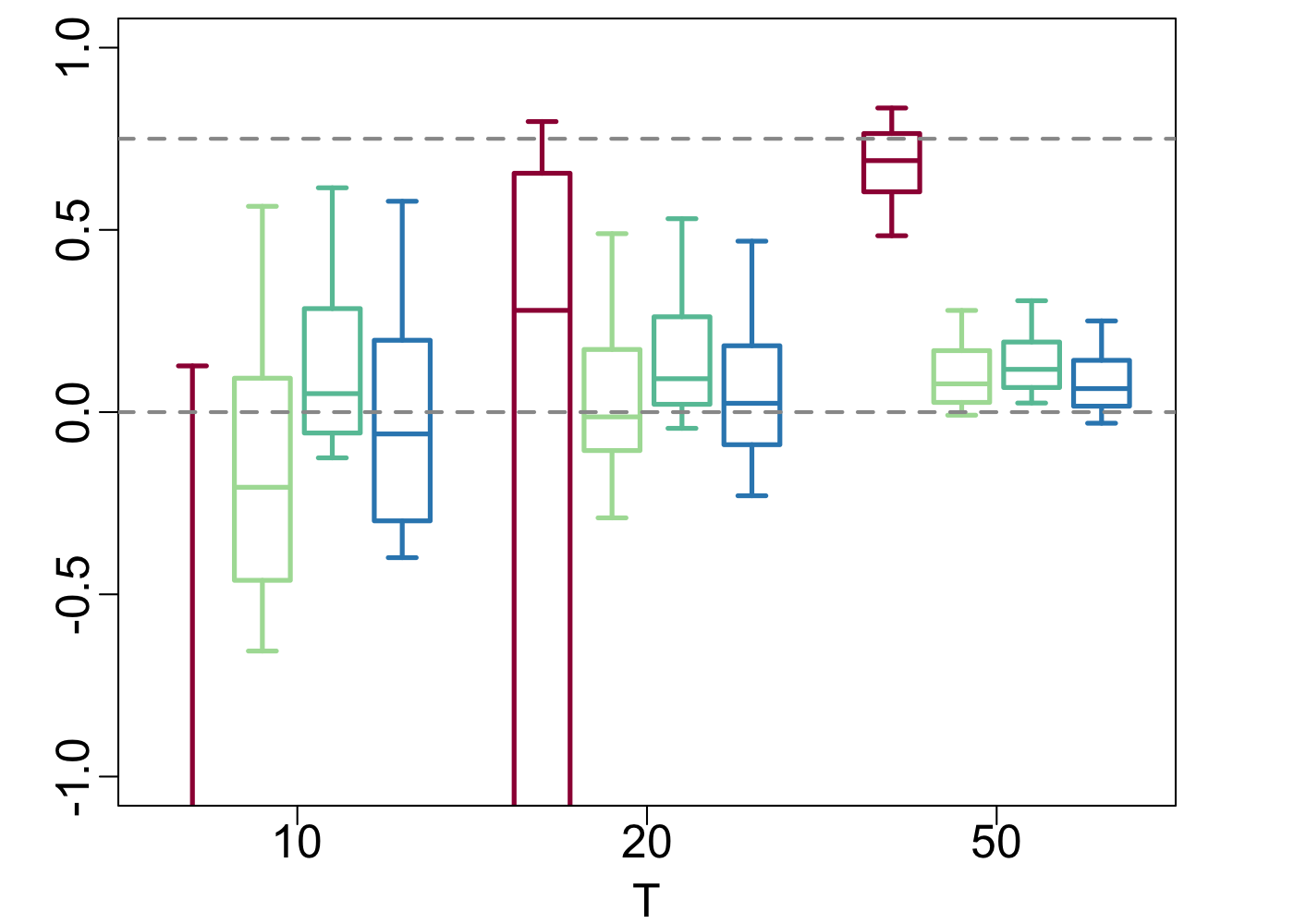} \\
  \end{tabular}
\caption{Out-of-sample $R^2$ values for each estimation procedure applied to 100 data realizations, with true coefficient matrices of sparsity $q \in \{0.0, 0.5 \}$. The left panel is results for data generated from the BLIN model and the right panel shows that for data generated from the bilinear model.
The centers of the boxplots represent the median $R^2$ value, the boxes represent the middle 80\% of $R^2$ values, and the whiskers correspond to the maximum and minimum $R^2$ values across 100 simulated data sets. Plots are truncated such that $R^2$ values less than $-1$ are not shown.  }
\label{fig:misspec_appx}
\end{figure}
\spacingset{1.5}

For completeness, we examined in-sample $R^2$ values for the generated data as well. In this case, we performed no cross-validation, but simply estimated each model on each complete data realization, i.e. for a given generating model, $q$, and $T$. The in-sample $R^2$ values are shown in Figure~\ref{fig:misspec_R2in}.

\spacingset{1}
\begin{figure}[ht!]
\centering
\begin{tabular}{r c c }
&\multicolumn{2}{c}{\textbf{Generating model}} \\
  &&\\
  & { BLIN } 
  \hspace{-.2in} & \hspace{-.2in} 
  Bilinear \\
\begin{sideways} \hspace{.55in} $\mathbf{q = 0.9}$ \end{sideways} &
 \includegraphics[width=.4\textwidth]{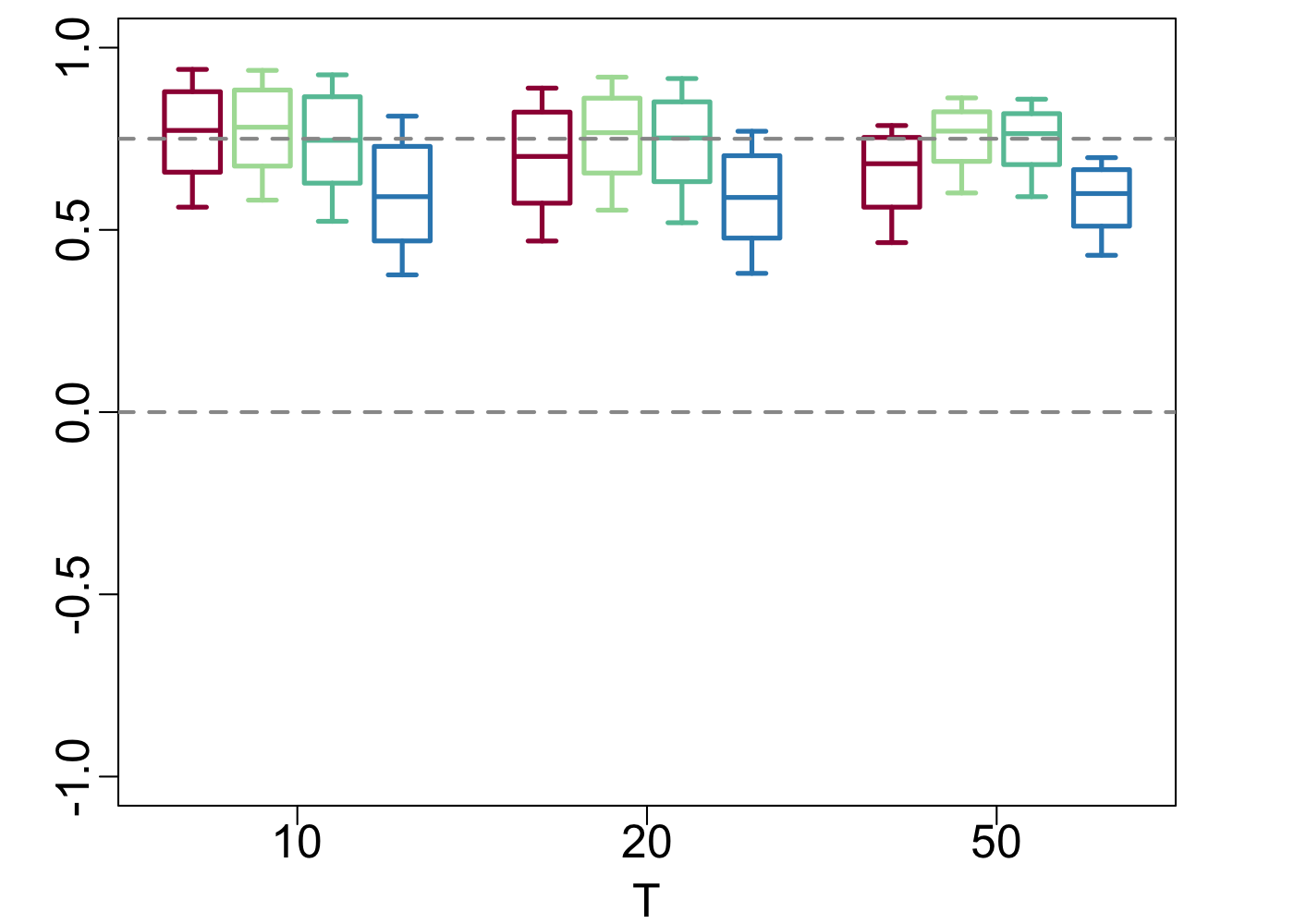}  
\hspace{-.2in} & \hspace{-.2in} 
 \includegraphics[width=.4\textwidth]{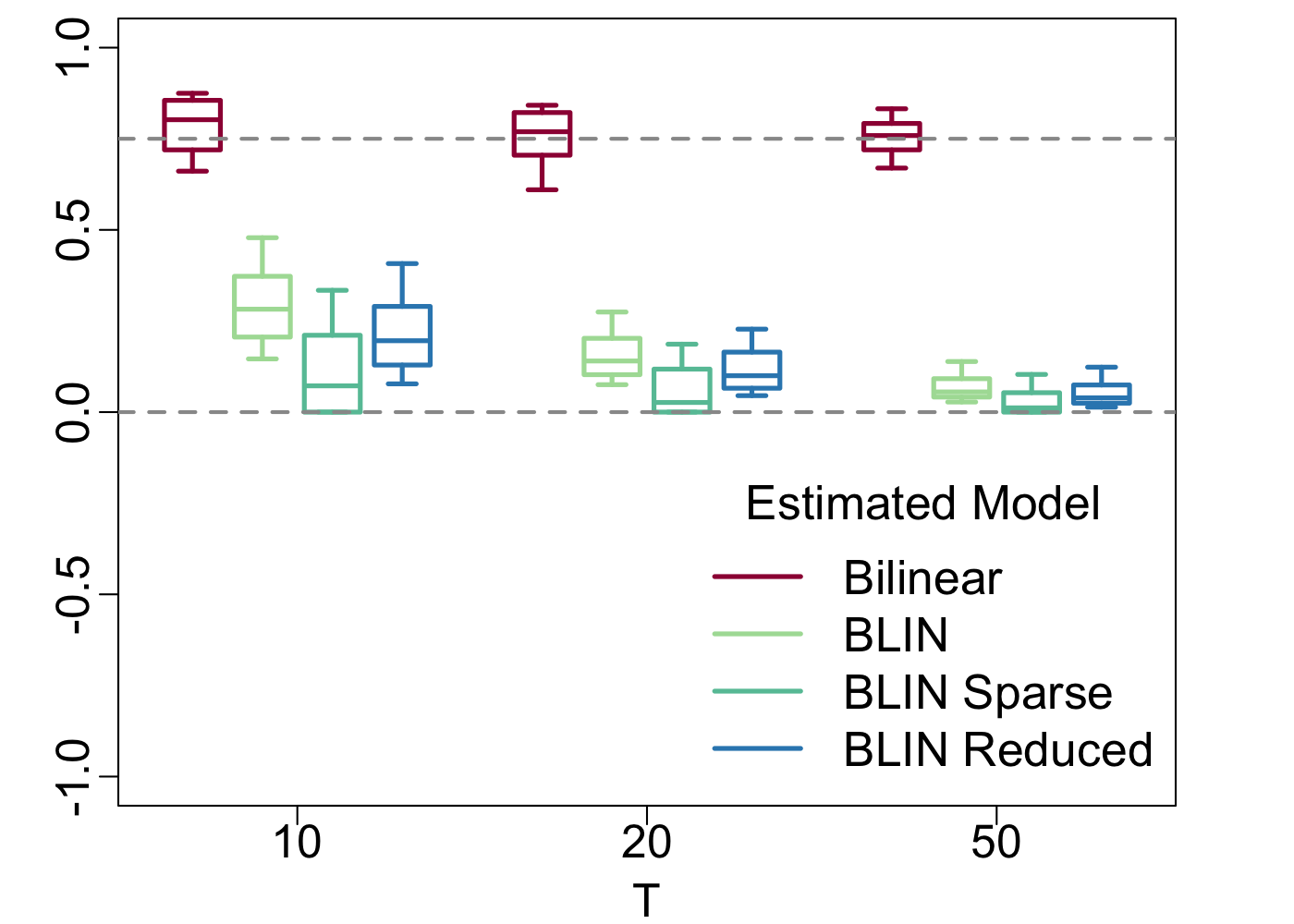} \\
 
\begin{sideways} \hspace{.55in} $\mathbf{q = 0.5}$ \end{sideways} &
 \includegraphics[width=.4\textwidth]{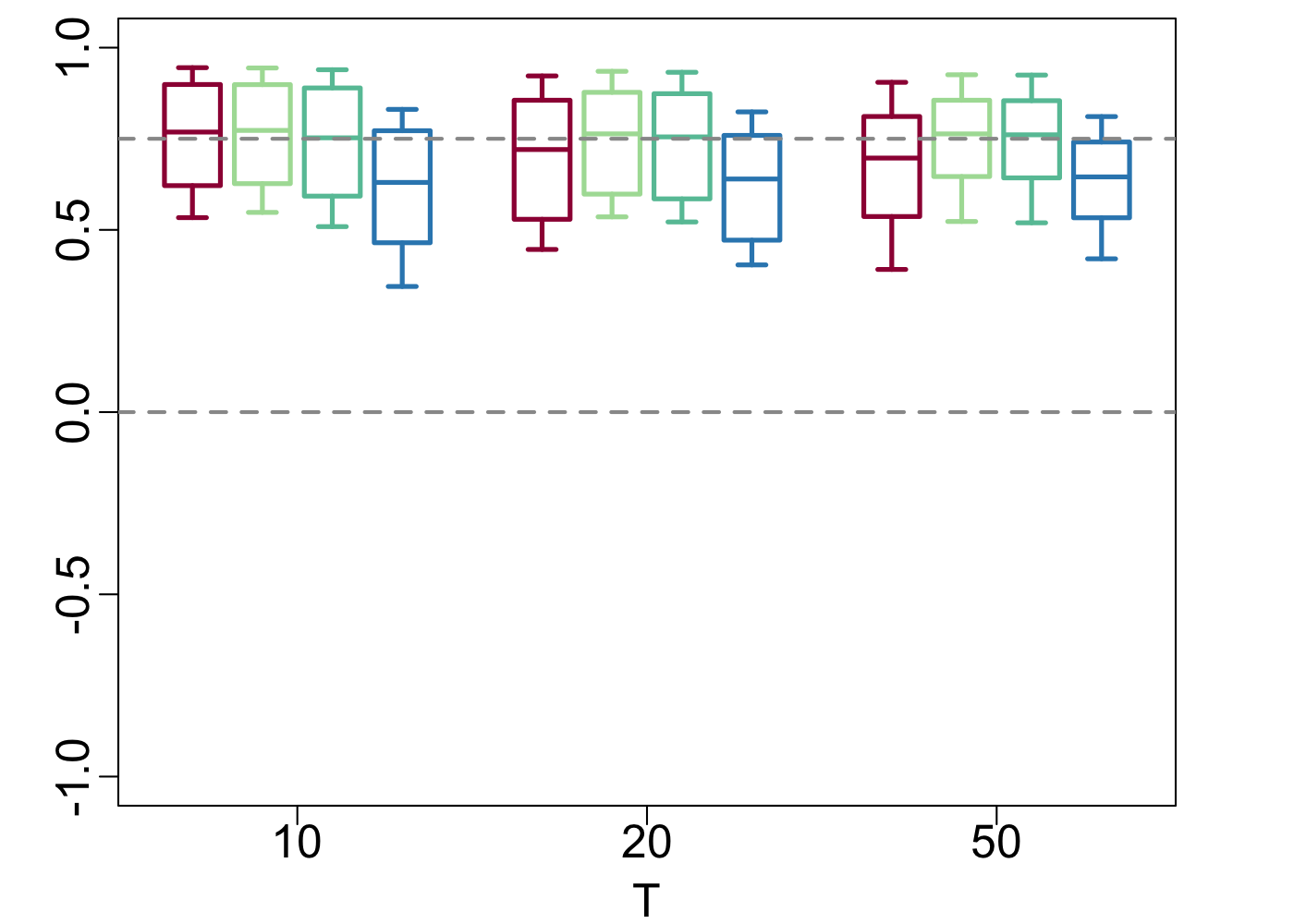}  
\hspace{-.2in} & \hspace{-.2in} 
 \includegraphics[width=.4\textwidth]{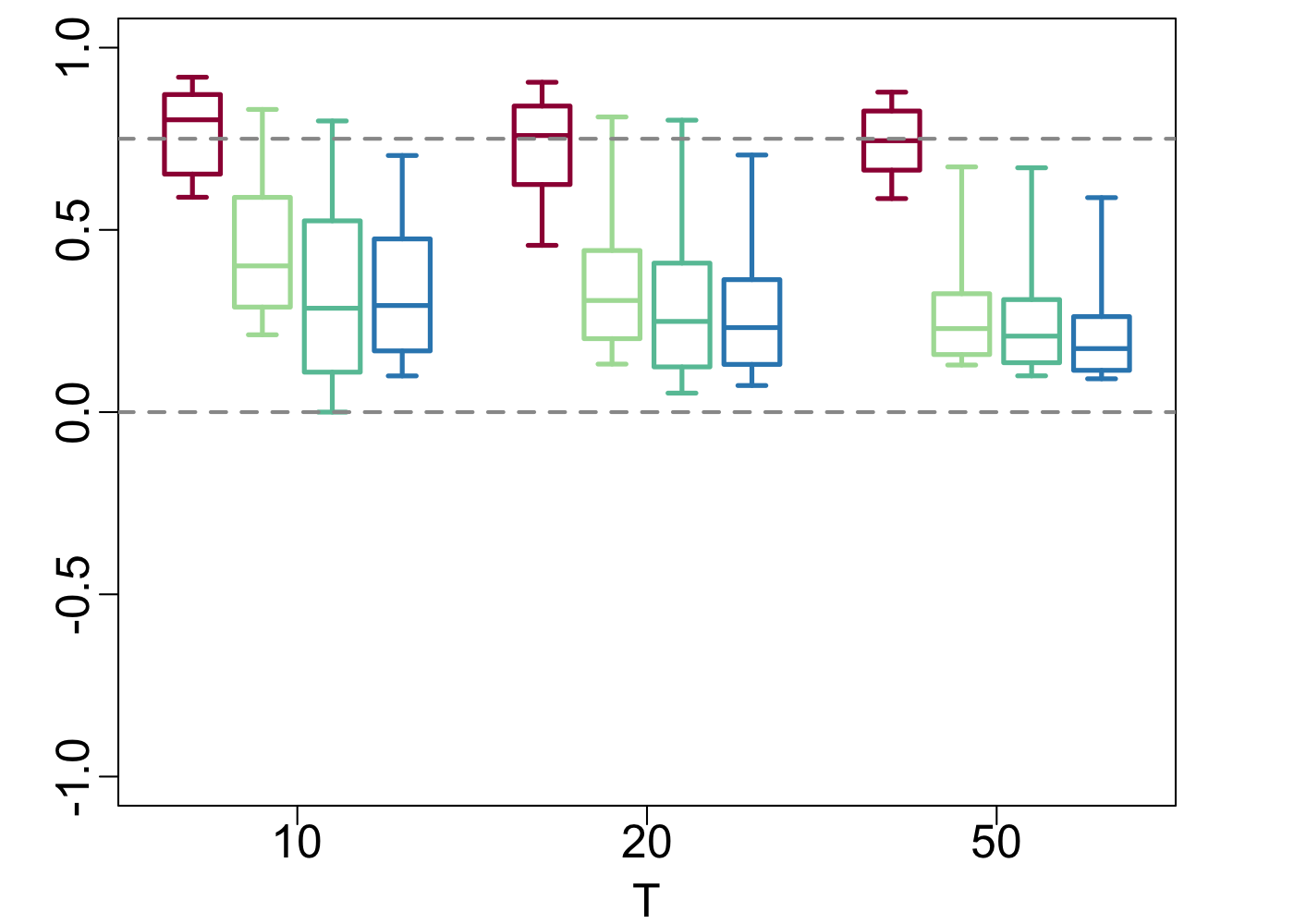} \\
 
 \begin{sideways} \hspace{.55in} $\mathbf{q = 0.0}$ \end{sideways} &
 \includegraphics[width=.4\textwidth]{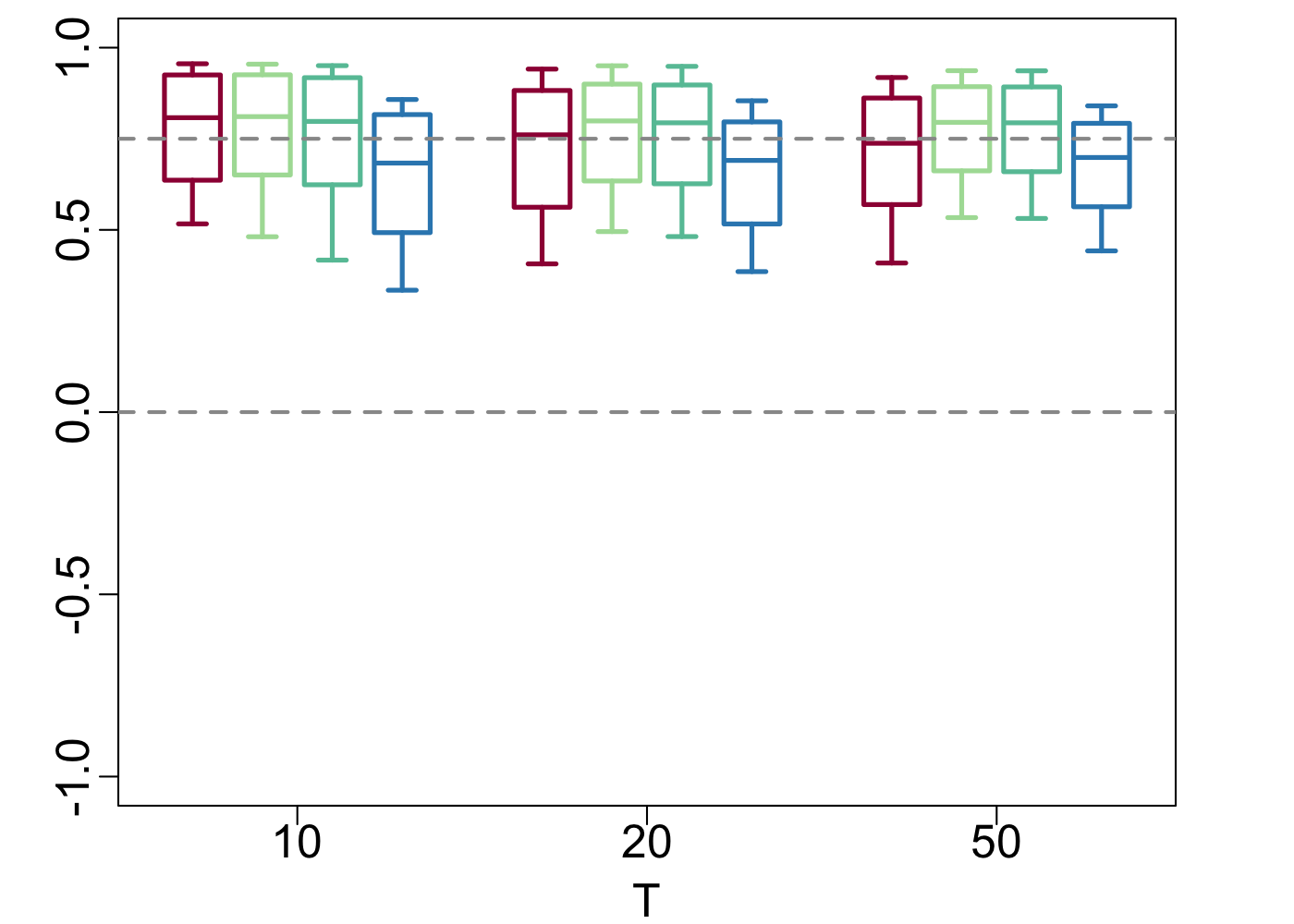}  
\hspace{-.2in} & \hspace{-.2in} 
 \includegraphics[width=.4\textwidth]{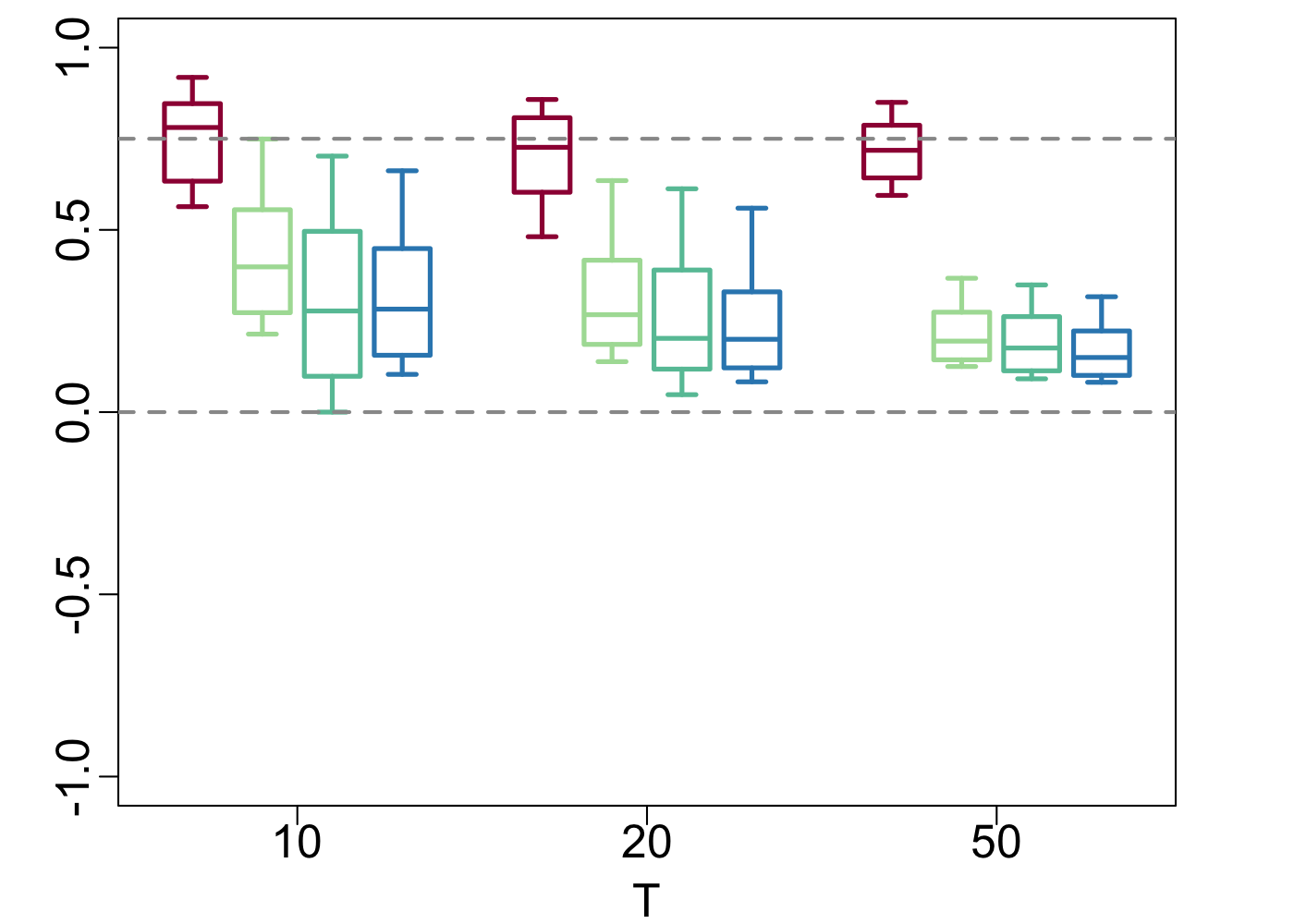} \\
  \end{tabular}
\caption{In-sample $R^2$ values for each estimation procedure applied to 100 data realizations, with true coefficient matrices of sparsity $q \in \{0.0, 0.5, 0.9 \}$. The left panel is results for data generated from the BLIN model and the right panel shows that for data generated from the bilinear model.
The centers of the boxplots represent the median $R^2$ value, the boxes represent the middle 80\% of $R^2$ values, and the whiskers correspond to the maximum and minimum $R^2$ values across 100 simulated data sets. }
\label{fig:misspec_R2in}
\end{figure}
\spacingset{1.5}

\subsection{Likelihood of bilinear model}

In this section, we conduct an investigation into the source of poor performance of the bilinear estimator when data is generated from this same model. To do so, we examine in-sample and out-of-sample $R^2$ values at various coefficient values. These coefficient values are weighted averages of true and estimated coefficients.
Recall that, for normally-distributed data, large $R^2$ values correspond to large likelihood values and vice versa, so we use the terms, e.g., ``high $R^2$'' and ``large likelihood'' interchangeably for this discussion.

Recall that the bilinear model was estimated in a 10-fold cross-validation (see Section~\ref{section:simulation}), and thus, the bilinear estimator was computed on a subset of the data containing approximately 90\% of the full simulated data set. 
For a particular subset of a particular simulated data set, we computed in-sample and out-of-sample $R^2$ values  at a mixture of the true coefficients and the coefficients estimated on the particular subset. That is, we computed $R^2$ as a function of $\btheta$, where $\btheta$ is a vector of the entries in $\A$ and $\B$ matrices defined by the mixture
\begin{align}
\btheta^T = (1-\xi)[\text{vec}(\A), \text{vec}(\B)] + \xi[\text{vec}(\hat{\A}), \text{vec}(\hat{\B})], \label{eq_coef_mixture}
\end{align}
where $\{\A, \B \}$ are the true coefficients and $\{ \hat{\A}, \hat{\B} \}$ are the coefficients estimated on the 90\% data subset. Thus, we examined the likelihood surface along a line in the coefficient space connecting the true and estimated coefficients. This likelihood surface may show multiple local optima, should they exist. 

In Figure~\ref{fig:fig_bilinear_misspec}, we show the $R^2$ values when generating from the bilinear model and estimating the bilinear model. We see that, when $\xi=1$, the value of in-sample $R^2$ is high, as it must be, since $\btheta$ is the MLE when $\xi=1$. However, we also see that in-sample $R^2$ is high near $\xi=0$, i.e. near the true coefficients. Between $\xi=0$ and $\xi=1$, the in-sample $R^2$ drops substantially. This is evidence of multiple modes, or multiple local optima, in the likelihood of the bilinear model. Of course, we observe that the out-of-sample $R^2$ near $\xi=0$ is also high, that is, selecting coefficient values near the true coefficients generalize well in prediction out-of-sample. However, near $\xi=1$, even though the in-sample $R^2$ mode is higher than the one near $\xi=0$, the  out-of-sample $R^2$ value is low. Thus, choosing the mode with highest in-sample $R^2$, as is done in maximum-likelihood estimation procedures, does not necessarily generalize well out-of-sample, as shown in this case. This fact, taken with the poor representation of $\A$ by the $\hat{\A}$ matrix in Figure~\ref{fig:fig_bilinear_misspec} (we observe similar issues for $\B$ and $\hat{\B}$), suggests that $\xi=1$ is far from $\xi=0$ in coefficient space. That is, the likelihood of the bilinear model (at least for $T = 10$ and $q=0.50$) has multiple modes that are not near one another. Finding all of the modes is an extremely difficult -- if not impossible -- task, and selecting an incorrect mode may result in poor estimator performance. We find similar patterns to those in Figure~\ref{fig:fig_bilinear_misspec} for $T=20$ and $q \in \{0.0, 0.9 \}$ when generating from the bilinear model and estimating the bilinear model.

A similar analysis for the BLIN model confirms the unimodality implied by Proposition~\ref{prop_rankXb} (see Figure~\ref{fig:fig_blin_misspec}). Since the in-sample $R^2$ value does not dip between $\xi=0$ and $\xi=1$, the estimator (at $\xi=1$) resides in a mode that contains the true coefficient values. The fact that the out-of-sample $R^2$ is flat between $\xi=0$ and $\xi=1$ suggests that any value of $\xi$ in this range will generalize well out-of-sample, that is, any of these estimators will perform well. As $\hat{\A}$ looks remarkably similar to $\A$ in Figure~\ref{fig:fig_blin_misspec} ($\hat{\B}$ also looks similar to $\B$), taken with the flat out-of-sample likelihood, we can be confident the estimator $\hat{\btheta}$ and the true value $\btheta$ are close in the coefficient space.

\spacingset{1}
\begin{figure}[ht!]
\centering
 \includegraphics[width=.9\textwidth]{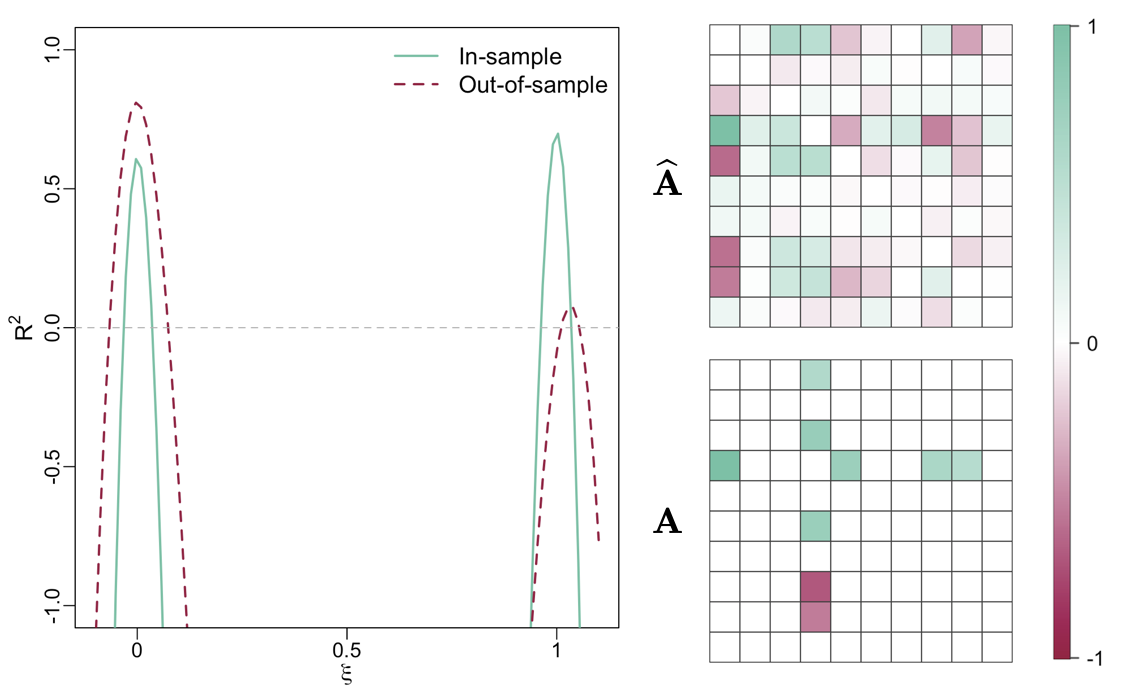} 

\caption{$R^2$ values for a mixture of true coefficients and coefficients estimated for the bilinear model (see \eqref{eq_coef_mixture}) when generating from the bilinear model with $q=0.9$ and $T = 10$, for a single fold of the cross-validation. The estimated $\A$ matrix (top right) is scaled by a constant $c$ (and $\B$ by $c^{-1}$) such that the estimated $\A$ and $\B$ are as close as possible in $L_2$ norm to the true $\A$ and $\B$ matrices. The pair of matrices $\A$ and $\hat{\A}$ are normalized to lie in $[-1,1]$ for sake of plotting.
}
\label{fig:fig_bilinear_misspec}
\end{figure}
\spacingset{1.5}

\spacingset{1}
\begin{figure}[ht!]
\centering
 \includegraphics[width=.9\textwidth]{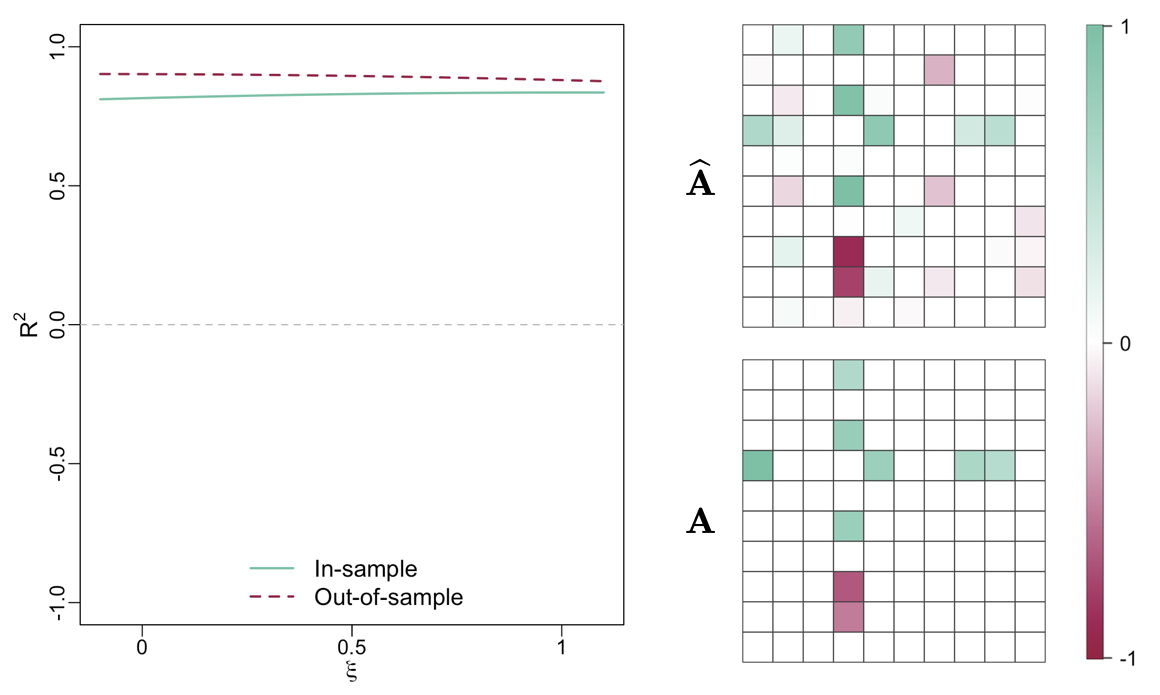} 

\caption{$R^2$ values for a mixture of true coefficients and coefficients estimated for the bilinear model (see \eqref{eq_coef_mixture}) when generating from the bilinear model with $q=0.9$ and $T = 10$, for a single fold of the cross-validation. The estimated $\A$ matrix (top right) is scaled by a constant $c$ (and $\B$ by $c^{-1}$) such that the estimated $\A$ and $\B$ are as close as possible in $L_2$ norm to the true $\A$ and $\B$ matrices. The pair of matrices $\A$ and $\hat{\A}$ are normalized to lie in $[-1,1]$ for sake of plotting.}
\label{fig:fig_blin_misspec}
\end{figure}
\spacingset{1.5}

\subsection{Convergence Study}
\label{section_sim_conv}
We conduct a simulation study (to larger dimension $T$ up to $10^5$) to examine the theoretical results of Section~\ref{section_misspec}. In line with Section~\ref{section_misspec}, we generate from and estimated the bilinear and full BLIN models, using the BLIN representation in \eqref{blin_Xt} with $\X_t = \Z_t$ for all $t \in \{1,2,\ldots, T \}$. For simplicity, we generate every element in $\X_t$ and $\E_t$ with independent and identically distributed (iid) standard normal random variables.
We fix  $\A$ and $\B$ throughout the study, with $\A^T = \U \V^T$ and $\B = \bR \S^T$, where $\{\U, \V\}$ are $S \times S$ matrices and $\{\bR, \S \}$ are $L \times L$ matrices, all with iid standard normal entries. Although these $\A$ and $\B$ matrices resemble those in the reduced rank BLIN model, by construction, they are full rank.
The number of time periods $T$ grows from $10^2$ to $10^5$ in $10^{1/2}$ multiplicative increments. We choose $S=10$, $L=9$, and repeat for 1,000 replications at each combination of $T$ and generating model. 
We compare the estimated coefficients to the truth by normalizing the off-diagonal entries such that, if the off-diagonals of the estimate $\hat{\A} = c\A$ for any $c \in \R$ from any estimation procedure, then we say that the off-diagonal estimates of $\A$ have no error. To do this, we compute the mean squared error between the off-diagonal entries in $\hat{\A}$ divided by their sum and the off-diagonal entries in $\A$ divided by their sum. We do the same for $\B$. Finally, we calculate the mean squared error of the autoregressive parameters as appropriate for the generating and estimating model. For instance, when generating from the bilinear model but estimating the BLIN model, the $(i,j)$ contribution to the mean squared error is  $( \hat{a}_{ii} + \hat{b}_{jj}- a_{ii}b_{jj} )^2$, and the mean square error of the diagonals is the sum over all $i$ and $j$. Lastly, we estimate the convergence rate of each estimator to the truth by fitting a simple linear regression model to ${\rm log}_{10}(MSE)$ against ${ \rm log}_{10}(T)$ for each coefficient and combination of generating and estimating models.
\spacingset{1}
\begin{figure}[ht]
\centering
\begin{tabular}{r c c c }
\begin{sideways} \hspace{.65in}{\small \textbf{Bilinear} } \end{sideways}&
 \includegraphics[width=.3\textwidth]{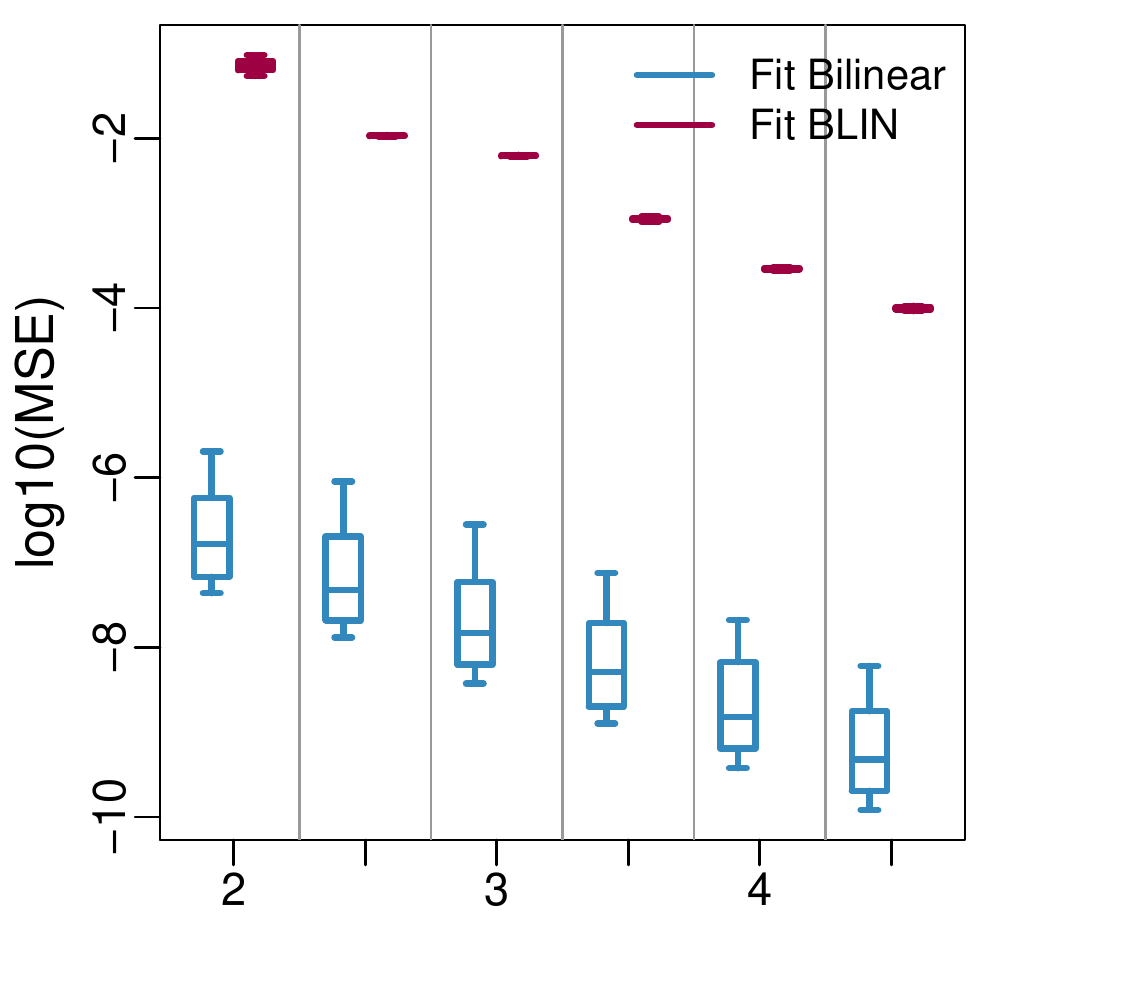}  &
\includegraphics[width=.3\textwidth]{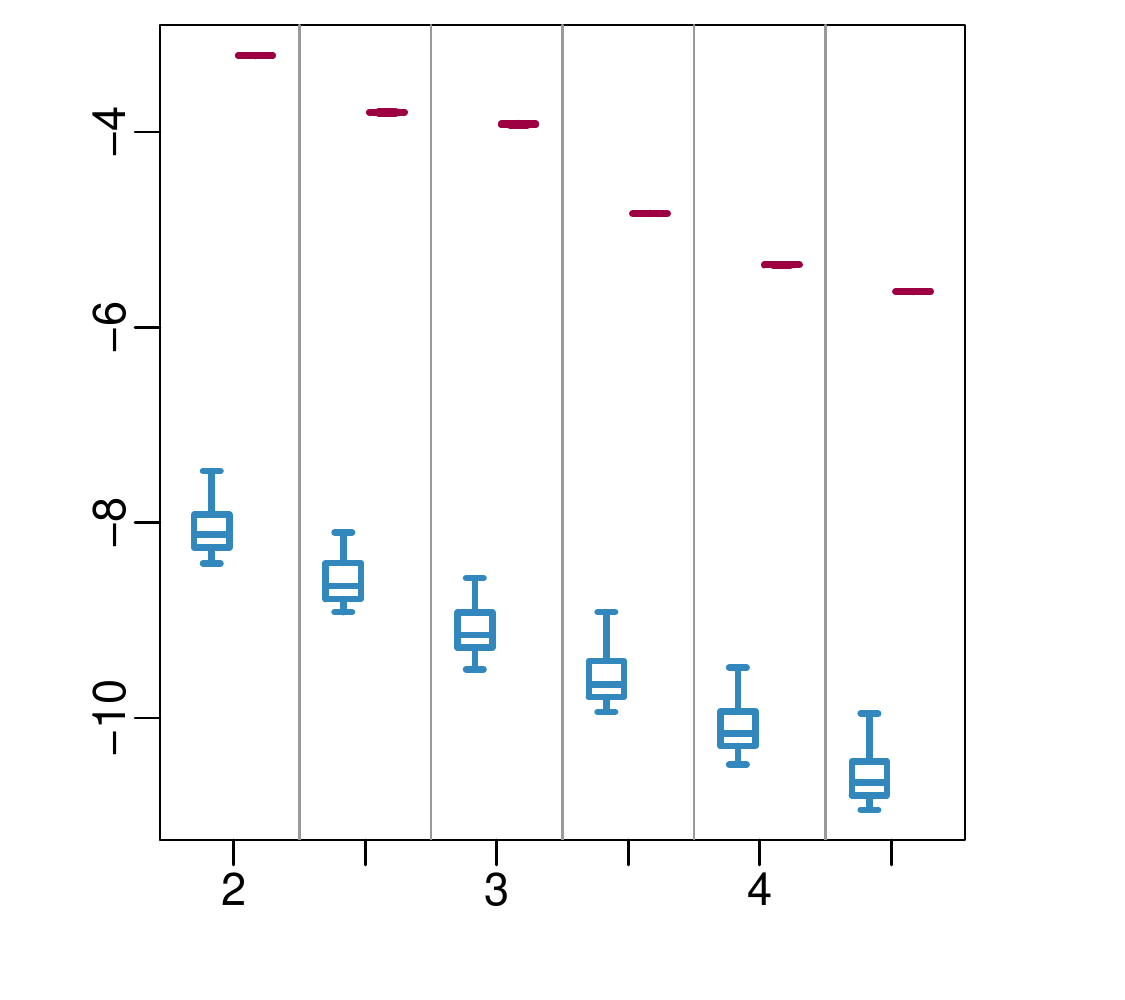}  &
  \includegraphics[width=.3\textwidth]{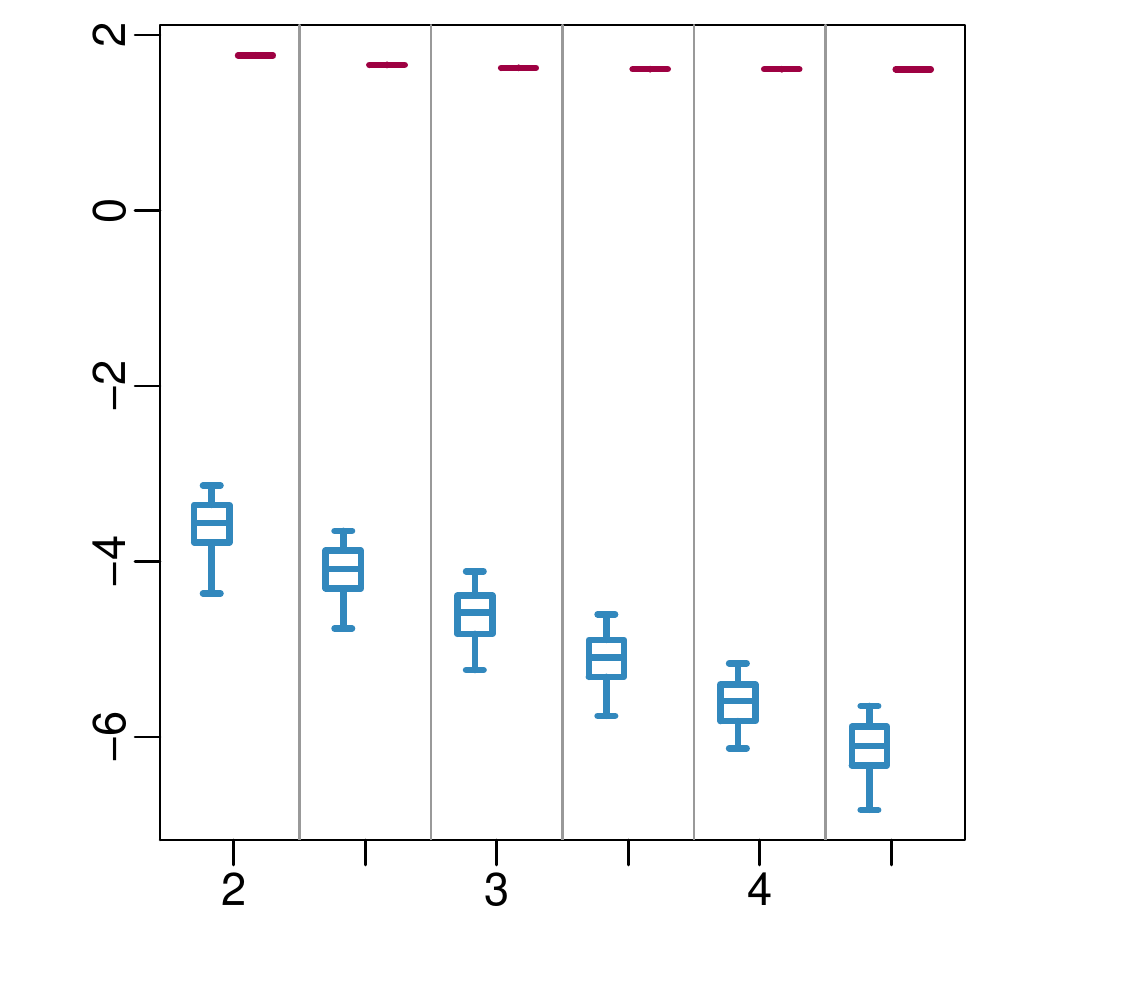}  \\
  \begin{sideways} \hspace{.7in}{ \small \textbf{BLIN} } \end{sideways}&
   \includegraphics[width=.3\textwidth]{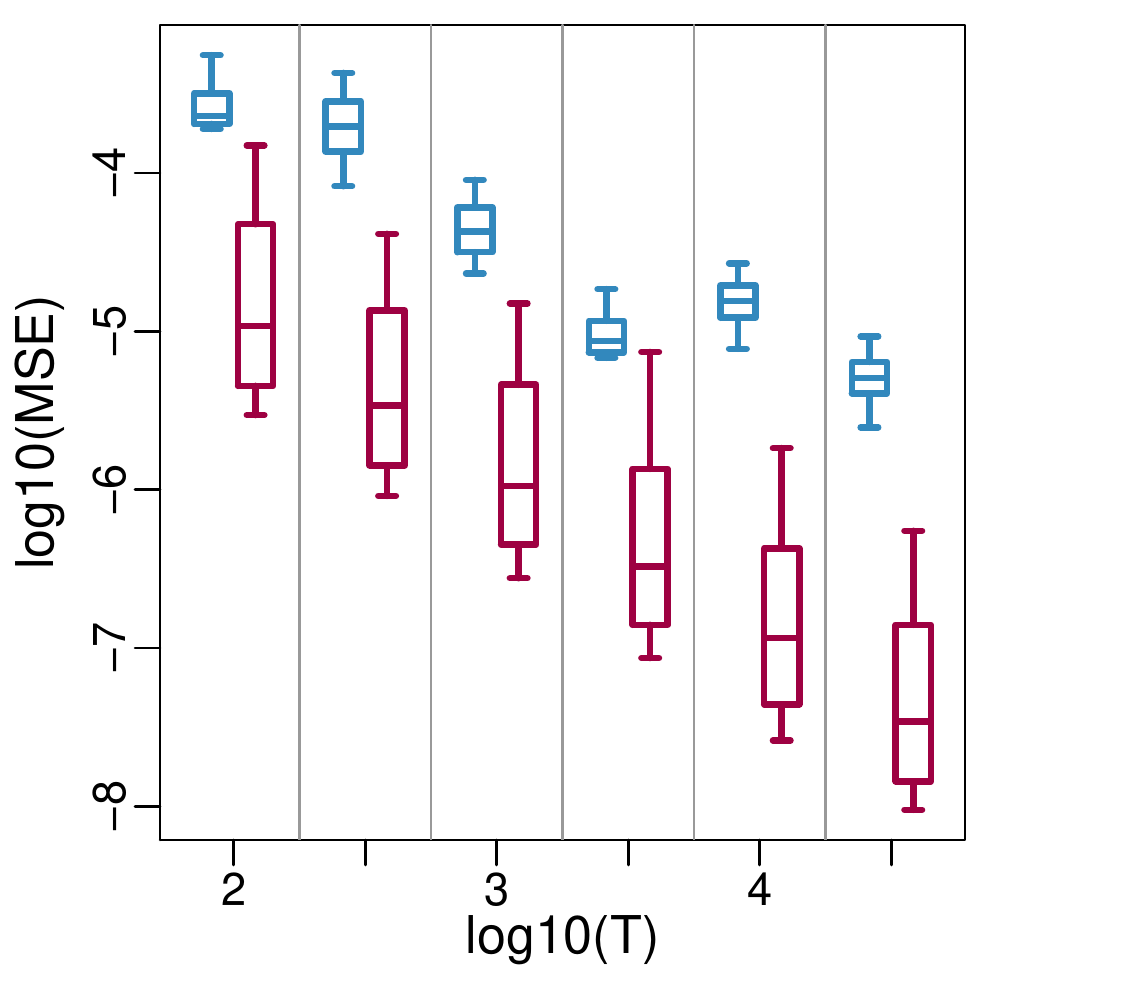}  &
\includegraphics[width=.3\textwidth]{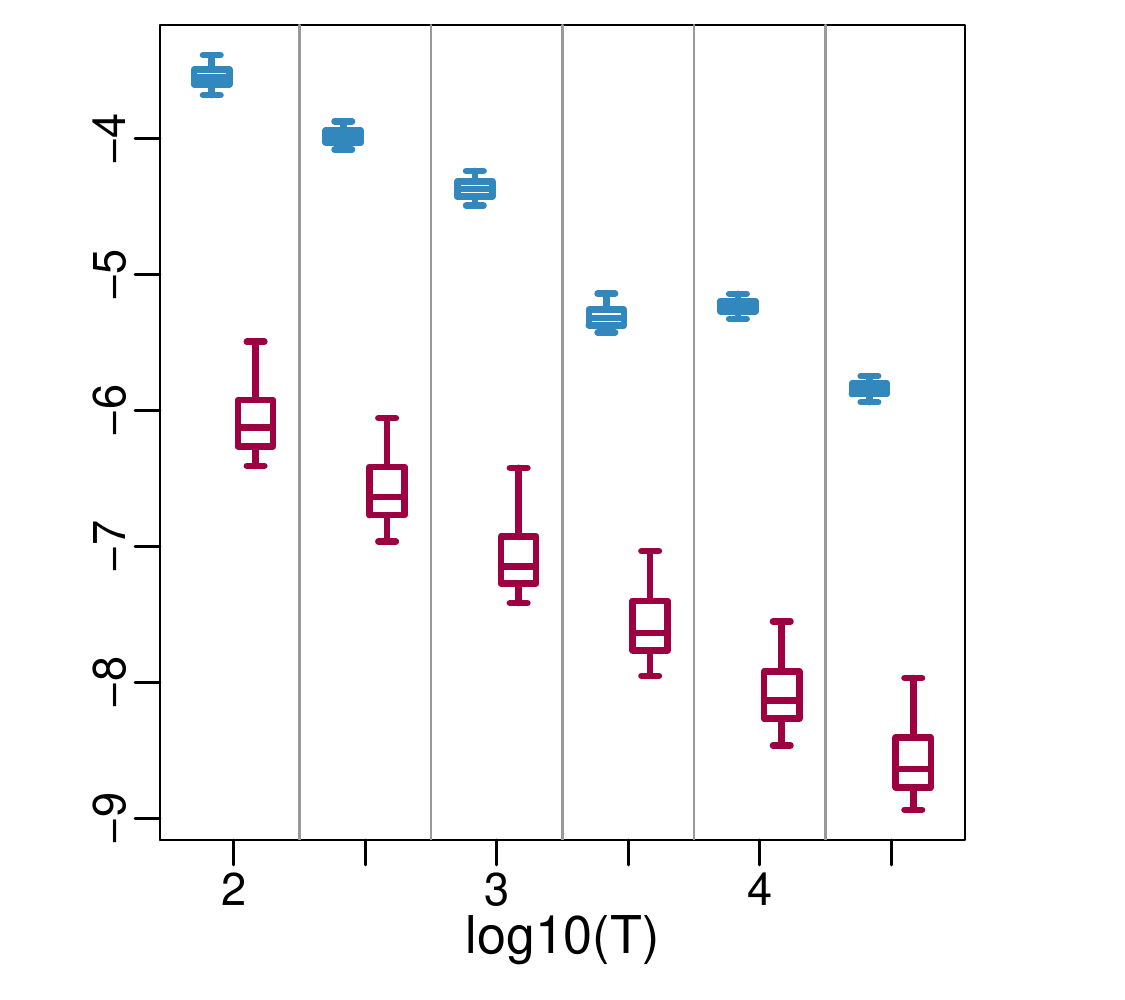}  &
  \includegraphics[width=.3\textwidth]{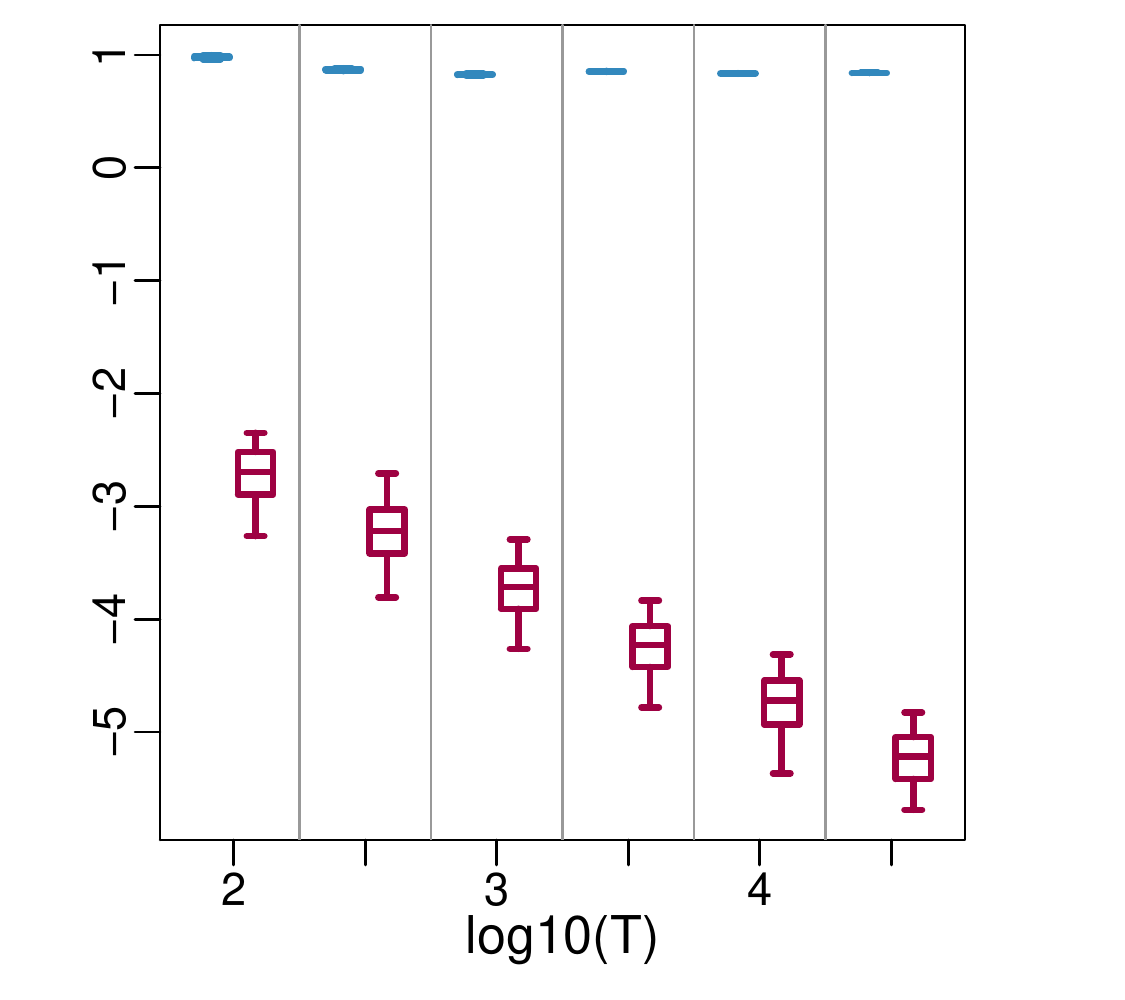}  \\
  &
   {\small \textbf{$\hat{\A}$, off-diag} } & 
  {\small \textbf{$\hat{\B}$, off-diag} } &
  {\small \textbf{diagonals} } \\
  \end{tabular}
\caption{Mean square error between  coefficient estimates and true coefficients used to generate data from the bilinear model (top row) and BLIN model (bottom row). The estimators are the least squares estimators of the BLIN and bilinear models. }
\label{fig:misspec_study2}
\end{figure}
\spacingset{1.5}

In Figure~\ref{fig:misspec_study2}, both the BLIN and bilinear estimators of the off-diagonal estimator of $\A$ and $\B$ have MSEs that tend to zero regardless of the generating model. This suggests that both estimators are consistent under both generating models, confirming Theorem~\ref{thm:BLIN_bilinear_offdiag}. In addition, we observe that neither diagonal estimator appears to converge when the estimation is not of the true model, e.g. the BLIN diagonals do not converge to the true bilinear diagonals. In this simulation, we have that $\bPsi \propto \I_S$ and $\bOmega \propto \I_L$, however, the diagonals of $\A$ and $\B$ are non-constant, verifying that this latter condition is critical in Proposition~\ref{prop:BLIN_bilinear_diag}. 

In Table~\ref{tab_conv_rates}, we 
observe $\sqrt{T}$ convergence for all off-diagonal elements estimated by the BLIN model (when data is generated from either the BLIN or bilinear models), which agrees with the discussion in Section~\ref{section_misspec}. Although the bilinear model attains the same $\sqrt{T}$ convergence rate when correctly specified, the convergence rate is significantly slower when the true generating model is the BLIN model. This suggests that the least squares estimator of the BLIN model is more efficient when misspecified than that of the bilinear model.  
\spacingset{1}
\begin{table}[ht]
\centering
\caption{Estimates of decrease in (base 10 logarithm of) mean squared error of least squares estimates of BLIN and bilinear models, under simulated data from each model, when $T$ increases by a factor of 10. Bolded and starred slopes are those which are significantly \emph{not} -1 $(p < 0.05)$, where the $p$-value is computed from a simple linear regression of $\log_{10}MSE$ on $\log_{10}T$.\vspace{.1in}}
\begin{tabular}{r | l l} 
\textbf{ Estimator } &
\textbf{Bilinear} &
\textbf{BLIN} \\
\hline
BLIN $\A$ &  -1.13 & -1.00 \\
BLIN $\B$ &  -1.10 & -1.00 \\
BLIN diag. & \textbf{-0.05$^*$} & -1.00 \\
bilinear $\A$ &  -1.00 & \textbf{-0.71$^*$} \\
bilinear $\B$ &  -1.01 & \textbf{-0.92$^*$} \\
bilinear diag. &  -1.01 & \textbf{-0.04$^*$} 
\end{tabular}
\label{tab_conv_rates}
\end{table}
\spacingset{1.5}

\section{Multipartite Relational Data}
\label{sec_multi}
In this paper, we primarily discuss the setting where $\Y_t$ is a matrix, equivalently a 2-mode array, with replications over time, such that $\{ \Y_t \}_{t=1}^T$  may be considered a tensor, or a 3-mode array. 
The $S \times L \times T$ array $\Y$ may be constructed by concatenating the $T$ matrices $\{\Y_t \}_{t=1}^T$, each of dimension $S \times L$. That is, $\Y = [\Y_1 ; \Y_2 ; \ldots ;\Y_T]$, where we let `$;$' denote this concatenation operation. 
The BLIN model is easily extendable to model $K$-mode arrays.
A $K$-mode BLIN model is appropriate when there are more than two actor types in the relational dataset, which may be termed \emph{multipartite} as opposed to bipartite. In the ICEWS data example, we consider interaction type, in addition to source country and target country, as the third mode. In this example, each relation is one of four interaction types from a source state to a target state. 
Then, each $\Y_t$ is a 3-mode array with entries $y_{ijk}^t$ measuring relation intensity, where $i$ is the source state, $j$ is the target state, $k$ is the relation type, and $t$ is the week of the observation. Please see Section~\ref{section_temporal_SID} for more detail.

We are interested in making inference on influence networks $\{\B_1,...,\B_K\}$, where each influence network is associated with a specific mode of the array.  For the ICEWS data example, our motivation is inference on influence networks of source countries $(\B_1)$, target countries $(\B_2)$, and interaction types $(\B_3)$. As discussed in Section~\ref{sec_model}, the $(i,j)$ component of $\B_k$ reflects the influence of the $i^{\text{th}}$ ``slice'' of $\X_t$ on the $j^{\text{th}}$ slice of $\Y_t$, where the slice is along mode $k$. For example, the $(i,j)$ entry of $\B_3$ estimates the influence of $\x_{\cdot \cdot i}$ on $\y_{\cdot \cdot j}$, when controlling for row and column dependencies between $\Y_t$ and $\X_t$. In the example above, this entry characterizes the influence of interaction type $i$ on interaction type $j$ when controlling for source country and target country influences.

We lean on the Tucker product (\cite{tucker1964extension}) and related results to write the array BLIN model for general $K$-mode arrays \citep[see also][]{de2000multilinear,kolda2009tensor,hoff2011separable}. First, we rewrite the expectation of the BLIN model of~\eqref{blin_Xt} using the Tucker product notation as
\begin{align}
E[\Y \, | \, \X, \Z] &= \X \times \{\A^T, \I_L, \I_T \} + \Z \times \{\I_S, \B^T, \I_T \}, \label{eqbitenTucker}
\end{align}
where the $S \times L \times T$ arrays $\X$ and $\Z$ are the results of concatenating $\{\X_t \}_{t=1}^T$ and $\{\Z_t \}_{t=1}^T$, respectively, as $\{\Y_t \}_{t=1}^T$ is concatenated to form $\Y$.

We now extend \eqref{eqbitenTucker} to general $K$-mode arrays, allowing a different predictor for each mode (i.e. $\X$ and $\Z$ for the 2-mode approach in \eqref{eqbitenTucker}).
Let $\Y_t$ be an $m_1 \times \ldots \times m_K$ array observed over $t \in \{1,2,\ldots,T \}$ time periods, with $K$ corresponding predictor arrays $\{ \X_t^{(k)} \}_{k=1}^K$ each of dimension $m_1 \times \ldots \times m_K$ as well. To form the full model, we concatenate the arrays, as above in \eqref{eqbitenTucker}, such
that $\Y$ is of dimension $m_1 \times m_2 \times \ldots \times m_K \times T$, as is every one of the $K$ predictor arrays $\{ \X^{(k)} \}_{k=1}^K$. Then, the expectation of the multipartite extension of the BLIN model is as follows:
\begin{align}
E \left[\Y \, | \, \{ \X^{(k)} \}_{k=1}^K \right] &= \sum_{k = 1}^K \X^{(k)} \times \{\I_{m_1}, \I_{m_2}, \ldots, \I_{m_{k-1}}, \B_k^T, \I_{m_{k+1}}, \ldots, \I_{m_K}, \I_R  \} \label{eqbitenK} \\
E\left[\y \, | \, \{ \x^{(k)} \}_{k=1}^K \right] &=  \sum_{k = 1}^K  \Bigg( \I_{
m_{k+1}m_{k+2}\ldots m_K T} \otimes \B_k \otimes \I_{m_1m_2\ldots m_{k-1}} \Bigg) \x^{(k)}, \label{eqbitenKvec} 
\end{align}
where each $\B_k$ is an $m_k \times m_k$ matrix of coefficients and $\y$ and $\x^{(k)}$ are the vectorizations of $\Y$ and $\X^{(k)}$, respectively. As the array form the BLIN model in \eqref{eqbitenK} is of course linear, an appropriate design matrix as in \eqref{eq_blinvec} may be constructed to estimate the sparse array BLIN model.

We now provide array manipulations such that a block coordinate descent, similar to that in Algorithm~\ref{alg:BLIN_it}, may be derived to iteratively estimate $\{\B_k \}_{k=1}^K$ in the array BLIN model. First, let $\M$ be any array of dimensions $m_1 \times m_2 \times m_3$. The mode-1 matricization of the array $\M_{(1)}$ is defined as an $m_1 \times m_2 m_3$ matrix and if $\M = \X \times\{\C_1,\C_2, \C_3 \}$, then $\M_{(1)} = \C_1 \X_{(1)} (\C_3 \otimes \C_2)^T$. The following mode-1 and mode-2 matricizations of the expectation of the matrix BLIN model are then
\begin{align}
E[\Y_{(1)} \, | \, \X, \Z] &= \A^T \X_{(1)} + \Z_{(1)}\big(\I_T \otimes \B \big), \label{eqbitenMode1}\\
E[\Y_{(2)} \, | \, \X, \Z] &= \X_{(2)} \big(\I_T \otimes \A \big)  +  \B^T \Z_{(2)}. \label{eqbitenMode2}
\end{align}
The mode-$i$ matricization of the expectation of the $K$-mode BLIN model in \eqref{eqbitenK} may be written as
\begin{align}
E \left[\Y_{(i)} \, | \, \{ \X^{(k)} \}_{k=1}^K \right] &= \B_i^T \X_{(i)}^{(i)} + \sum_{k \neq i} \X_{(i)}^{(k)}  \Bigg(  \I_{n_1^k} \otimes \B_k \otimes \I_{n_2^k} \Bigg), \label{eqbitenKmat} \\
& \text{ where } n_1^k = T \prod_{\substack{j \ge k+1 \\ j \neq i}}m_j \text{  and  } n_2^k = \prod_{\substack{j \le k-1 \\ j \neq i}}m_j. \nonumber
\end{align}
Representation the model in this form enables straightforward estimation of  $\B_i$ given $\{ \B_{k}\}_{k\neq i}$.

\section{Details of Data Analysis}
\label{section_SID_details}

In this section, we provide supporting materials for performing the analysis of the ICEWS data set.
In the notation of Section~\ref{sec_multi}, the BLIN model for this analysis    
may be written:
\begin{align} \Y_t
 &=  \left( \sum_{k=1}^{p_A} \Y_{t-k} \right) \times \{\A^T, \I_S, \I_R \}  
 + \left( \sum_{k=1}^{p_B} \Y_{t-k}  \right) \times \{\I_S, \B^T, \I_R \}  \\ \nonumber
& \hspace{1in} + \left( \sum_{k=1}^{p_C} \Y_{t-k}  \right) \times \{\I_S, \I_S, \C^T \} 
 + \E_t, 
 \end{align}
where $\Y_t$ is a $25 \times 25 \times 4$ array representing senders, receivers, and interaction types. 
We also examined replacing $\Y_t$ with $\D_t := \Y_{t} - \Y_{t-1}$, the week-over-week increase in interaction value. To evaluate the performance of the model to represent the undifferenced and differenced data, we estimated the sparse BLIN model on each of differenced and undifferenced data for all combinations of lags $\{p_A, p_B, p_C \}$ between 1 and 5 (inclusive). 

We found that the BLIN model represented the differenced data significantly better, that is, with in-sample $R^2$ values above $0.30$ for differenced data rather than those below $0.10$ for undifferenced data. Thus, we proceeded with the differenced data only. To select the model that best balanced model fit and model parsimony, we used an analog of Akaike's Information Criterion \citep{akaike1998information}:
\begin{align}
    \hat{AIC} = 2\left(||\hat{\A} ||_1 + ||\hat{\B} ||_1 + ||\hat{\C} ||_1 \right) + S^2 R T \ \text{log} \sum_{t=1}^T || \hat{\D}_t - \D_t ||^2_2, \label{AICcrit}
\end{align}
where $||\H ||_1$ is the number of nonzero entries in $\H$, e.g., and $\hat{\D}_t$ is the estimated mean corresponding to the estimated networks $\{ \hat{\A}, \hat{\B}, \hat{\C} \}$. In \eqref{AICcrit}, the first term is a penalty for the number of parameters, i.e. model complexity, and the second term quantifies the goodness of fit. Smaller values of $\hat{AIC}$ are better; the five smallest values are given in Table~\ref{tab_AICs}.

\spacingset{1}
\begin{table}[ht]
\centering
\caption{The five smallest information criteria ($\hat{AIC}$) for selected lags of sparse BLIN estimates of differenced data. } 
\label{tab_AICs}
\begin{tabular}{cccccc} 
$p_A$ & $p_B$ & $p_C$ & $\hat{AIC}$ & $R^2$ & Number of parameters \\
\hline
5   & 3  &  1 & 58307 & 0.308 & 563 \\
3 &  5  & 1 & 58460 & 0.308 & 613 \\
5 & 1 & 3 & 58738 & 0.307 & 628 \\
1 & 5 & 3 & 58860 & 0.307 & 463 \\
3 & 1 & 5 & 59074 & 0.307 & 617
\end{tabular}
\end{table}
\spacingset{1.5}

To select the network entries to be highlighted in  Figure~\ref{fig:mp_ICEWS}, we plotted the ordered nonzero coefficients (Figure~\ref{fig:coef_ordered}). Then, we visually cut off the coefficient values at natural breaks in the nonzero entries. The visual cutoffs are near the 2.5\%  and 97.5\% quantiles of the nonzero coefficient values.

\spacingset{1}
\begin{figure}[ht]
\centering
  \includegraphics[width=.6\textwidth]{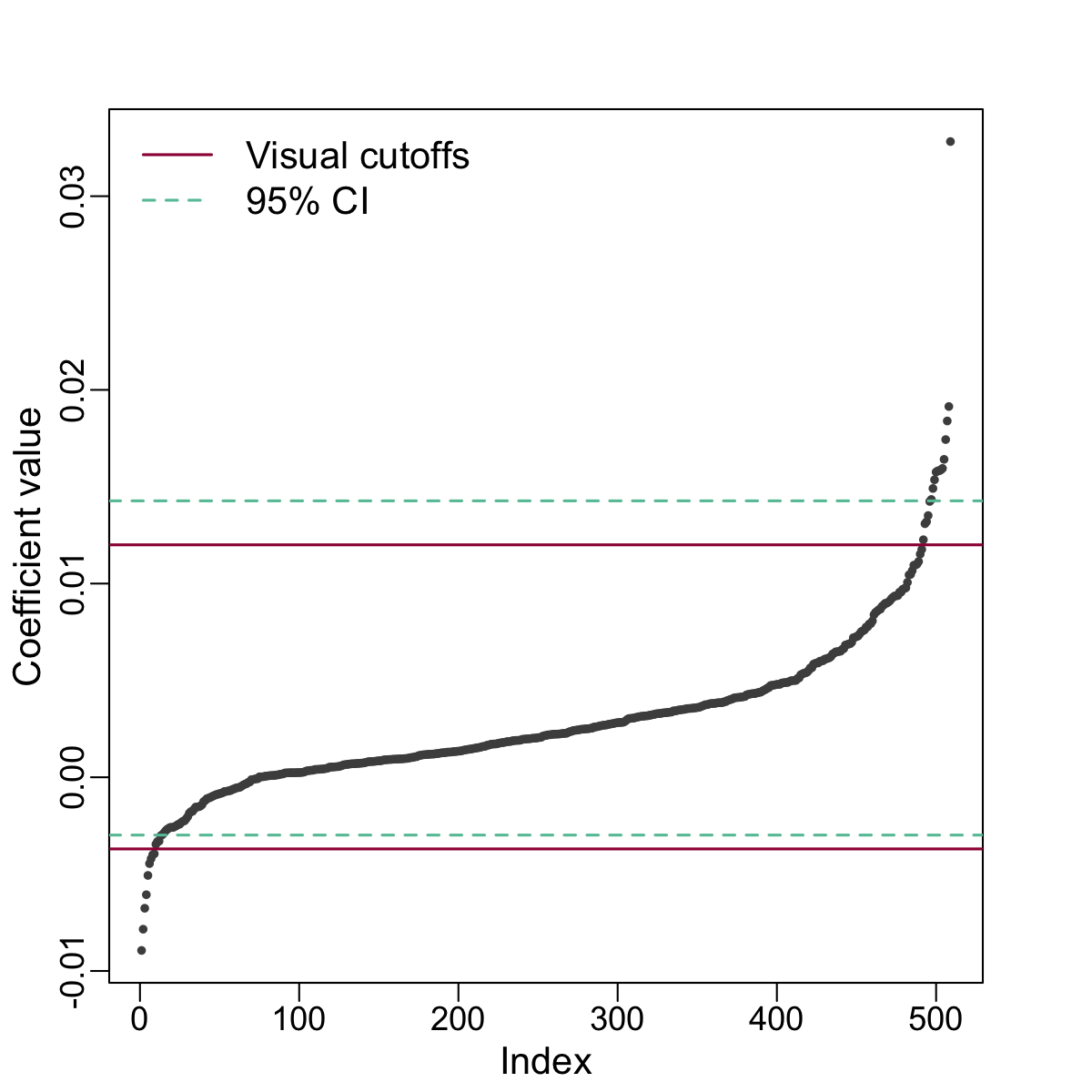}
 
\caption{Estimated influence network values and cutoffs for relations highlighted in Figure~\ref{fig:mp_ICEWS}. }
\label{fig:coef_ordered}
\end{figure} 
\spacingset{1.5}

\end{document}